\documentclass[12pt,a4paper]{article}
\usepackage[T1]{fontenc}
\usepackage{fancyhdr}
\usepackage{amsmath,amsthm,amsfonts,amssymb}
\allowdisplaybreaks
\usepackage{epic,eepic}
\usepackage{graphicx}
\usepackage{comment}
\usepackage{soul}
\usepackage{authblk}
\usepackage{fullpage}
\usepackage[active]{srcltx}
\usepackage{enumitem}
\usepackage{booktabs}

\usepackage{framed,color}

\usepackage{float}%
\floatstyle{plaintop}%
\restylefloat{table}

\usepackage[font={small,it}]{caption}
\usepackage{subcaption}

\usepackage[table]{xcolor}

\usepackage[toc,page]{appendix}
\graphicspath{{figures/}}

\usepackage[style=authoryear,backend=bibtex,maxcitenames=2,maxbibnames=99]{biblatex}
\bibliography{CLPM}

\usepackage{chngcntr}
\usepackage{dsfont}

\usepackage{algpseudocode}
\usepackage{algorithm}
\usepackage{setspace}

\usepackage{hyperref}

\makeatletter
\AtEveryBibitem{%
  \global\undef\bbx@lasthash%
  \clearfield{}}
\makeatother

\date{}
\title{\textbf{Continuous Latent Position Models for Instantaneous Interactions}}
\author[1]{\footnote{The authors contributed equally to this work.}
Riccardo Rastelli}
\author[2]{Marco Corneli}
\affil[1]{\footnotesize School of Mathematics and Statistics, University College Dublin, Dublin, Ireland\\
\footnotesize riccardo.rastelli@ucd.ie}
\affil[2]{\footnotesize Universit\'e C\^{o}te d'Azur, INRIA, Center of Modeling, Simulation and Interactions, MAASAI Team, Nice, France\\
\footnotesize marco.corneli@univ-cotedazur.fr}

\theoremstyle{plain}

\newtheorem{proposition}[]{Proposition}

\theoremstyle{definition}

\theoremstyle{remark}
\newtheorem*{remark}{Remark}

\numberwithin{equation}{section}
\numberwithin{figure}{section}

\newcommand{\bz}{\textbf{z}}
\newcommand{\bZ}{\textbf{Z}}
\newcommand{\blambda}{\boldsymbol{\lambda}}
\newcommand{\scalar}[2]{\left\langle #1 , #2 \right\rangle}
\newcommand{\norm}[1]{\parallel #1 \parallel}

\begin{document}
\rowcolors{2}{gray!25}{white}
\counterwithout{figure}{section}
\counterwithout{figure}{subsection}
\counterwithout{equation}{section}
\counterwithout{equation}{subsection}

\maketitle
\begin{abstract}
\baselineskip=20pt
\noindent
We create a framework to analyse the timing and frequency of instantaneous interactions between pairs of entities. This type of interaction data is especially common nowadays, and easily available. Examples of instantaneous interactions include email networks, phone call networks and some common types of technological and transportation networks. Our framework relies on a novel extension of the latent position network model: we assume that the entities are embedded in a latent Euclidean space, and that they move along individual trajectories which are continuous over time. These trajectories are used to characterize the timing and frequency of the pairwise interactions. We discuss an inferential framework where we estimate the individual trajectories from the observed interaction data, and propose applications on artificial and real data.
\\

\end{abstract}

\baselineskip=20pt
\section{Introduction}\label{sec:introduction}
The Latent Position Model (LPM, \cite{hoff2002latent}) is a widely used statistical model that can be used to characterize a network through a latent space representation.
The model embeds the nodes of the network as points in the real plane, and then uses these latent features to explain the observed interactions between the entities. 
This provides a neat and easy-to-interpret graphical representation of the observed interaction data, which is able to capture some extremely common empirical features such as transitivity and homophily.

In this paper, we propose a new LPM that can be used to model repeated instantaneous interactions between entities, over an arbitrary time interval. 
The time dimension is continuous, and an interaction between any two nodes may happen at any point in time.
We propose a data generative mechanism which is inspired by the extensive literature on LPMs, and we define an efficient estimation framework to fit our model.

Since the foundational work of \textcite{hoff2002latent}, the literature on LPMs has been developed in many directions, both from the methodological and from the applied point of views. 
Recent review papers on the topic include \textcite{salter2012review, rastelli2015properties, raftery2017comment, sosa2020review}.

As regards statistical methodology, the original paper of \textcite{hoff2002latent} defined a framework to infer and interpret a LPM for binary interactions. 
The authors introduced two types of LPMs: the \textbf{projection model} and the \textbf{distance model}. 

The \textbf{projection model} postulates that the probability that an edge appears between any two nodes is determined by the dot product of the latent coordinates of the two respective nodes. 
As a consequence, a crucial contribution for the edge probability is given by the direction the nodes point towards.
By contrast, the distance model defines the connection probability as a function of the Euclidean distance between the two nodes. Nodes that are located close to each other are more likely to connect than nodes that are located far apart.
Both models provide a clear representation of the interaction data which can be used to study the network's topology, or to construct model-based summaries and visualizations, or predictions.

More recently, the projection model and its variations have been extensively studied and used in a variety of applications (see \cite{hoff2005bilinear, hoff2018additive} and references therein). 
This model has also clear connections to a rich machine learning literature on spatial embeddings, which include \textcite{rahimi2007random, lee1999learning, halko2011finding}. 
Variations of the projection model have been extended to dynamic settings \parencite{durante2014nonparametric, durante2016locally}, and other types of networks frameworks \parencite{durante2017nonparametric}.

As regards the \textbf{distance model}, this has been extended by \textcite{handcock2007model} and \textcite{krivitsky2009representing} to represent clustering of the nodes and more flexible degree distributions.
In the context of networks evolving over time, dynamic extensions of the model have been considered in \textcite{sarkar2006dynamic}, and more recently in several works including \textcite{sewell2015latent} and \textcite{friel2016} for binary interactions.
The recent review paper of \textcite{kim2018review} provides additional references on dynamic network modeling.
Other relevant and interesting works that revolve around the distance models in either static or dynamic settings include \textcite{gollini2014joint, salter2017latent} for multi-view networks, \textcite{sewell2016latent} for dynamic weighted networks, and \textcite{gormley2007latent, sewell2015analysis} for networks of rankings.
We also mention \textcite{raftery2012fast, fosdick2019multiresolution, rastelli2018computationally, tafakori2019measuring} which introduce original and closely related modeling or computational ideas.

Crucially, we note that all existing dynamic LPMs consider a discrete time dimension, whereby the interactions are observed at a number of different points in time\footnote{One exception is \textcite{durante2016locally} but we find their approach and goals fundamentally different from ours.}. 
By contrast, a fundamental original aspect of this paper is that it considers a fully continuous LPM whereby the interactions are instantaneous, and they can happen at any point in time.
Continuous networks of this type are especially common and widely available, as they include email networks \parencite{klimt2004enron}, functional brain networks \parencite{park2013structural}, and other networks of human interactions (see \cite{cattuto2010dynamics, barrat2013temporal} and references therein).

Some of the approaches that have been proposed in the statistics literature to model instantaneous interactions include \textcite{corneli2017multiple} and \textcite{matias2018semiparametric}, however we note that these approaches rely on extensions of the stochastic blockmodel \parencite{nowicki2001estimation}, and not on the LPM.
Another relevant strand of literature focuses instead on modeling this type of data using Hawkes processes (see \cite{junuthula2019block} and references therein). \\

We propose our new Continuous Latent Position Model (CLPM) both for the projection model framework and for the distance model framework. 
In our approach, each of the nodes is characterized by a latent trajectory on the real plane, which is assumed to be a piece-wise linear curve. 
The interactions between any two nodes are modeled as events of a inhomogeneous Poisson point process, whose rate is determined by the instantaneous positions of the nodes, at each point in time.
The piece-wise linear curve assumption gives sufficient flexibility regarding the possible trajectories, while not affecting the purely continuous nature of the framework, in that the rate of the Poisson process is not piece-wise constant. 
This is a major difference with respect to other approaches that have been considered (\cite{corneli2017multiple} and one of the approaches of \cite{matias2018semiparametric}).

We propose a penalized likelihood approach to perform inference, and we use optimization via gradient descent to obtain optimal estimates of the model parameters. 
We have created a software that implements our estimation method, which is publicly available from \textcite{clpm_github}.

The paper is structured as follows: in Section \ref{sec:model} we introduce our new model and its two variants (i.e. the projection and distance model), and we derive the main equations that are used in the paper; in Section \ref{sec:inference} we describe our approach to estimate the model parameters; in Section \ref{sec:exp_simu} we illustrate our procedure on three synthetic datasets, whereas in Section \ref{sec:applications} we propose real data applications.
We give final comments and conclusions in Section \ref{sec:conclusions}.

\section{The model}\label{sec:model}

\subsection{Modeling the interaction times}\label{subsec:mod1}
The data that we observe is stored as a list of interactions (or \emph{edge list}) in the format $\mathcal{E} :=\{(\tau_e, i_e, j_e)\}_{e \in \mathbb{N}}$, where $\tau_e \in [0,T]$ for all $e$ is the interaction time between the nodes $i_e$ and $j_e$, with $i_e,j_e$  $\in \left\{1,\dots,N\right\}$.
We consider undirected interactions without self loops, although extensions to the directed case are straightforward.
We emphasize that all interactions are instantaneous, i.e. their length is not relevant or not recorded. An interaction between two nodes may occur at any point in time $\tau_e \in [0, T]$.
Let us now formally introduce the list of the interaction times between two arbitrary nodes $i$ and $j$:  
\begin{equation}\label{eq:int_ij}
    \mathcal{E}_{ij} = \{\tau_1^{(i,j)}, \dots, \tau_{E_{ij}}^{(i,j)}\},
\end{equation}
where $E_{ij}$ is the total number of times $i$ interacts with $j$ before $T$.
We assume that the interaction times in the above equation are the realization of a inhomogeneous Poisson point process with instantaneous rate function denoted with $\lambda_{ij}(t) \geq 0,\ \forall t \in [0,T]$ and nodes $i$ and $j$.
Using a more convenient (but equivalent) characterization, we state that the waiting time for a new interaction event between $i$ and $j$ is exponential with a variable rate that changes over time.
Then, if we assume that the inhomogeneous point processes are independent for all pairs $i$ and $j$, the likelihood function for the rates can be written as:
\begin{equation}\label{eq:likelihood_1}
    \mathcal{L}\left(\blambda\right) = \prod_{i,j:\ i < j} \left[\left(\prod_{e \in \mathcal{E}_{ij}} \lambda_{ij}(\tau_e)\right) \exp\left\{-\int_{0}^{T} \lambda_{ij}(t)dt\right\}\right],
\end{equation}
where, for simplicity, we have removed the superscript $(i,j)$ from $\tau_e^{(i,j)}$.
In the sections below we will specify the conditions that make the processes independent.

\subsection{Latent positions}\label{subsec:mod2}
Our goal is to embed the nodes of the network into a latent space, such that the latent positions are the primary driving factor behind the frequency and timing of the interactions between the nodes.
Crucially, since the time dimension is continuous and interactions can happen at any point in time, we aim at creating a modeling framework which also evolves continuously over time.
Thus, the fundamental assumption of our model is that, at any point in time, the Poisson rate function $\lambda_{ij}(t)$ is determined by the latent positions of the corresponding nodes, which we denote $\bz_i(t) \in \mathbb{R}^2$ and $\bz_j(t) \in \mathbb{R}^2$.

\begin{remark}
We assume that the number of dimensions of the latent space is equal to $2$, because the main interest of the proposed approach is in latent space visualization of the network. However,  we note that the generative model presented in this section can be easily extended to the case $\bz_i(t) \in  \mathbb{R}^d$, with $d>2$.
\end{remark}

To facilitate the inference task, the trajectories are assumed to be piece-wise linear curves, characterized by a number of user-defined
change points in the time dimension. 
These change points are in common across the trajectories of all nodes and they determine the points in time when the linear motions of the nodes change direction and speed.
This means that we must define a grid of the time dimension through $K$ change points
\[
0=\eta_1 < \eta_2 < \dots < \eta_{K-1} < \eta_K = T 
\]
that are common across all trajectories.
We stress that this modeling choice is only meant to restrict the variety of continuous trajectories that we may consider, as it allows us to use a tractable parametric structure while keeping a high flexibility regarding the trajectories that can be obtained. Also, we make the assumption that, within any two consecutive critical points, the speed at which any given node moves remains constant.
As a consequence, we only need to store the coordinates of the nodes at the change points, since all the intermediate positions can then be obtained with:
\begin{equation}\label{eq:trajectories_1}
    \bz_i\left((1-t)\eta_k + t\eta_{k+1}\right) = (1-t)\bz_i\left(\eta_k\right) + t\bz_i\left(\eta_{k+1}\right) \quad \quad \forall t \in [0,1]
\end{equation}
for any change points $\eta_k$ and $\eta_{k+1}$ and node $i$.

\subsubsection{Projection Model}\label{subsubsec:PM}
Similarly to the foundational paper of \textcite{hoff2002latent}, we introduce two possible characterizations of the rates through the latent positions: one is inspired by the projection model, the other is inspired by the distance model.
In our projection model, we assume that:
\begin{itemize}
    \item[i)] the latent positions are constrained within the first quadrant of $\mathbb{R}^2$, i.e. $\bz_i(t) \in \mathbb{R}^+ \times \mathbb{R}^+$, for all $t \in [0,T]$ and node $i$;
    \item[ii)] the rate is exactly equal to the dot product between the positions, i.e. \\
    $\lambda_{ij}(t)~=~\scalar{\bz_i(t)}{\bz_j(t)}$, for all $t \in [0,T]$ and nodes $i$ and $j$.
\end{itemize}
As a consequence, the further the nodes are positioned from the origin, the more frequent their interactions will be, especially towards other nodes that are aligned in the same direction\footnote{Indeed, we could say equivalently
\begin{equation*}
\lambda_{ij}(t) = \cos\left(\alpha_{ij}\right) \norm{\bz_i(t)} \norm{\bz_j(t)}
\end{equation*}
where $\alpha_{ij}\in [0,\frac{\pi}{2}]$ is the angle between $\bz_i(t)$ and $\bz_j(t)$ and $\norm{\cdot}$ is the Euclidean norm.
}.
Viceversa, we are not expecting frequent interactions for nodes that are located too close to the origin, or between pairs of nodes forming an angle which is close to $90$ degrees.
The restriction of the latent space to the first quadrant guarantees that the rate remains always non-negative, and we argue that this assumption does not diminish the flexibility nor interpretability of the model.


By taking the logarithm of Eq.~\eqref{eq:likelihood_1} and replacing $\lambda_{ij}(t)$, the log-likelihood for the projection model is
\begin{equation}\label{eq:log_lik_proj}
    \log \mathcal{L}(\bZ) = \sum_{i,j: i<j}\left\{\left(\sum_{e \in \mathcal{E}_{ij}}\log\left(\scalar{\bz_i(\tau_e)}{\bz_j(\tau_e)}\right)\right) - \int_{0}^T \scalar{\bz_i(t)}{\bz_j(t)}dt \right\}
\end{equation}

\begin{remark}
When expressing the Poisson rate $\lambda_{ij}(\cdot)$ as a function of the latent trajectories, we move from the unconditional independence assumption leading to Eq.~\eqref{eq:likelihood_1} to a conditional one. The timelines of events for all pairs of nodes are independent \emph{given} the latent trajectories.
\end{remark}

As proven in Appendix~\ref{sec:proof_projection}, the integral term appearing in Eq.~\eqref{eq:log_lik_proj} can be calculated analytically, thus leading to a closed form expression for the log-likelihood of the projection model.

\subsubsection{Distance Model}\label{subsubsec:DM}
Here, we introduce a version of the latent position model that uses the latent Euclidean distances between the nodes, rather than the dot products.
We argue that, in this context, the distance model provides both more flexibility and easier interpretability.
The simulation studies that we perform in Section~\ref{sec:exp_simu} will highlight these advantages.

In the distance model, we assume that:
\begin{equation}\label{eq:distance_1}
    \log \lambda_{ij}(t) = \beta - \|\bz_i(t) - \bz_j(t)\|^2
\end{equation}
where the last term corresponds to the squared Euclidean distance between nodes $i$ and $j$ at time $t$. We also introduced an intercept term $\beta$ as per the original LPM by \textcite{hoff2002latent}. The intercept term affects all nodes but extensions of the model where it become specific to each node can also be considered for both the projection and the distance models.
By taking the logarithm of Eq.~\eqref{eq:likelihood_1} and using Eq.~\eqref{eq:distance_1} the log-likelihood of the distance model becomes: \begin{equation}
  \log \mathcal{L}(\beta,\bZ) = \sum_{i,j} \left\{ \left(\sum_{e \in \mathcal{E}_{ij}} (\beta - \parallel \bz_{i}(\tau_e)  - \bz_{j}(\tau_e) \parallel^2) \right) - \int_{0}^T e^ {\beta - \parallel \bz_{i}(s)  - \bz_{j}(s) \parallel^2} ds \right\}
  \label{eq:log_l}
\end{equation}
Similarly to the projection model, also the above log-likelihood has a closed form, since the integral inside the brackets can be calculated analytically (proof in Appendix~\ref{sec:closed_form}).

\subsection{Penalized likelihood}\label{subsec:Penalty}
Due to the piece-wise linearity assumption in Eq.~\eqref{eq:distance_1}, for each node we only need to estimate its positions at times $\{\eta_k\}_{k \in [K]}$.
In order to avoid over fitting and obtain more interpretable and meaningful results, we use prior distributions over the latent positions at $\{\eta_k\}_{k \in [K]}$, as a means to penalize large velocities of the nodes in the latent space.

\paragraph{Projection Model:} as a penalization, in this case we require that:
\begin{equation}\label{eq:prior_proj_1}
    \scalar{\frac{\bz_i(\eta_{k+1})}{\|\bz_i(\eta_{k+1})\|}}{\frac{\bz_i(\eta_{k})}{\|\bz_i(\eta_{k})\|}} \stackrel{\perp}{\sim} \mathcal{TN}\left(\mu, (\eta_{k+1}-\eta_{k})\sigma^2,0,1 \right)\quad \quad \forall k=1,\dots,K-1
\end{equation}
for all nodes, where $\mu$ and $\sigma^2$ are hyper-parameters to be fixed by the user.
The above assumption states that the cosine of the angle between the position of a node at a change point, and the position of the same node at the following change point, follows a truncated Gaussian distribution (in  $[0,1]$).
 In the applications we choose $\mu \approx 1$ and a small value of $\sigma^2$ so that we require that any node rotates around the origin as little as necessary.

\paragraph{Distance Model.} We define Gaussian random walk priors on the critical points of the latent trajectories:
\begin{equation}\label{eq:prior_1}
\begin{split}
    &\bz_i(0) \stackrel{iid}{\sim}\mathcal{N}(0,\sigma_0^2 I_2) \\
    &\bz_i(\eta_{k+1})\ |\ \bz_i(\eta_k)\stackrel{\perp}{\sim} \mathcal{N}\left(\bz_i(\eta_k), (\eta_{k+1}-\eta_{k})\sigma^2I_2 \right)\quad \quad \forall k=1,\dots,K-1
\end{split}    
\end{equation}
for every node $i$ where  $I_2$ is the identity matrix of order two.
The equation above (with $\sigma^2=1$) would correspond to a Brownian motion process for the $i$-th latent trajectory, except that we would only observe it at the change points, where the latent positions are estimated.
However, as the number of change points increases, the prior that we specify tends to a scaled Brownian motion on the plane.
The parameters $\sigma_0^2$ and $\sigma^2$ are user-defined, hence, similarly to the projection model, the Gaussian priors can be used as penalizations.
In order to obtain sensible penalizations, we choose small values of the variance parameters, as to ensure that the nodes are scattered around the origin of the space, and that the speed of the nodes along the trajectories is not too large.
In this way, the nodes are forced to move as little as necessary, making the latent visualization of the network easier to read and interpret.

\begin{remark}
The likelihood function of the original latent distance model of \textcite{hoff2002latent} is not identifiable with respect to translations, rotations, and reflections of the latent positions. 
This is a challenging issue in a Bayesian setting that relies on sampling from the posterior distribution.
In fact, the posterior samples become non-interpretable, since the affine transformations may have occurred during the collection of the sample \parencite{shortreed2006positional}.
These non-identifiabilities are not especially relevant in an optimization setting, since usually the equivalent configurations of model parameters lead to the same qualitative results and interpretations.
However, a case for non-identifiability can be made for dynamic networks, since translations, rotations, and reflections can occur across time, thus affecting results and interpretation.
The penalizations that we introduce in this paper ensure that the nodes move as little as necessary, thus disfavouring any rotations, translations and reflections of the space. As a consequence, the penalizations directly address the identifiability issues and the latent point process remains comparable across time.
\end{remark}

\section{Inference}\label{sec:inference}
In this section, we discuss the inference for the distance model described in Section~\ref{subsubsec:DM}, but an analogous procedure is considered for the projection model.

Recalling that we work with undirected graphs, and thanks to Eq.~\eqref{eq:prior_1}, the log-likelihood is
\begin{equation}\label{eq:loss_sgd}
\begin{split}
\log \mathcal{L}(\beta, \bZ) 
&= \sum_{i=1}^{n} \left\{ \frac{1}{2}\sum_{\substack{j=1 \\ j\neq i}}^n \left[\left(\sum_{e \in \mathcal{E}_{ij}} (\beta - \parallel \bz_{i}(\tau_e)  - \bz_{j}(\tau_e) \parallel^2) \right) - \int_{0}^T e^ {\beta - \parallel \bz_{i}(s)  - \bz_{j}(s) \parallel^2} ds\right]  \right.\\
&\hspace{1cm}\left. - \frac{1}{2\sigma^2} \sum_{k=1}^K \norm{\bz_i(\eta_{k}) - \bz_i(\eta_{k-1})}^2 - \frac{1}{2\sigma_0^2} \norm{\bz_i(\eta_{0})}^2 \right\} + C, \\
\end{split}
\end{equation}
where $C$ is a constant term that does not depend on  ($\bZ$, $\beta$) and the integral can be explicitly computed as shown in Appendix~\ref{sec:closed_form}.
Since this log-likelihood has a closed form, we implement it and rely on automatic differentiation~\parencite{griewank1989automatic,baydin2018automatic} to maximize it numerically, with respect to $(\beta, \bZ)$, via gradient ascent. 

We have implemented the estimation algorithm and visualization tools in a software repository, called \texttt{CLPM}, which is publicly available \parencite{clpm_github}.

Moreover, as it can be seen in Eq.~\eqref{eq:loss_sgd}, the log-likelihood is additive in the number of nodes. Potentially, this remark allows one to speed up the inference of the model parameters by means of stochastic gradient ascent~\parencite{bottou2010large}. Indeed, let us introduce $\psi_1, \dots, \psi_n$ such that
\begin{equation*}
\begin{split}
\psi_i(\beta, \bZ) :&=  \frac{1}{2}\sum_{\substack{j  = 1 \\ j \neq i}}^n \left[\left(\sum_{e \in \mathcal{E}_{ij}}(\beta - \parallel \bz_{i}(\tau_e)  - \bz_{j}(\tau_e) \parallel^2) \right) - \int_{0}^T e^ {\beta - \parallel \bz_{i}(s)  - \bz_{j}(s) \parallel^2} ds\right] \\
&- \frac{1}{2\sigma^2} \sum_{k=1}^K \norm{\bz_i(\eta_{k}) - \bz_i(\eta_{k-1})}^2 - \frac{1}{2\sigma_0^2} \parallel \bz_i(\eta_0)\parallel^2
\end{split}
\end{equation*}
and a discrete random variable $\Psi(\beta, \bZ)$ such that
\[
\pi := \pi_i := \mathbb{P}\{\Psi (\beta, \bZ)=\psi_i(\beta, \bZ)| \bZ\} = \frac{1}{n}, \qquad\forall  i \in \{1,\dots, n\}
\]
where we stress that the above probability  is conditional to $\bZ$ and given the model parameter $\beta$. Then, let us denote $\nabla$ the gradient operator with respect to $(\beta, \bZ)$  and $\mathbb{E}_{\pi}$ the expectation taken with respect tho the probability measure $\pi$ introduced above (and hence with $\bZ$ given). Then, we have the following 
\begin{proposition}
$n\nabla \Psi{(\beta,\bZ)}$ is an unbiased estimator of $\nabla \log \mathcal{L}(\beta, \bZ)$.
\end{proposition}
\begin{proof}
\[
\mathbb{E}_{\pi}\left[ n\nabla \Psi(\beta,\bZ) \right] = \sum_{i=1}^n \nabla \psi_i{(\beta,\bZ)} = \nabla {\sum_{i=1}^n \psi_i}(\beta,\bZ) = \nabla {\log \mathcal{L}(\beta,\bZ)},
\]
where the last equality follows from the additivity of the gradient operator.
\end{proof}
The above proposition allows us to sample (subsets of) nodes uniformly at random, with re-injection, and use each sample (a.k.a. mini-batch) to update the model parameters via stochastic gradient ascent, as shown in \textcite{bottou2010large}.

\section{Experiments: synthetic data}\label{sec:exp_simu}
In this section we illustrate applications of our methodology on artificial data. 
We propose two types of frameworks: in the first one, we consider dynamic block structures (which involve the presence of communities, hubs,  and isolated points). In this case, our aim is to inspect how the network dynamics are captured by \texttt{CLPM}. In the second framework, we generate data using the distance \texttt{CLPM} and we aim at recovering the simulated trajectories for each node.

\subsection{Dynamic block structures}\label{subsec:simu_1}
\paragraph{Simulation study 1.} In this first experiment, we use a data generative mechanism that relies on a dynamic blockmodel structure for instantaneous interactions \parencite{corneli2017multiple}.
We specifically focus on a special case of a dynamic stochastic blockmodel where we can have community structure, but we cannot have disassortative mixing, i.e. the rate of interactions within a community cannot be smaller than the rate of interactions between communities.
In this framework, the dynamic stochastic blockmodel approximately corresponds to a special case of our distance \texttt{CLPM}, whereby the nodes clustered together essentially are located nearby.

In the generative framework that we consider the only node-specific information is the cluster label, hence, this structure is not as flexible as the \texttt{CLPM} as regards modeling node's individual behaviours.
So, our goal here is to obtain a latent space visualization for these data, and to ensure that \texttt{CLPM} can accurately capture and highlight the presence of communities. 
An aspect of particular importance is how \texttt{CLPM} reacts to the creation and dissolution of communities over time: for this purpose, our generated data includes changes in the community structure over time.

For this setup, we consider the time interval $[0,40]$ (for simplicity we use seconds as a unit measure of time), and divide this into $4$ consecutive time segments of $10$ seconds each.
In each of the $4$ time segments, $60$ nodes are arranged into different community structures.
Thus, any changes in community structure are synchronous for all nodes and they happen at the endpoints of a time segment.
The rate of interactions between any two nodes is determined by their group allocations in that specific time segment.
The rate remains constant in each time segment, so that we effectively have a piece-wise homogeneous Poisson process over time, for each dyad.

We denote with $X^{(s)} \in \mathbb{N}^{N\times N}$ a simulated weighted interaction matrix which counts how many interactions occur in the $s$-th time segment for each dyad:
\[
X^{(s)}_{ij} | \mathbf{C} \sim \mathcal{P}\left(\theta_{\mathbf{c}_i \mathbf{c}_j}^{(s)} \right), 
\]
where $\mathcal{P}(\cdot)$ indicates the Poisson probability mass function, and $\mathbf{C}$ is a latent vector of length $N$ indicating the cluster labels of each of the nodes. 
Once we know the number of interactions for each dyad and each segment, the timing of these interactions can be sampled from a uniform distribution in the respective time segment.
More in detail, the rate parameters are characterized as follows:
\begin{itemize}
    \item[i)] in the time segment $[0,10[$, the expected number of interactions is the same for every pair of nodes: $\theta^{(1)}_{\mathbf{c}_1 \mathbf{c}_j} = 1$, for all $i$ and $j$;
    \item[ii)] in the time segment $[10,20[$, three communities emerge, in particular $\theta^{(2)}_{11} = 10$, $\theta^{(2)}_{22} = 5$ and $\theta^{(2)}_{33} = 1$, whereas the rate for any two nodes in different communities is $1$; 
    \item[iii)] in the time segment $[20,30[$, the first community splits and each half joins a different existing community. The two remaining communities are characterized by $\theta^{(3)}_{11} = \theta^{(3)}_{22} = 5$. Again, any two nodes in different communities interact with rate $1$;
    \item[iv)] in the time segment $[30,40]$, we are back to the same structure as in i).
\end{itemize}
Throughout the simulation, node $1$ always behaves as a hub, and node $60$ is always isolated.
This means that node $1$ interacts with rate $10$ at all times with any other node, whereas node $60$ interacts with rate $0.01$ at all times with any other node, regardless of any cluster label.

In Figure \ref{fig:sim_a_projection_results} we show a collection of snapshots at some critical time points, for the projection model.
\begin{figure}
     \centering
     \begin{subfigure}[b]{0.495\textwidth}
         \centering
         \includegraphics[width=\textwidth]{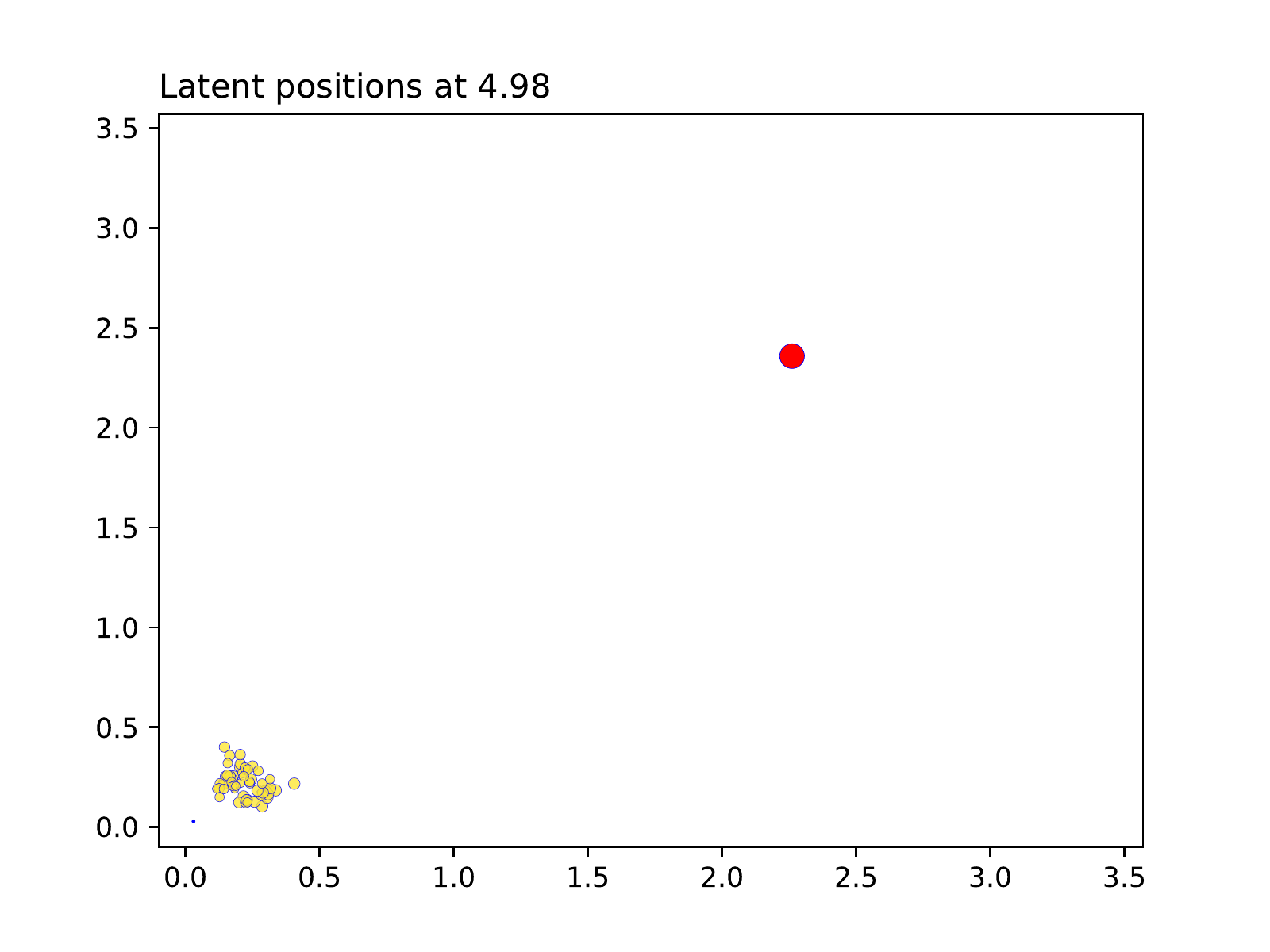}
     \end{subfigure}
     \hfill
     \begin{subfigure}[b]{0.495\textwidth}
         \centering
         \includegraphics[width=\textwidth]{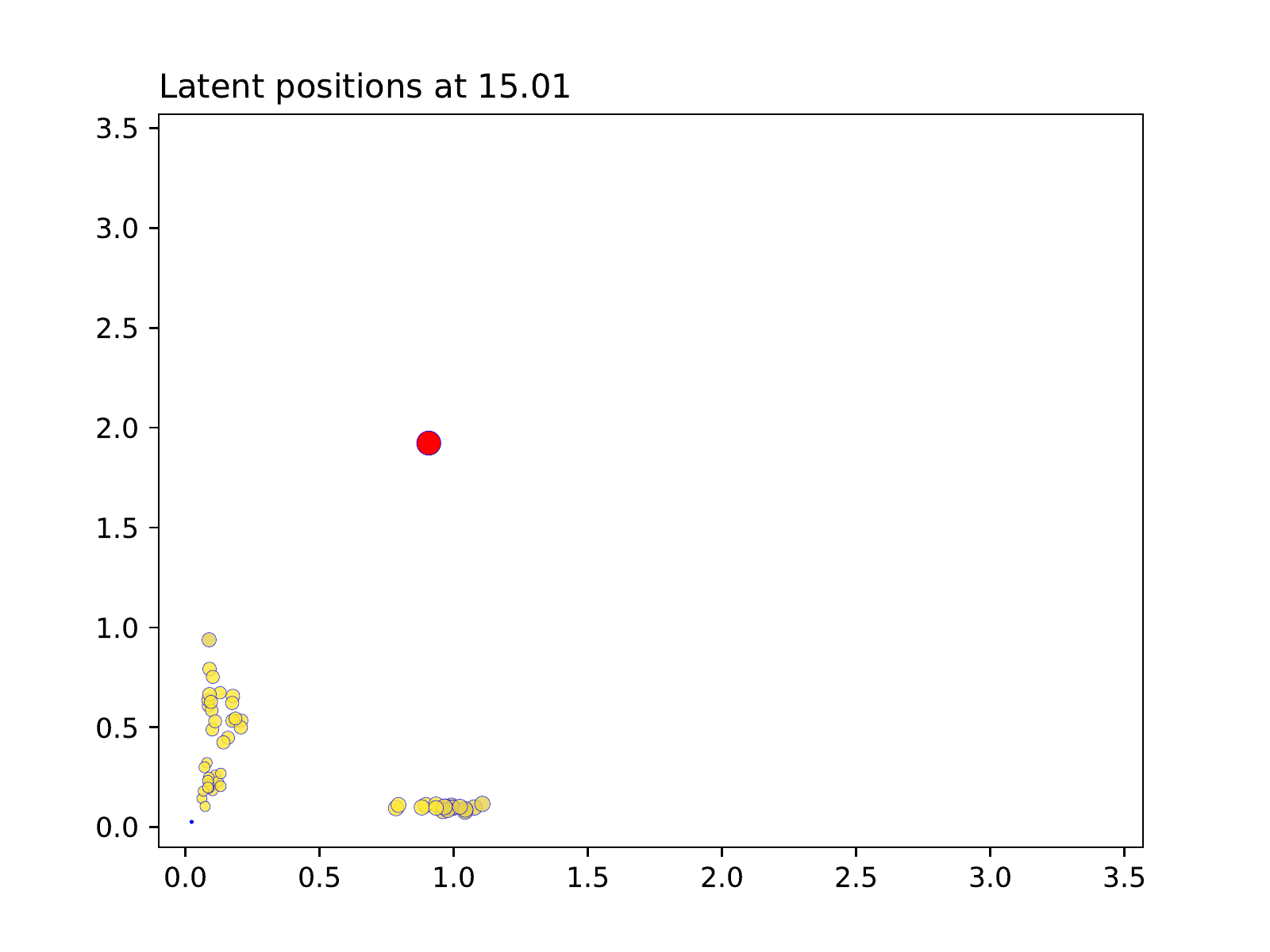}
     \end{subfigure}
     \centering
     \begin{subfigure}[b]{0.495\textwidth}
         \centering
         \includegraphics[width=\textwidth]{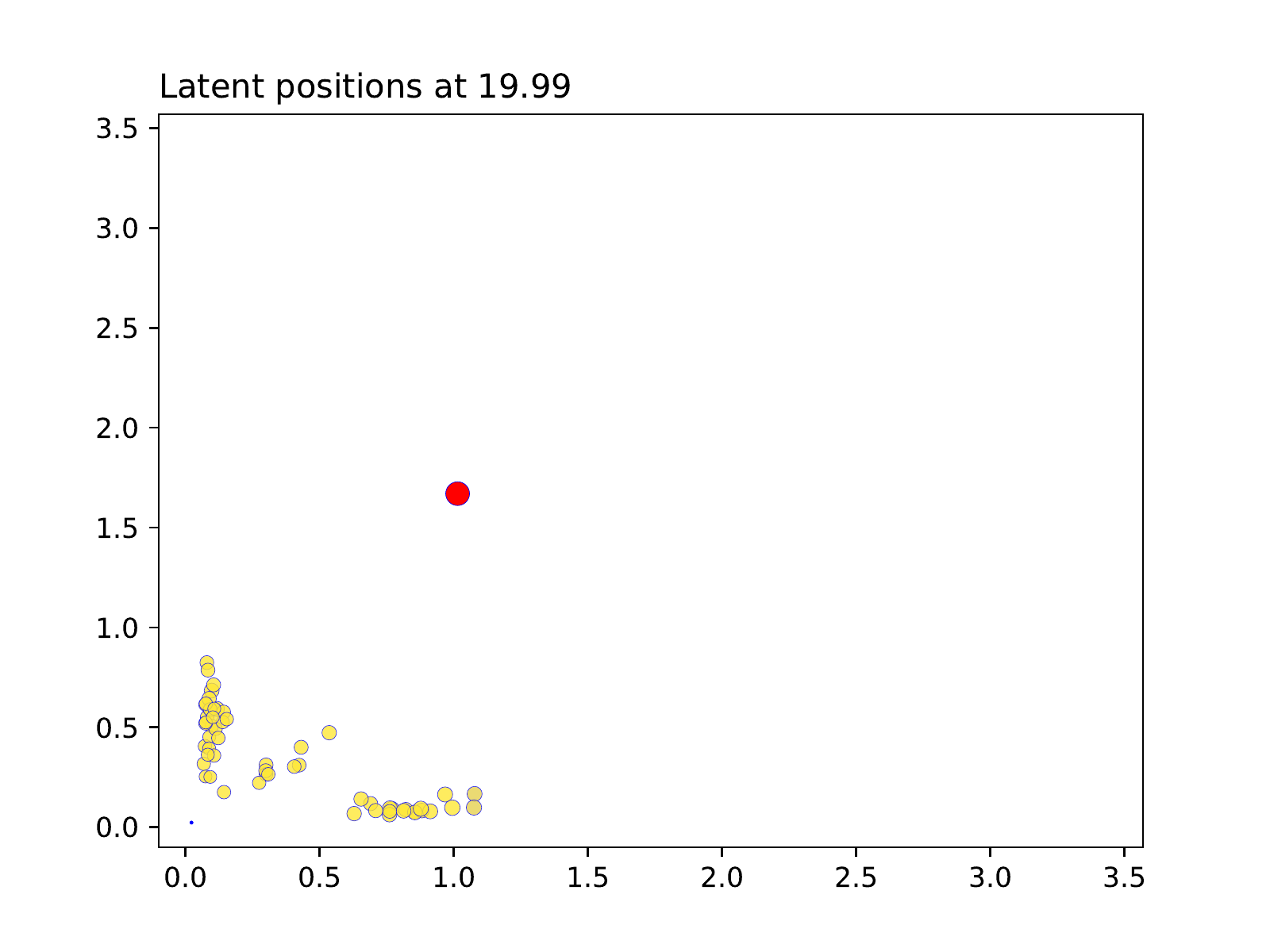}
     \end{subfigure}
     \hfill
     \begin{subfigure}[b]{0.495\textwidth}
         \centering
         \includegraphics[width=\textwidth]{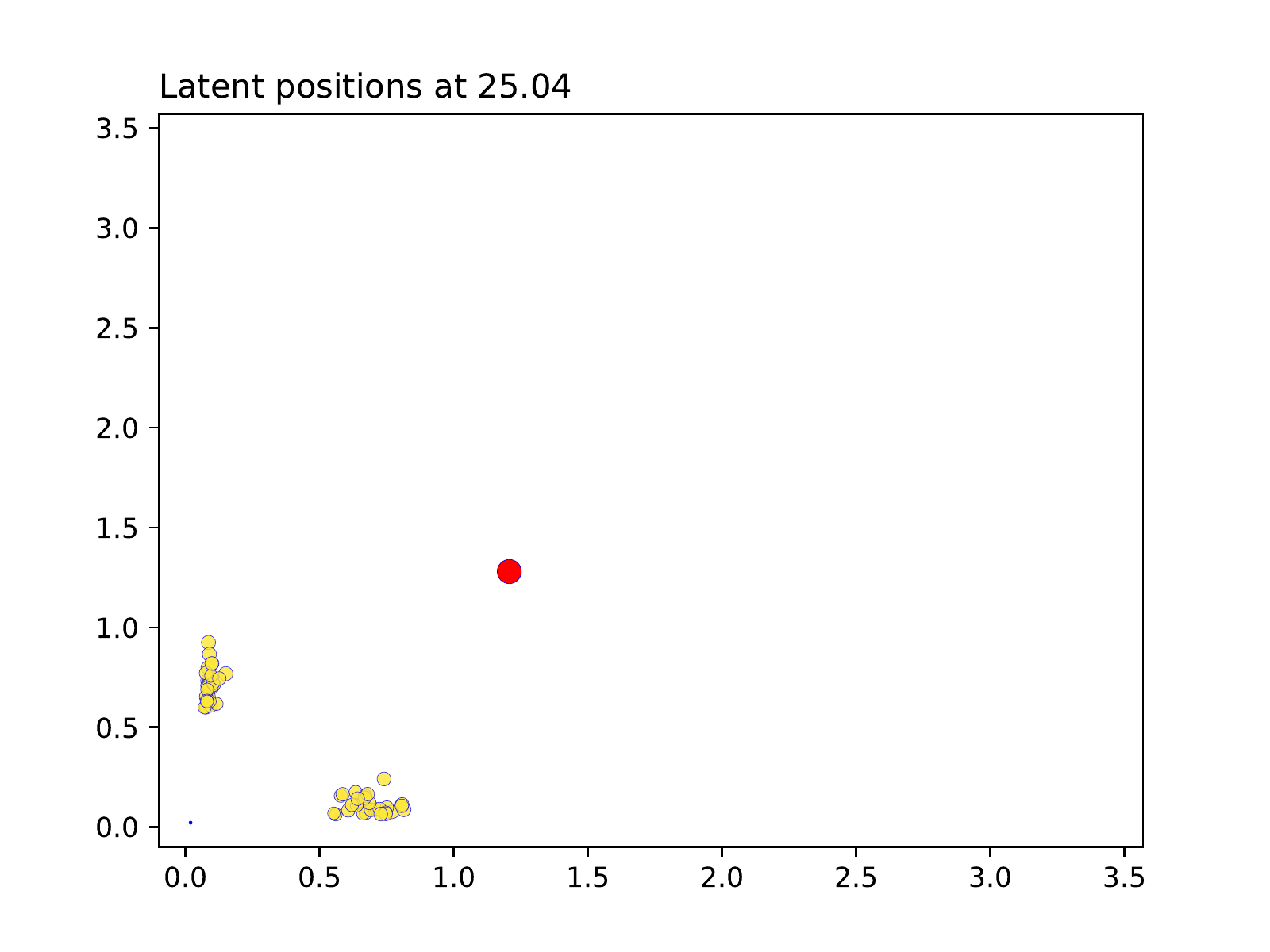}
     \end{subfigure}
     \begin{subfigure}[b]{0.495\textwidth}
         \centering
         \includegraphics[width=\textwidth]{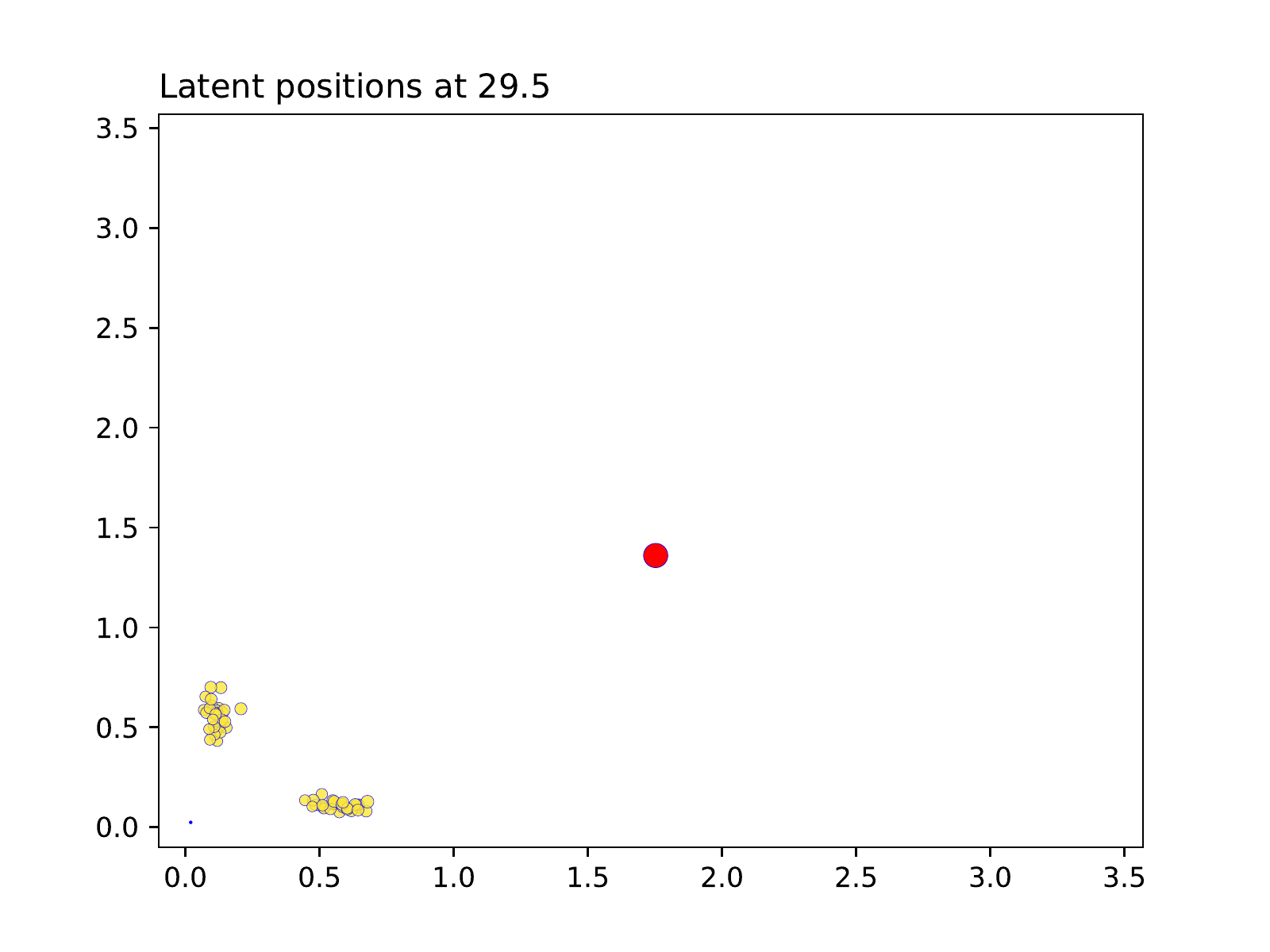}
     \end{subfigure}
     \hfill
     \begin{subfigure}[b]{0.495\textwidth}
         \centering
         \includegraphics[width=\textwidth]{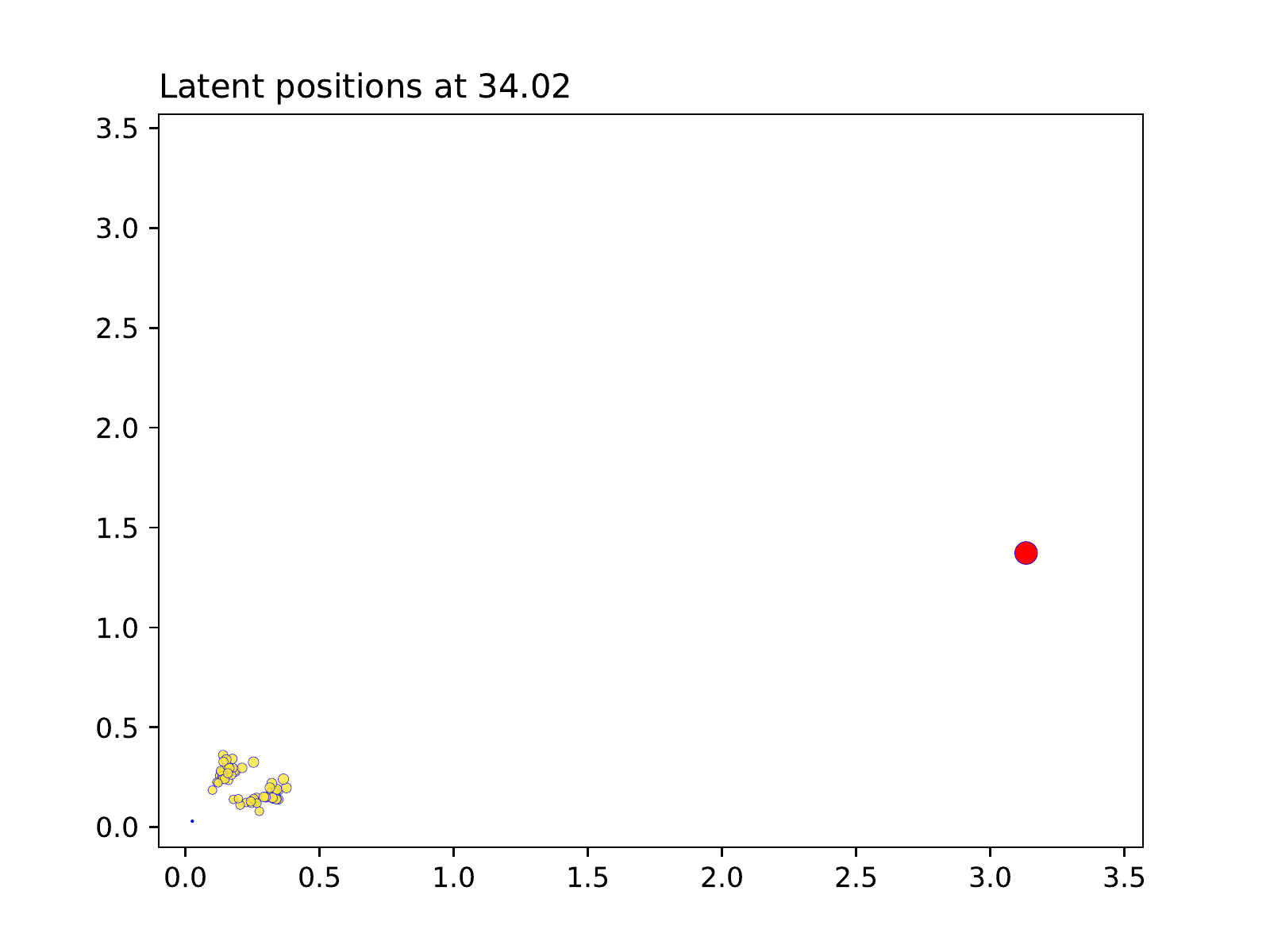}
     \end{subfigure}
        \caption{\textbf{Simulation study 1}: snapshots for the projection model.}
        \label{fig:sim_a_projection_results}
\end{figure}
The full videos of the results are provided in the code repository.
The main observation is that the communities are clearly captured at all times, and they are clearly visually separated.
While the two clusters with a stronger community structure are almost aligned to the axes, the non-community in the second time segment, which has low interaction rate, is instead positioned more centrally and it is more dispersed, but still separated from the others.
The hub is always located very far from the origin and from other points, since this guarantees a large dot product value with respect to all other nodes.
By contrast, the isolated node is always located at the origin of the space.

 Figure \ref{fig:sim_a_distance_results} shows instead the snapshots for the distance model.
\begin{figure}
     \centering
     \begin{subfigure}[b]{0.495\textwidth}
         \centering
         \includegraphics[width=\textwidth]{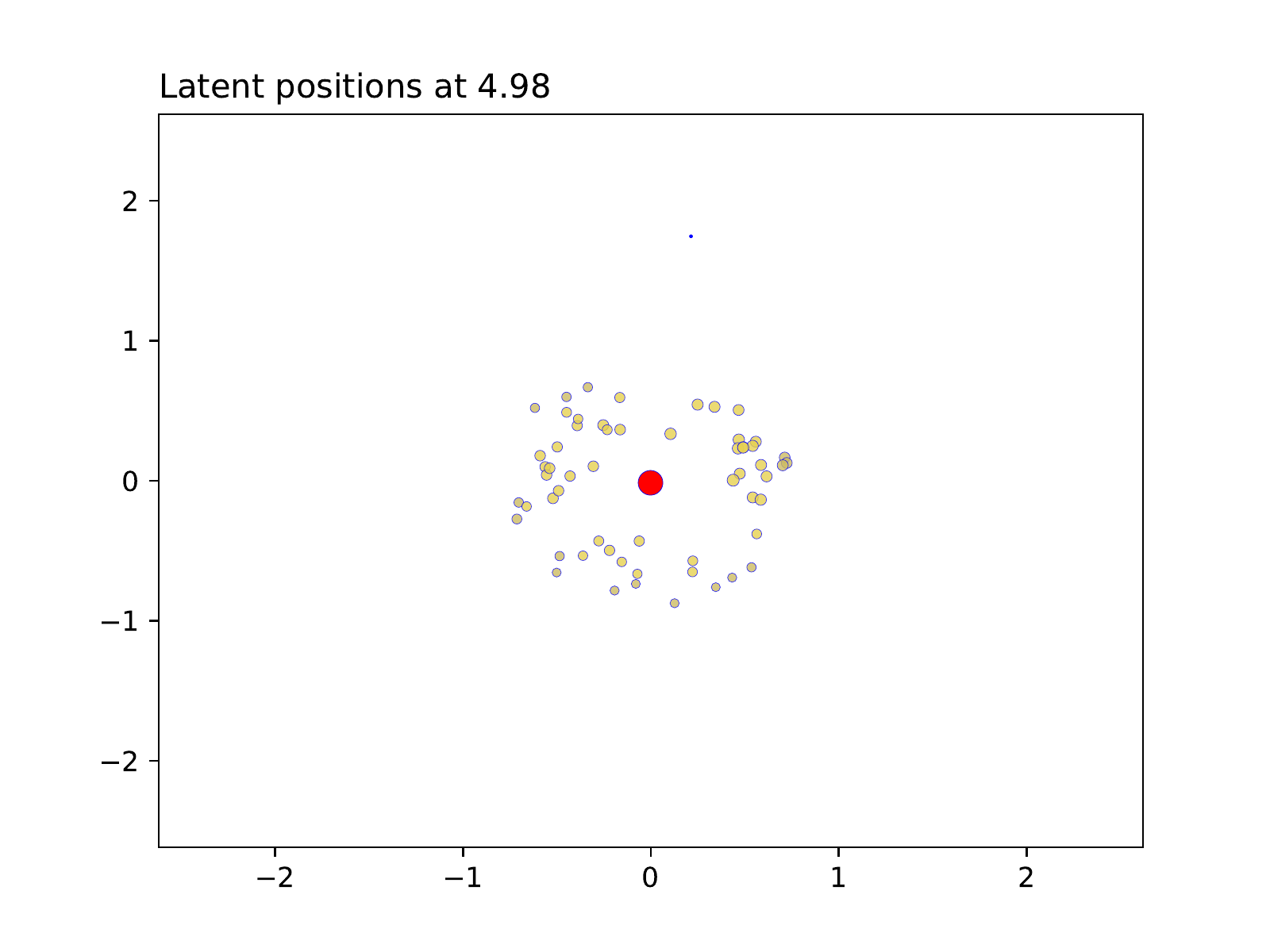}
     \end{subfigure}
     \hfill
     \begin{subfigure}[b]{0.495\textwidth}
         \centering
         \includegraphics[width=\textwidth]{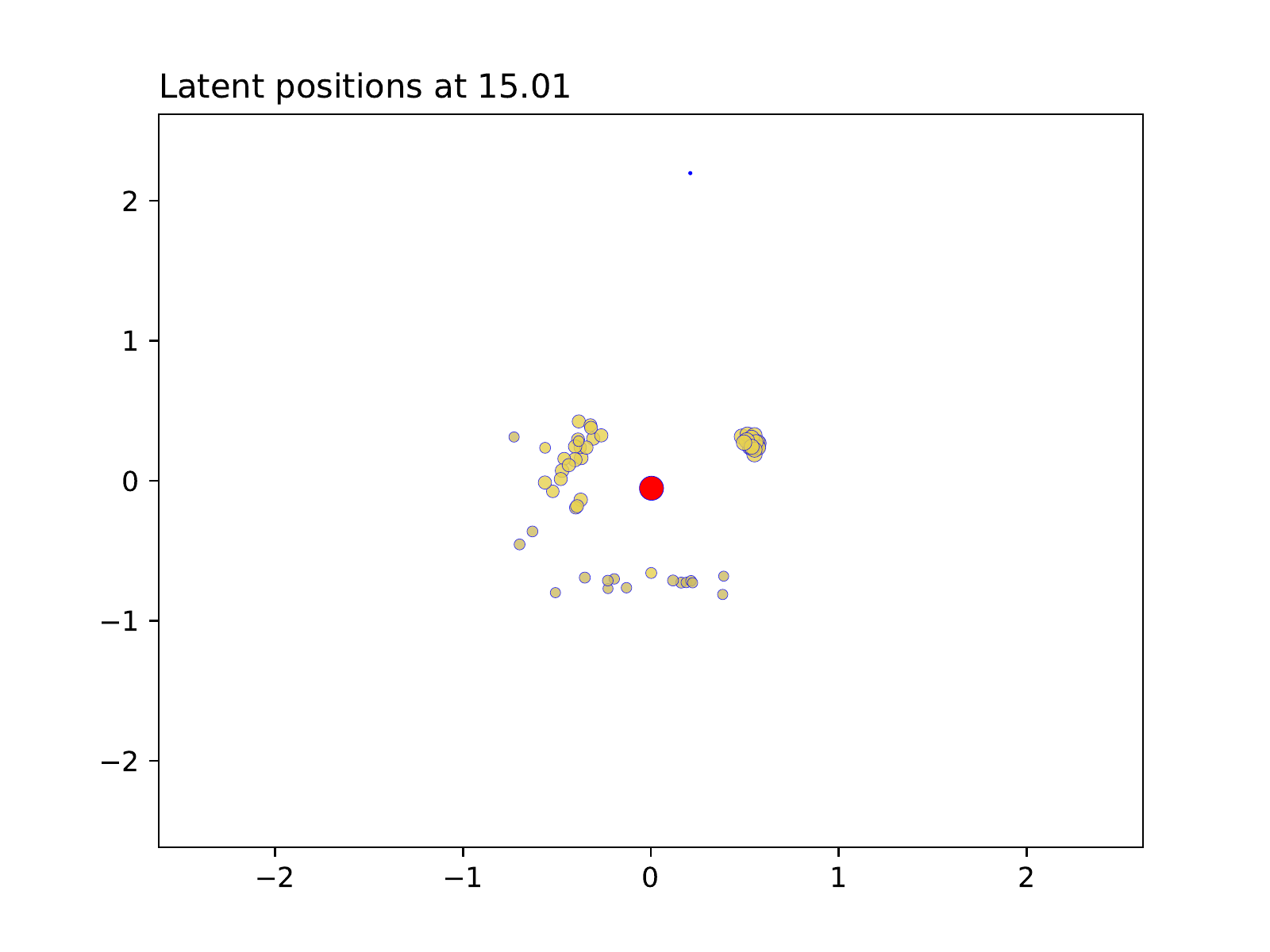}
     \end{subfigure}
     \centering
     \begin{subfigure}[b]{0.495\textwidth}
         \centering
         \includegraphics[width=\textwidth]{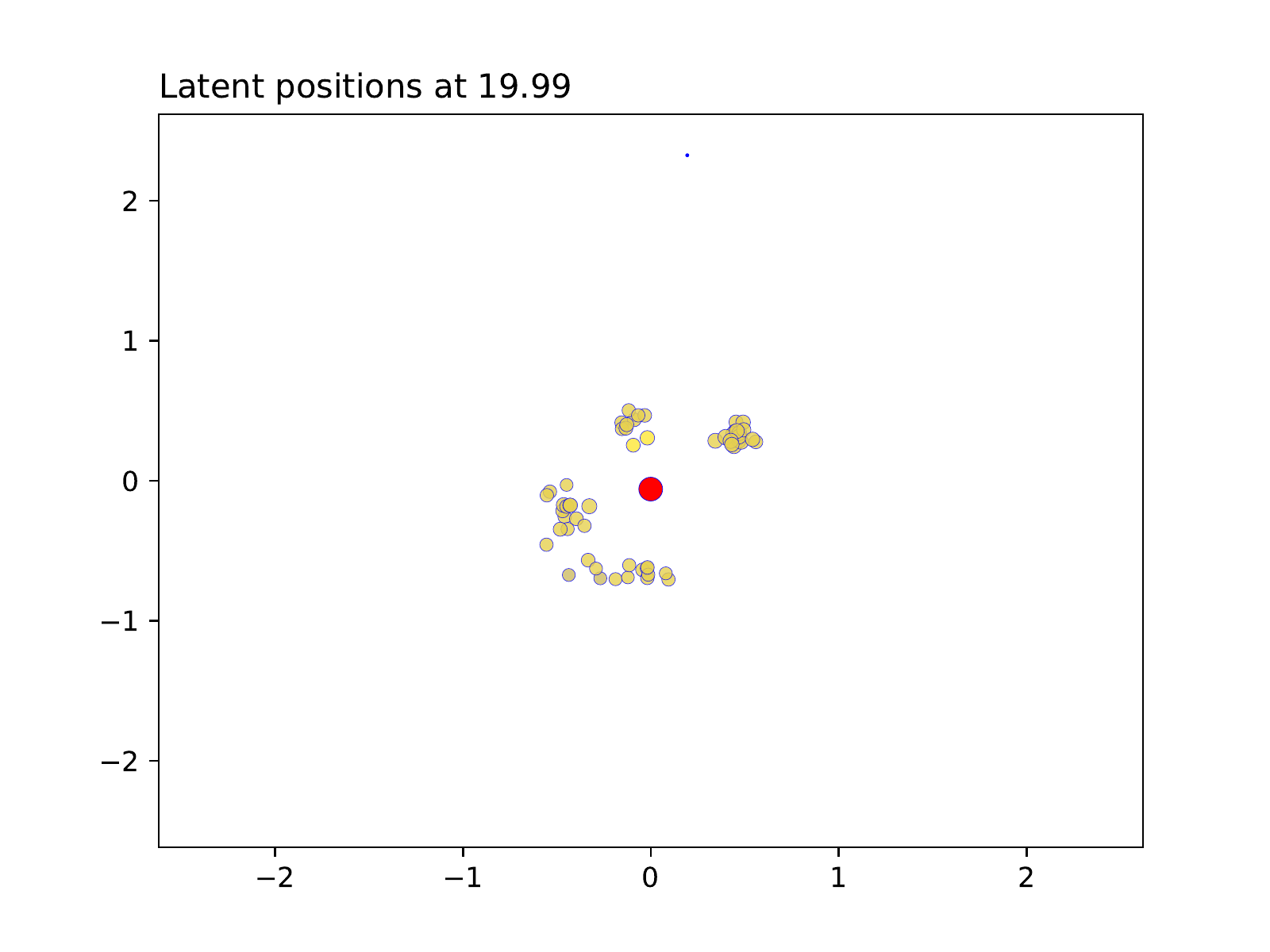}
     \end{subfigure}
     \hfill
     \begin{subfigure}[b]{0.495\textwidth}
         \centering
         \includegraphics[width=\textwidth]{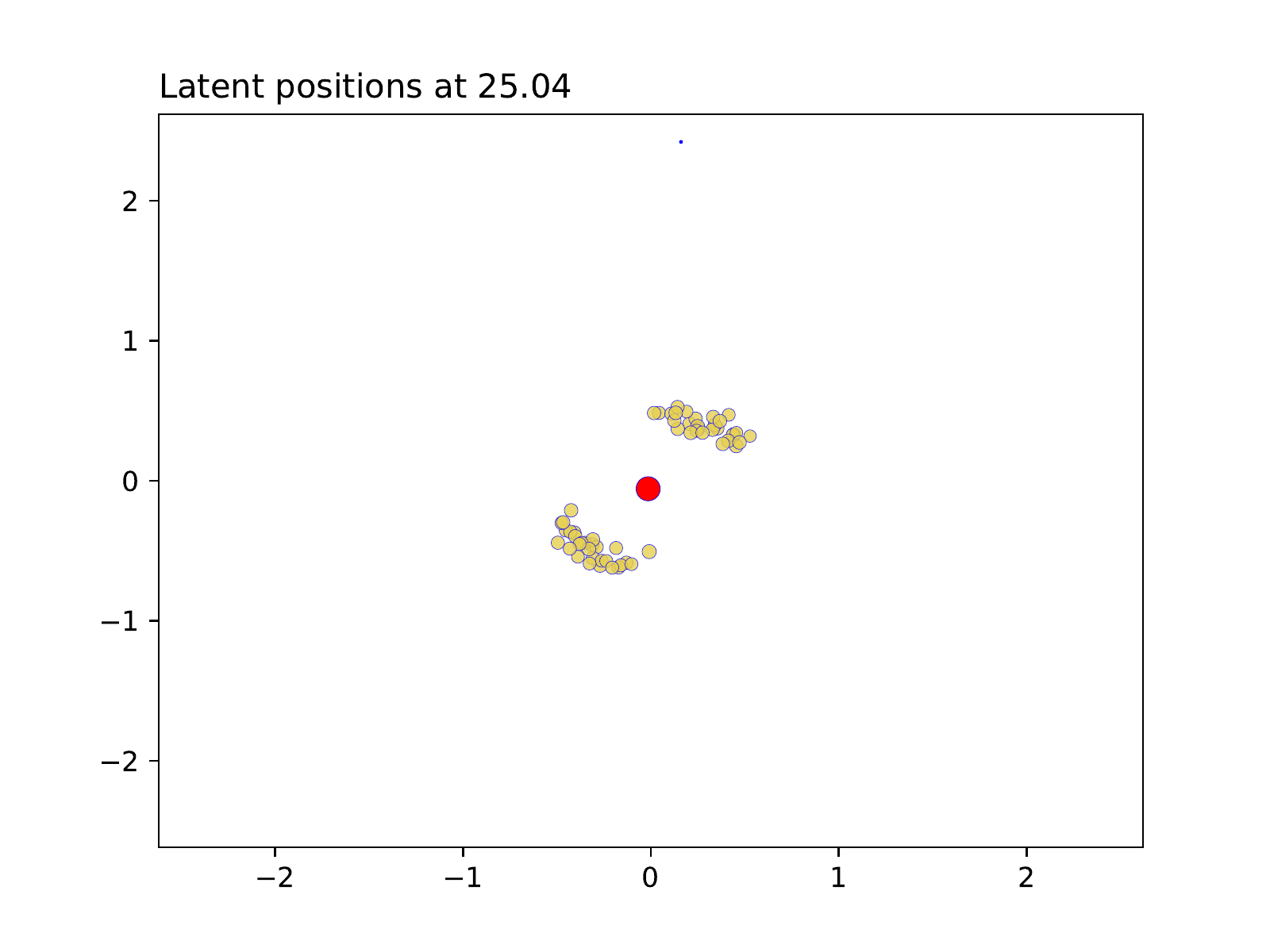}
     \end{subfigure}
     \begin{subfigure}[b]{0.495\textwidth}
         \centering
         \includegraphics[width=\textwidth]{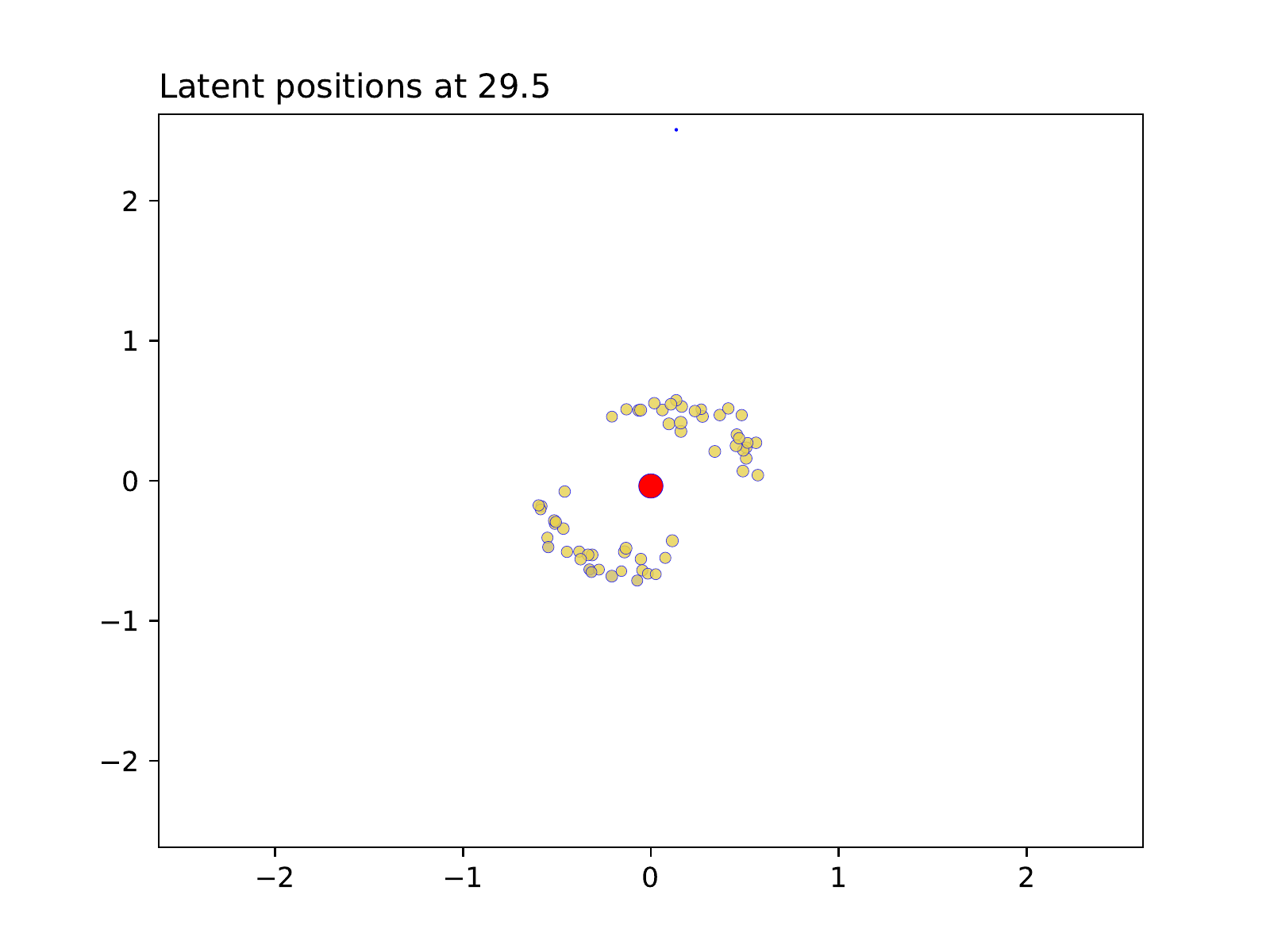}
     \end{subfigure}
     \hfill
     \begin{subfigure}[b]{0.495\textwidth}
         \centering
         \includegraphics[width=\textwidth]{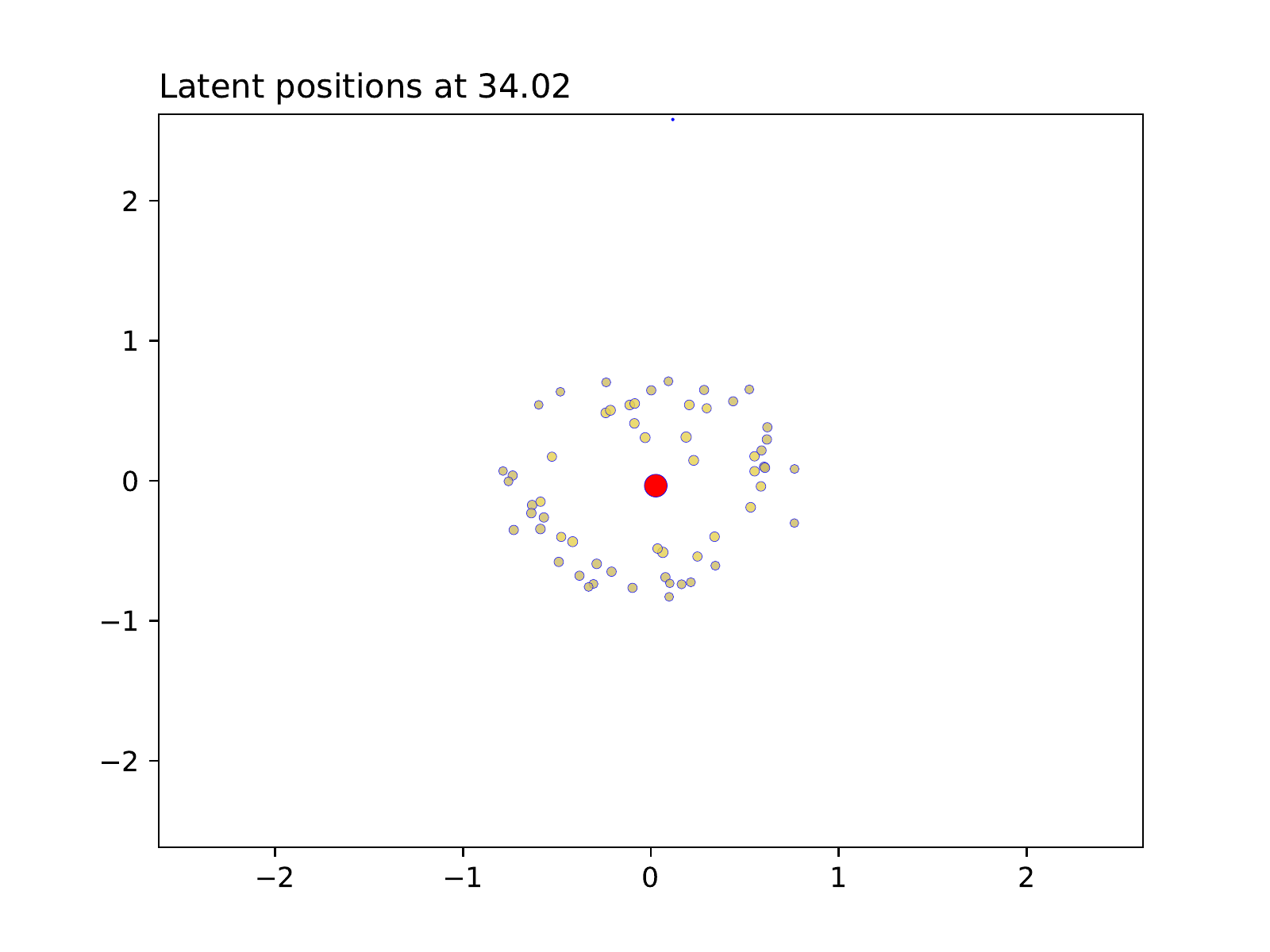}
     \end{subfigure}
        \caption{\textbf{Simulation study 1}: snapshots for the distance model.}
        \label{fig:sim_a_distance_results}
\end{figure}
In this case, the clusters are clearly separated at all times.
The cluster with a strong community structure is less dispersed than the clusters with a weaker community structure. 
The hub is constantly positioned in the centre of the space, as to minimize the distance from all of the nodes at the same time. 
The isolated node is instead wandering in the outskirts of the latent social space.
The creation and dissolution of communities only happens right at the proximity of start/end of each time segment.

\begin{remark}
One main difference between the projection model and the distance model is that, in two dimensions, the distance model seems more suited for representing a diverse community structure.
The main reason is that the projection model requires nodes of different communities (which have few interactions with each other) to be distributed along the axes, respectively. 
Indeed , this guarantees that they are close to perpendicular, hence facilitating a strong separation between communities. Now, in dimension two, this clearly can only happen with no more than two communities for the projection model. By contrast, the distance model does not have this limitation, and it can more easily accommodate a large number of completely separated communities.
We use this argument to favour the use of the distance \texttt{CLPM} in two dimensions, in our applications.
However, we note that the projection model may be of interest in higher dimension ($d>>2$) for purposes other than visualization (e.g. clustering or sub-space tracking).
\end{remark}

Technical details regarding the simulation's parameters, including penalization terms and number of change points, can be consulted on the CLPM code repository.

\ignorespacesafterend
\paragraph{Simulation study 2.}\label{subsec:simu_C}
In the second simulation study, we use again a blockmodel structure, however, in this case we approximate a continuous time framework by defining very short time segments, and letting the communities change from one time segment to the next. 
Since creations and dissolutions of communities would be unlikely in such a short period of time, we keep the community memberships unchanged, and we progressively increase the cohesiveness of the communities.
This means that we progressively increase the rates of interactions between any pairs of nodes that belong to the same community, while keeping any other rate constant.
The rate of interactions within each community starts at value $1$ and increases in a step-wise fashion over $40$ segments, up to the value $5$. The time interval is $[0,40]$, and we consider two communities.
Half way through the simulation, a special node moves from one community to the other.

For the projection \texttt{CLPM}, we show the results in Figure \ref{fig:sim_c_projection_results}, whereas Figure \ref{fig:sim_c_distance_results} shows the results for the distance model.
\begin{figure}
     \centering
     \begin{subfigure}[b]{0.495\textwidth}
         \centering
         \includegraphics[width=\textwidth]{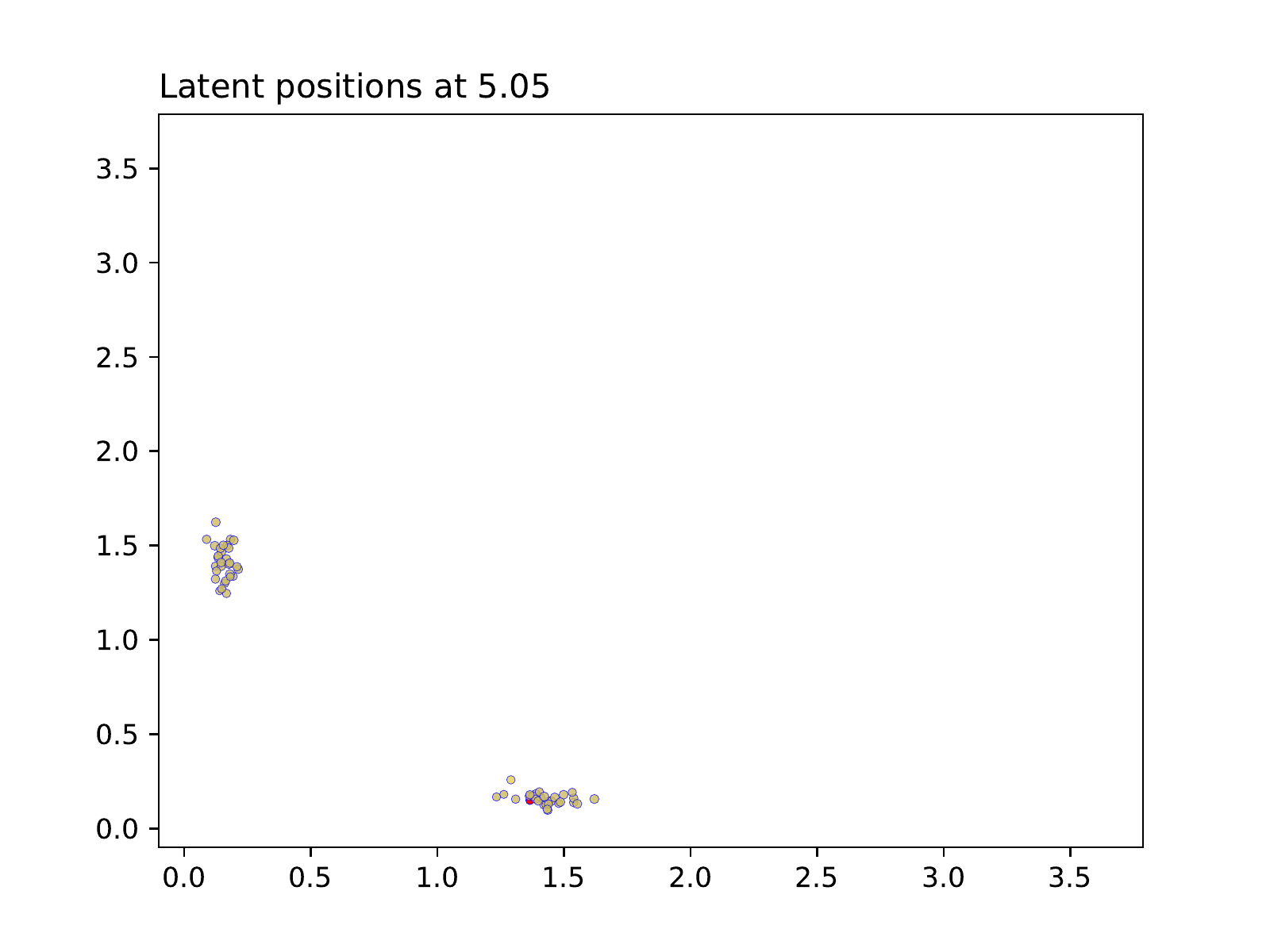}
     \end{subfigure}
     \hfill
     \begin{subfigure}[b]{0.495\textwidth}
         \centering
         \includegraphics[width=\textwidth]{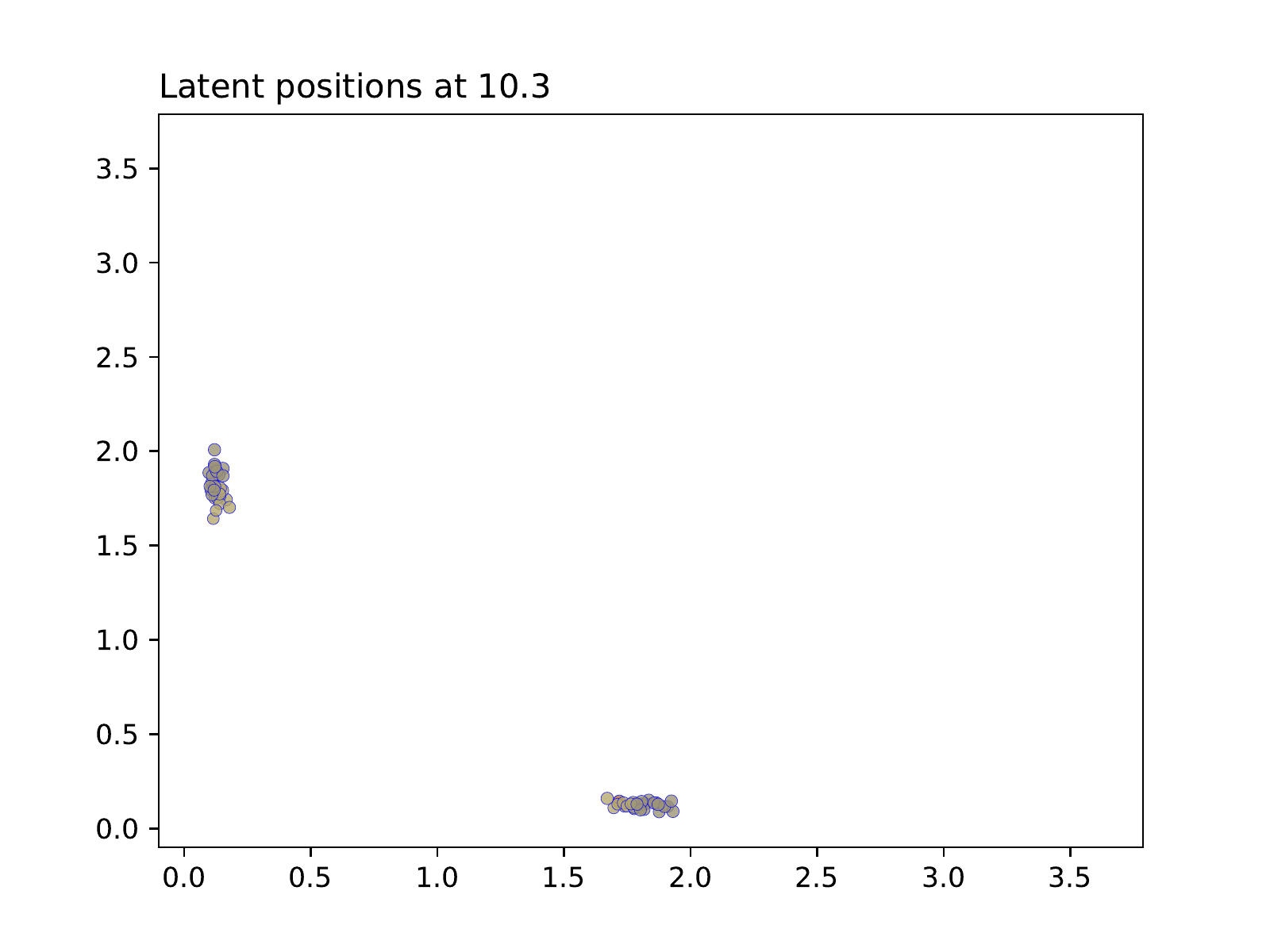}
     \end{subfigure}
     \centering
     \begin{subfigure}[b]{0.495\textwidth}
         \centering
         \includegraphics[width=\textwidth]{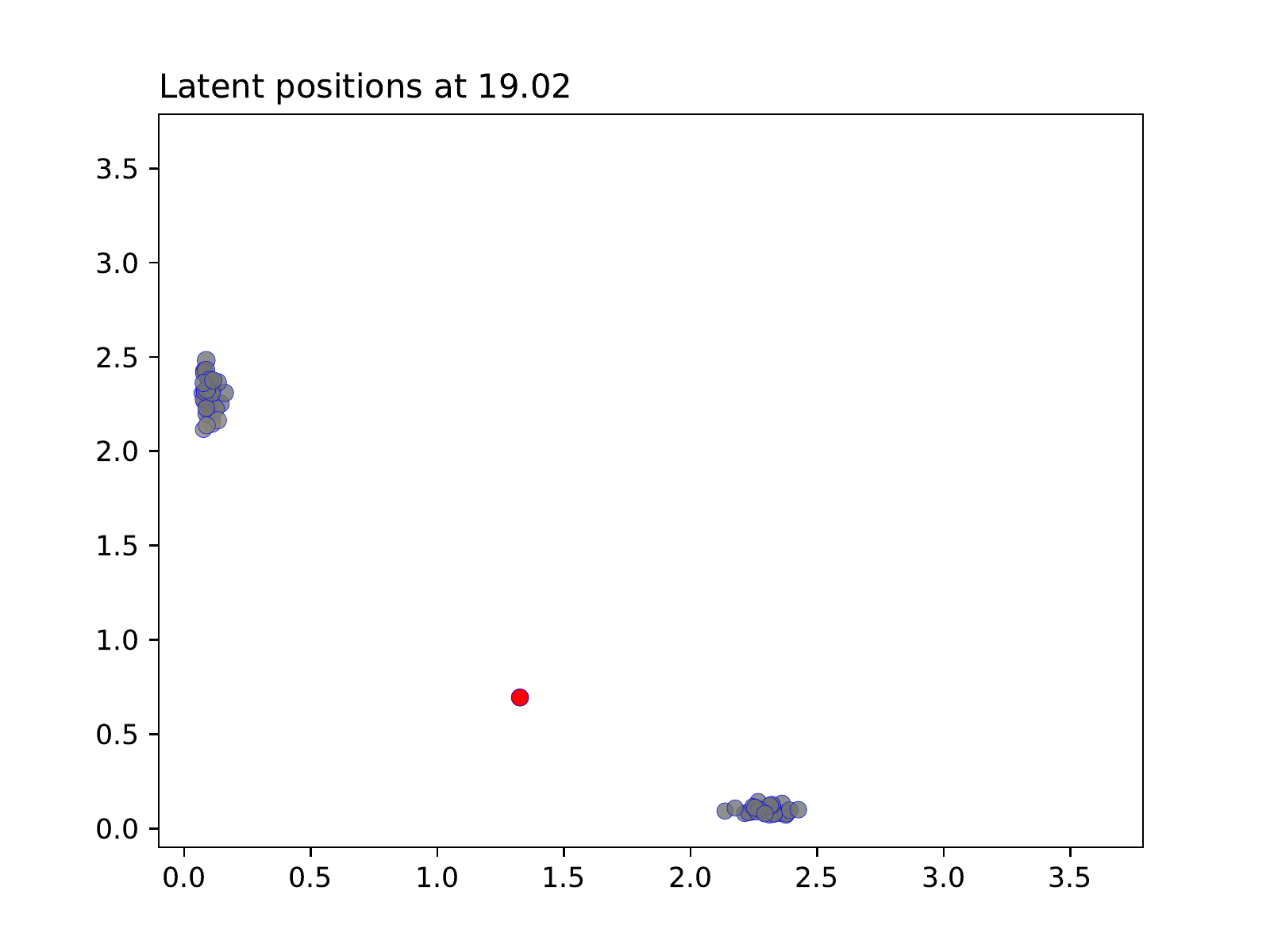}
     \end{subfigure}
     \hfill
     \begin{subfigure}[b]{0.495\textwidth}
         \centering
         \includegraphics[width=\textwidth]{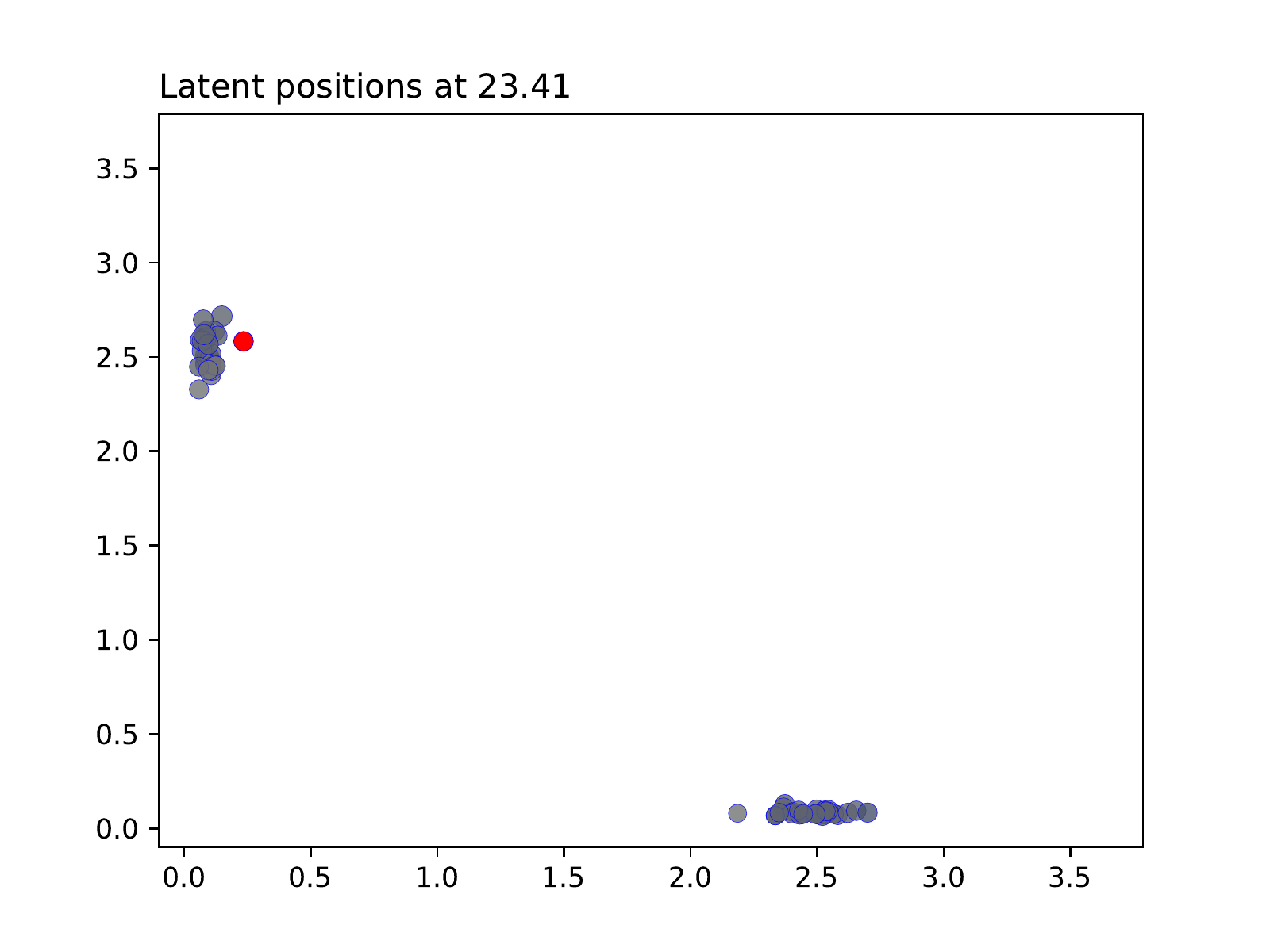}
     \end{subfigure}
        \caption{\textbf{Simulation study 2}: snapshots for the projection model.}
        \label{fig:sim_c_projection_results}
\end{figure}
\begin{figure}
     \centering
     \begin{subfigure}[b]{0.495\textwidth}
         \centering
         \includegraphics[width=\textwidth]{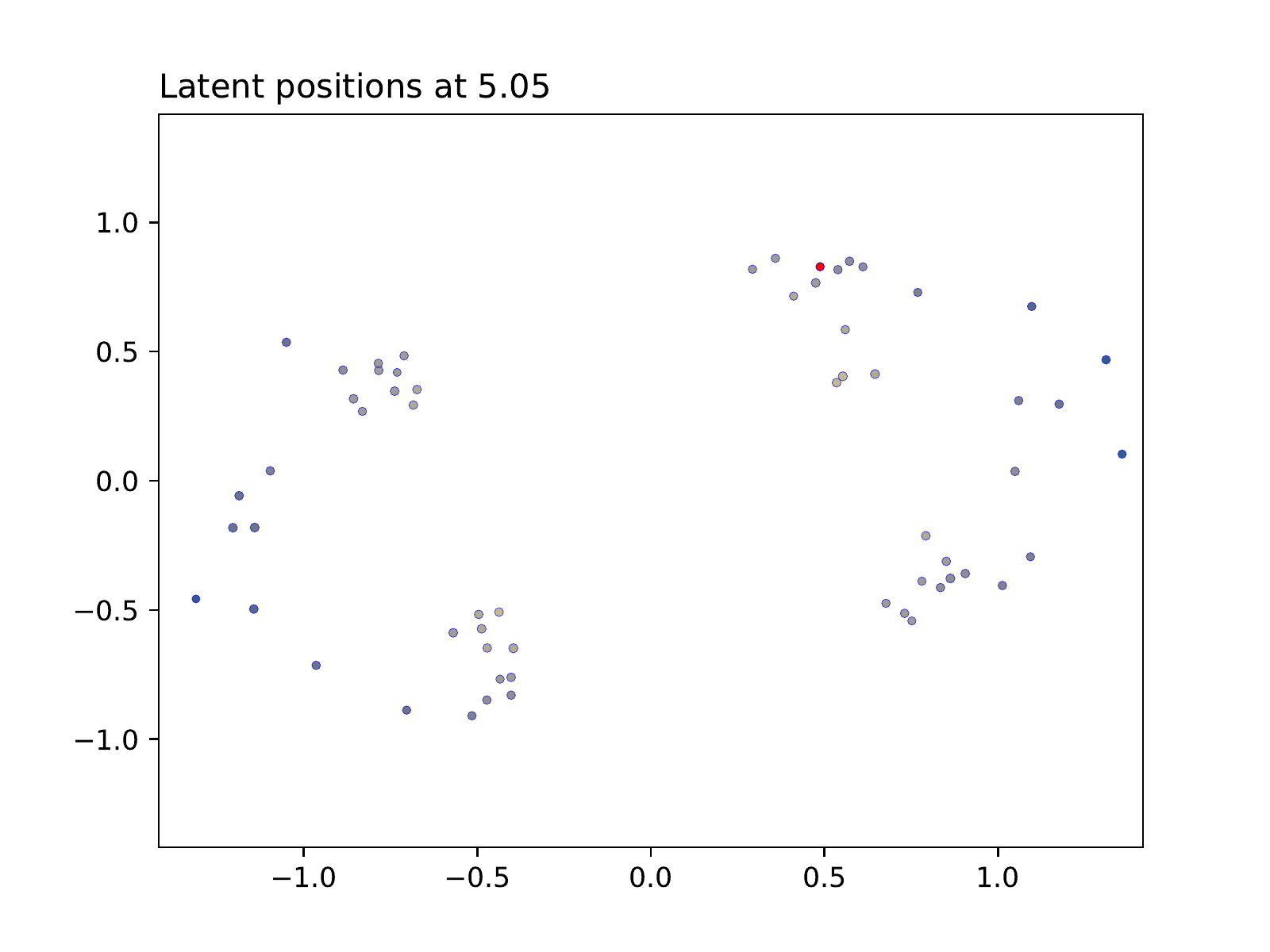}
     \end{subfigure}
     \hfill
     \begin{subfigure}[b]{0.495\textwidth}
         \centering
         \includegraphics[width=\textwidth]{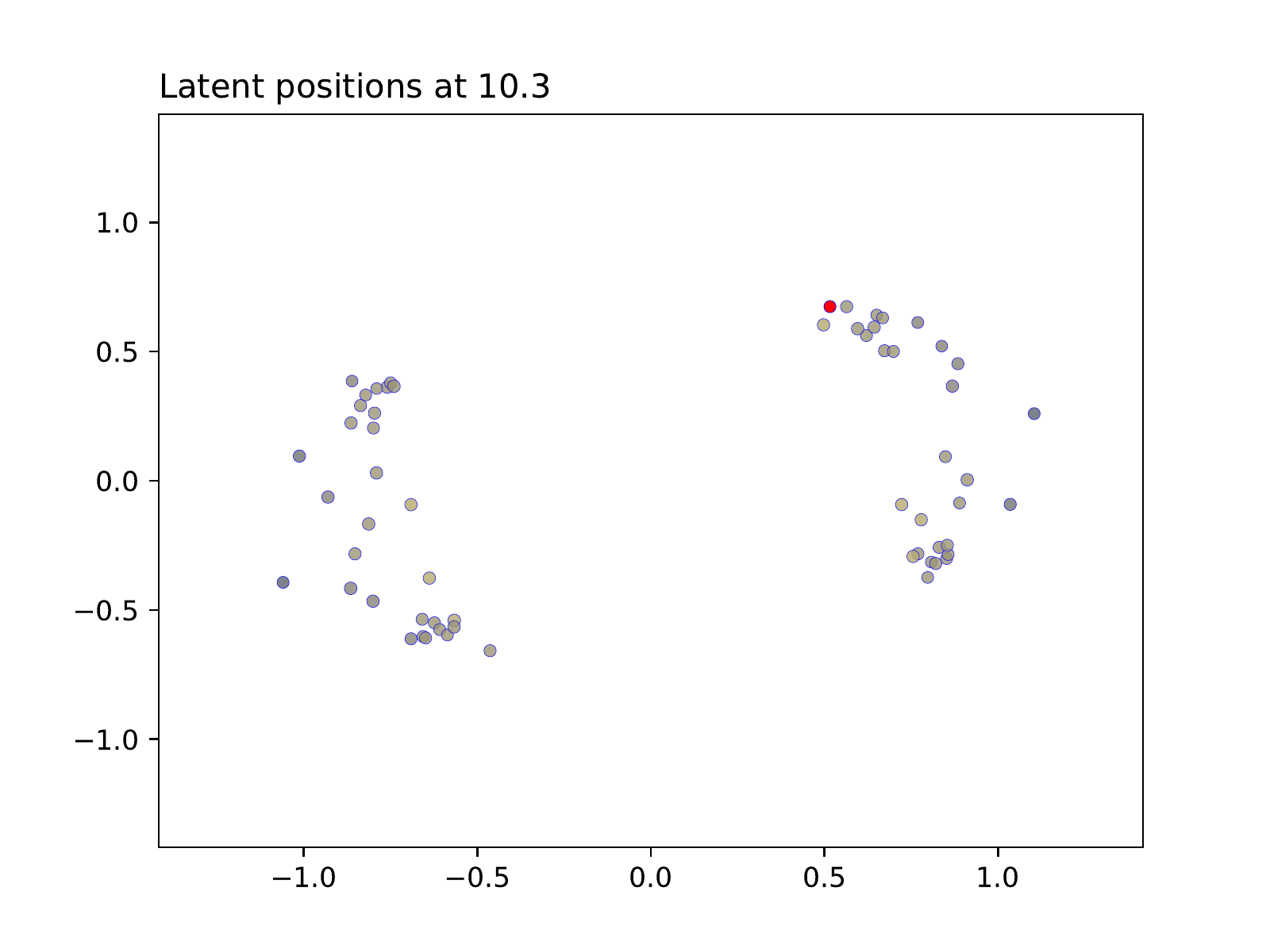}
     \end{subfigure}
     \centering
     \begin{subfigure}[b]{0.495\textwidth}
         \centering
         \includegraphics[width=\textwidth]{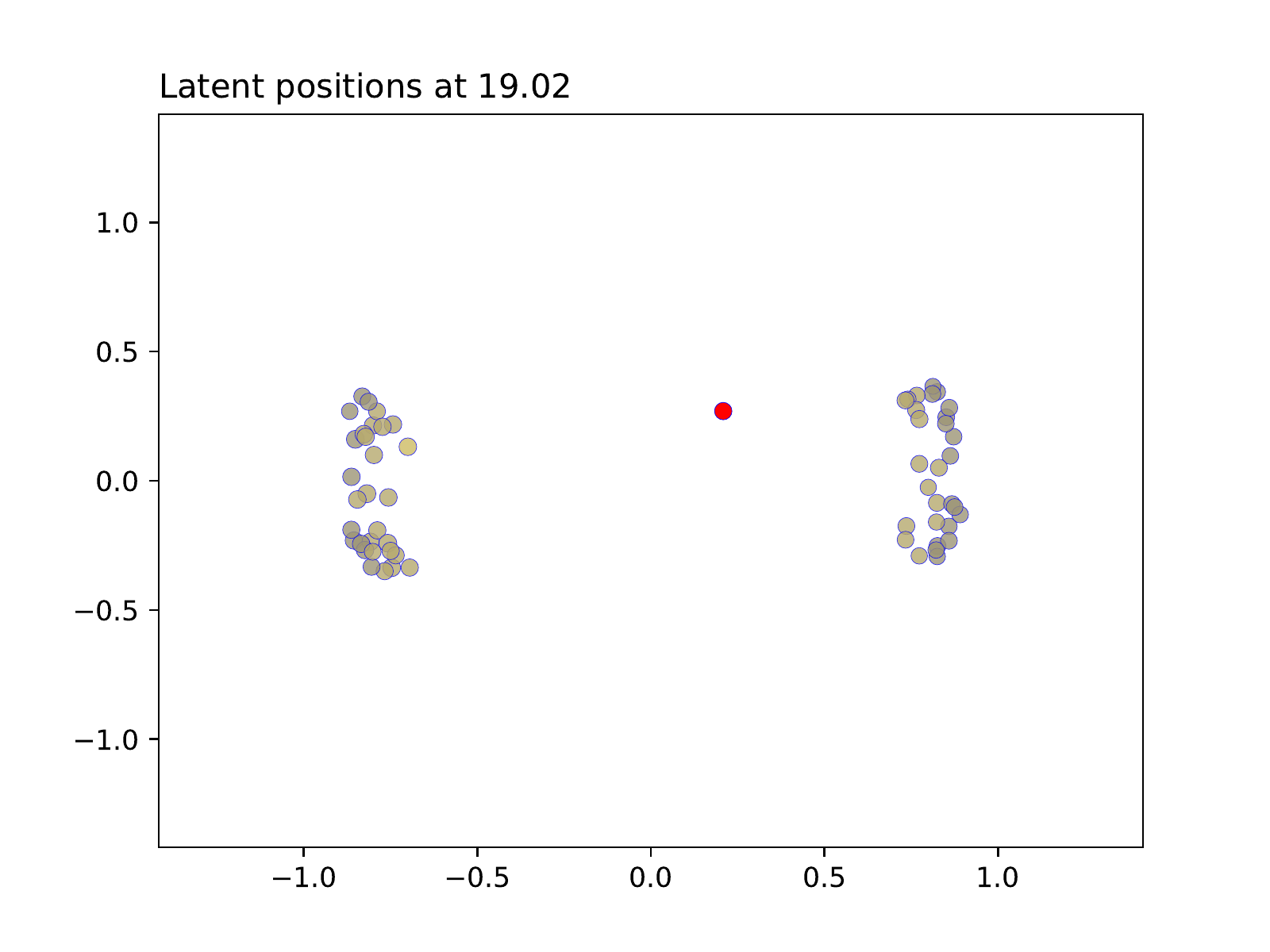}
     \end{subfigure}
     \hfill
     \begin{subfigure}[b]{0.495\textwidth}
         \centering
         \includegraphics[width=\textwidth]{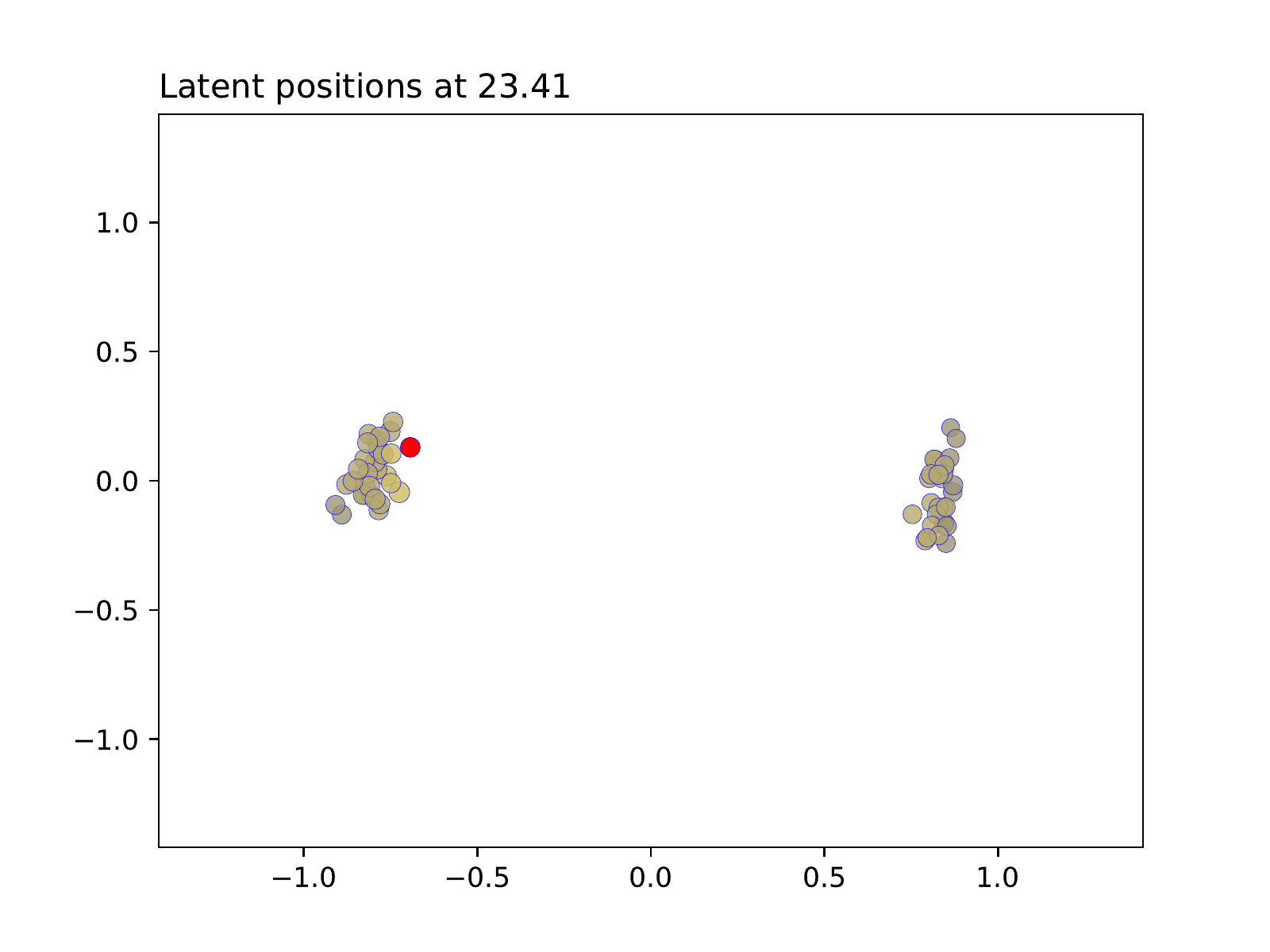}
     \end{subfigure}
        \caption{\textbf{Simulation study 2}: snapshots for the distance model.}
        \label{fig:sim_c_distance_results}
\end{figure}
Both approaches clearly capture the reinforcement of the communities over time by aggregating the nodes of each group.
We observe this behaviour both for the projection model and for the distance model.
The projection model also exhibits nodes getting farther from the centre of the space, since this would give them higher interaction rates, overall.
As concerns the special node moving from one community to the other, this is well captured in that the node transitions smoothly after approximately $20$ seconds, in both models.

\subsection{Distance Model}\label{subsec:simu_B}
\paragraph{Simulation study 3.} In this simulation study, we generate data from the latent distance model itself (Section~\ref{subsubsec:DM}).
In this case, our goal can be more ambitious and thus we aim at reconstructing the individual trajectory of each of the nodes, at every point in time, as accurately as possible.
To make the reading of the results easier, we assume that the nodes move along some pre-determined trajectories that are easy to visualize.
The $N=20$ nodes start on a ring which is centered at the origin of the space, and has radius equal to $1$.
The nodes are located consecutively and in line along the ring, with equal space in between any two consecutive nodes.
Then, they start to move at constant speed towards the centre of the space, which they reach after $5$ seconds.
After reaching the centre, they perform the same motion backwards, and they are back at their initial positions after $5$ more seconds.
The trajectories of the nodes make it so that, when the nodes are along the largest ring, their rate of interaction is essentially zero, however the rate increases as they are closer and closer to the center of the space.

Figure \ref{fig:sim_b_projection_results} shows a collection of snapshots for the projection model.
\begin{figure}
     \centering
     \begin{subfigure}[b]{0.495\textwidth}
         \centering
         \includegraphics[width=\textwidth]{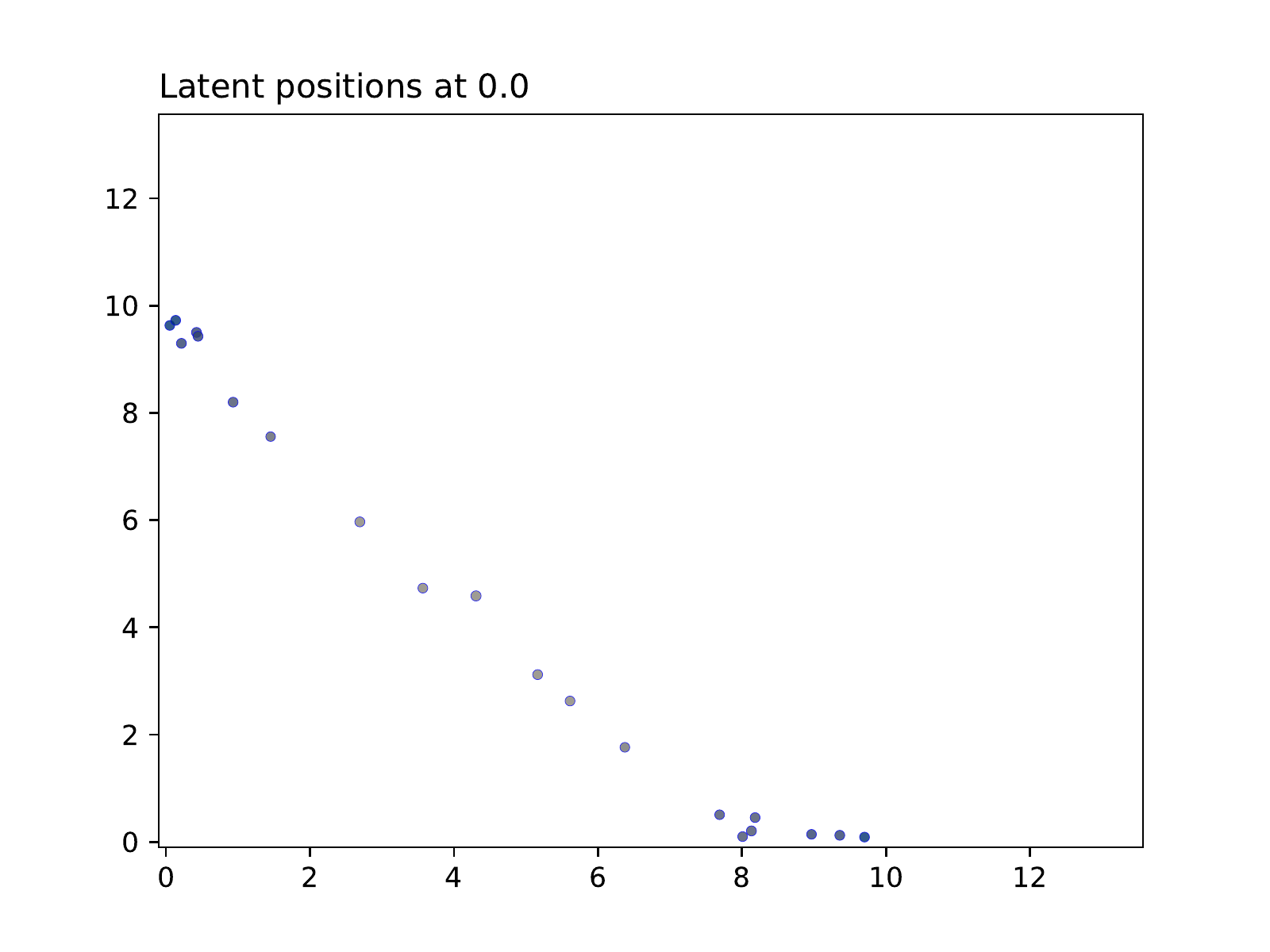}
     \end{subfigure}
     \hfill
     \begin{subfigure}[b]{0.495\textwidth}
         \centering
         \includegraphics[width=\textwidth]{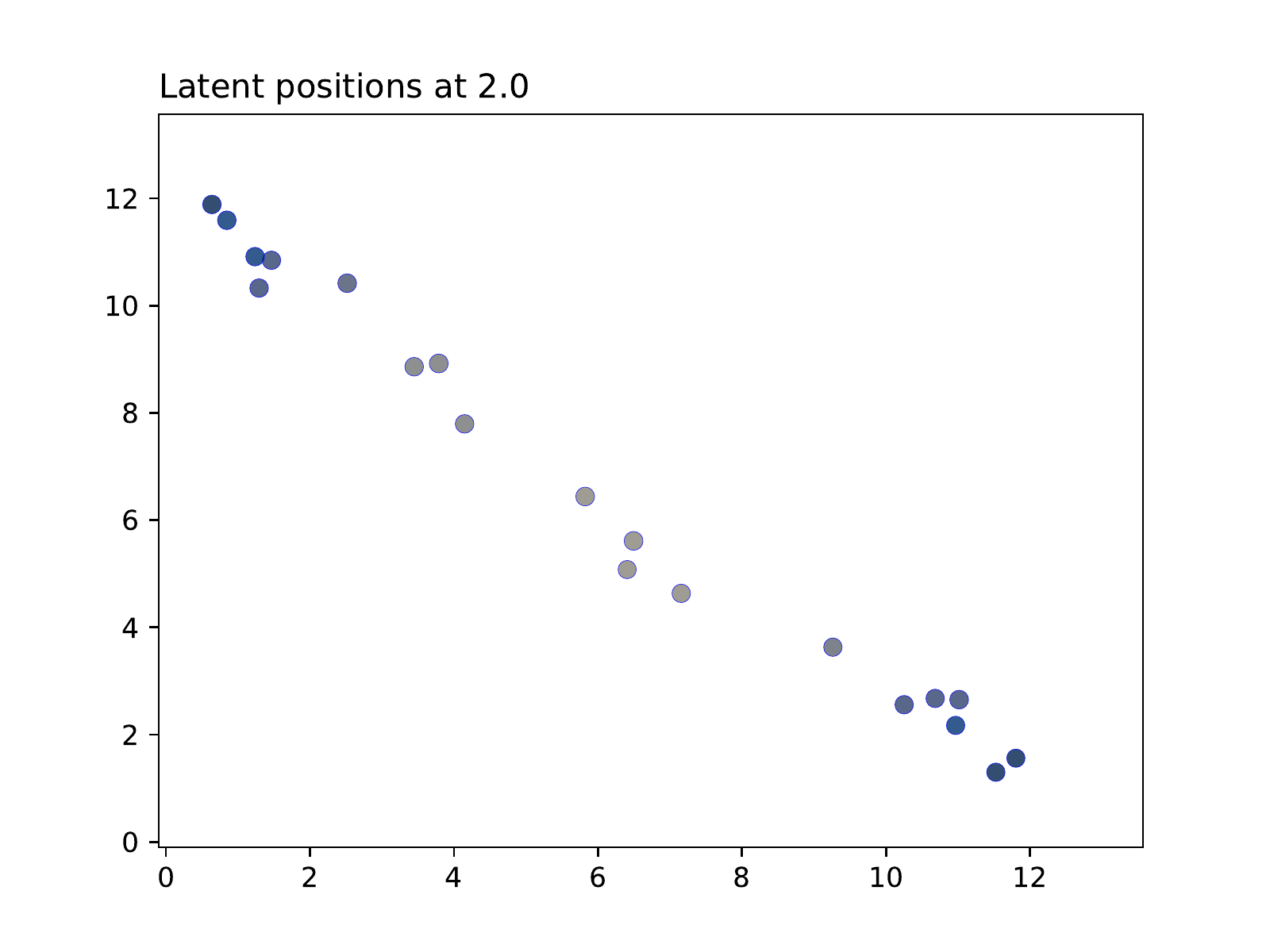}
     \end{subfigure}
     \centering
     \begin{subfigure}[b]{0.495\textwidth}
         \centering
         \includegraphics[width=\textwidth]{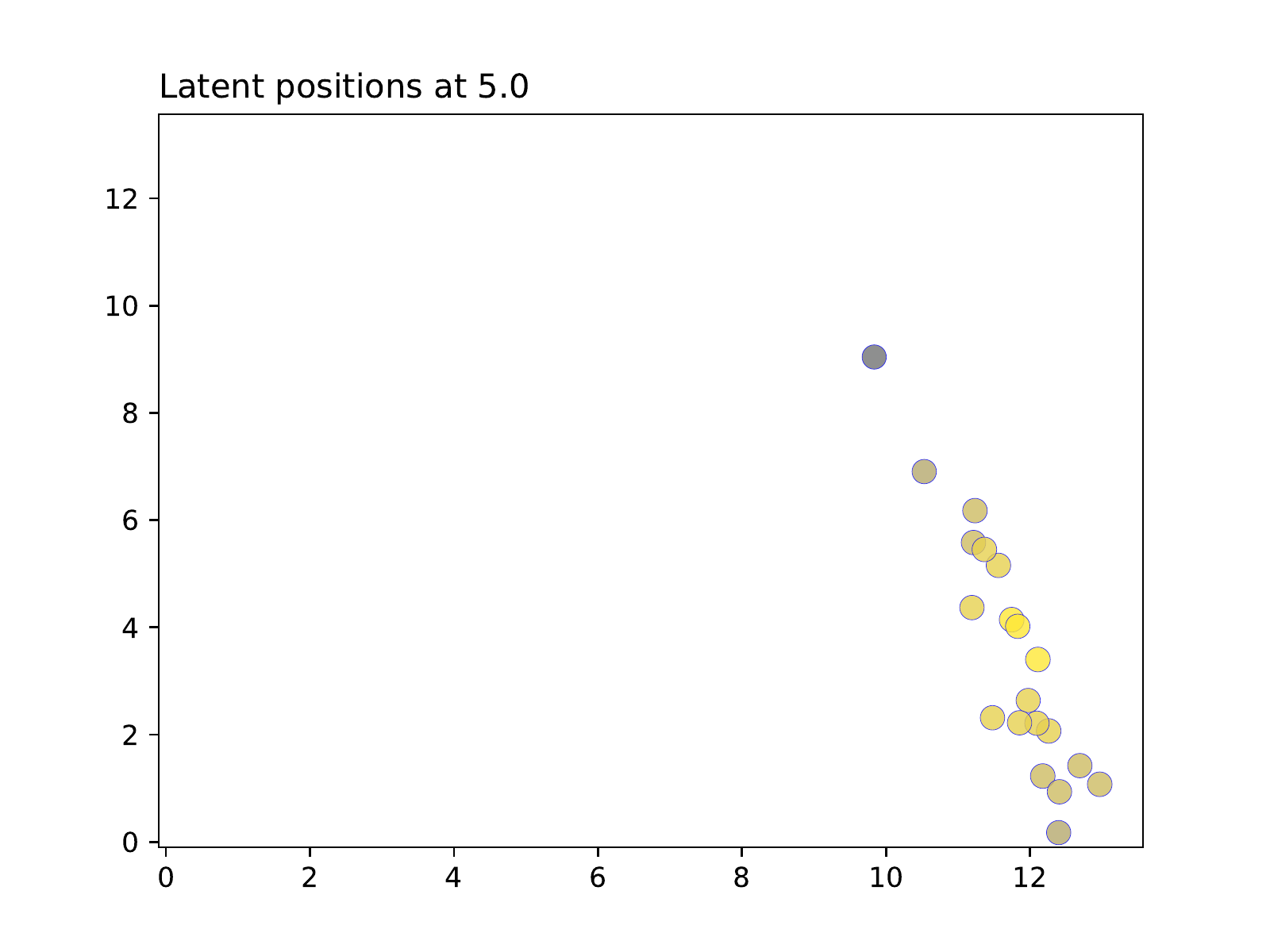}
     \end{subfigure}
     \hfill
     \begin{subfigure}[b]{0.495\textwidth}
         \centering
         \includegraphics[width=\textwidth]{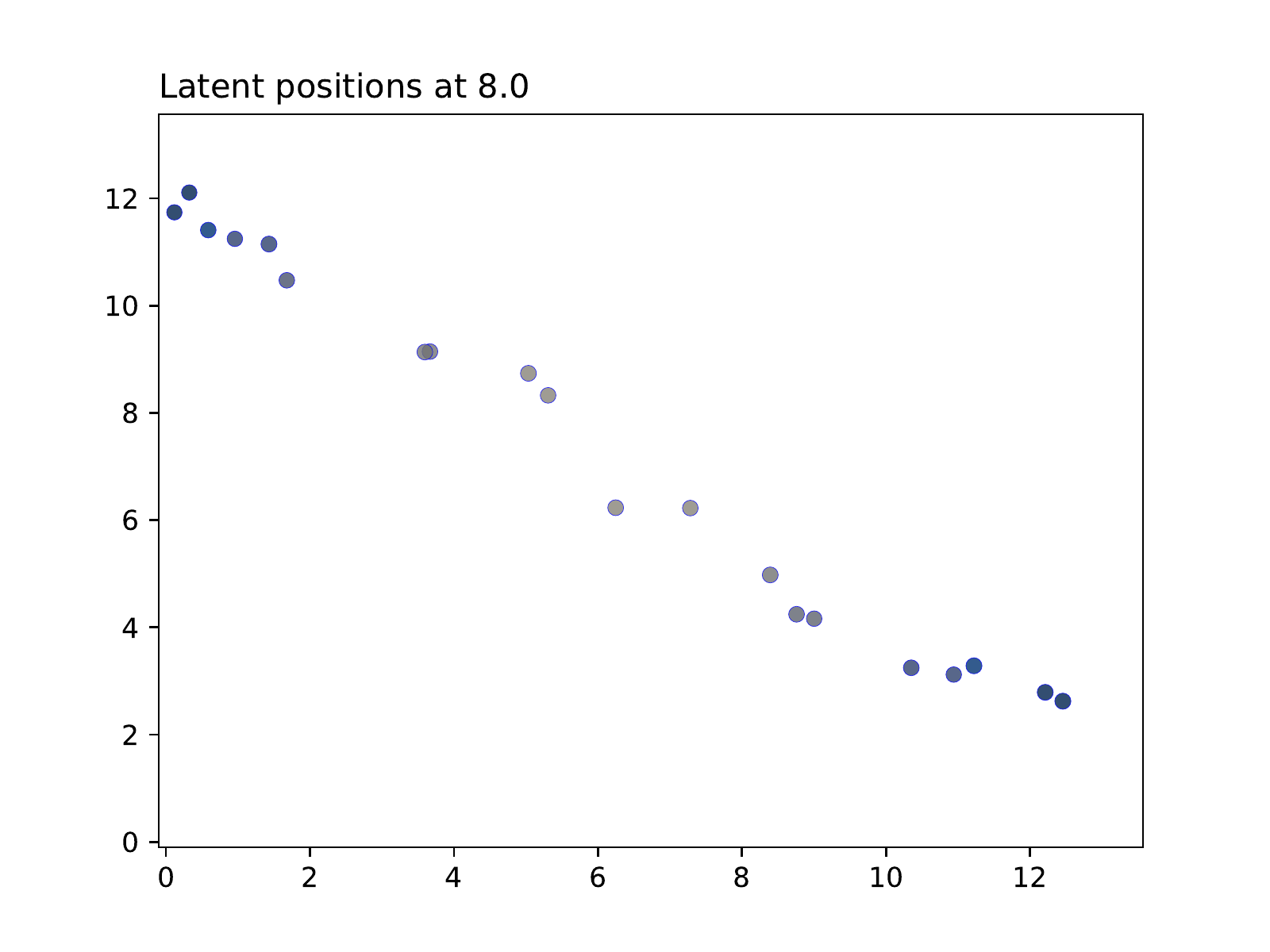}
     \end{subfigure}
        \caption{\textbf{Simulation study 3}: snapshots for the projection model.}
        \label{fig:sim_b_projection_results}
\end{figure}
The nodes are approximately equally spaced along a line and they progress outwards from the centre of the space. 
As they get far apart from the centre and from each other, their dot products increase and so do their interaction rates.
The projection model, which is not the same model that has generated the data, tends to spread out the nodes on the space, which is ideal and expected from these data.
However, this means that some of the nodes almost point in perpendicular directions, which is at odds with the fact that, half-way through the study, all nodes should interact with all others.

As concerns the results for the distance model, these are shown in Figure \ref{fig:sim_b_distance_results}, and they highlight that the true trajectories are essentially accurately recovered.
\begin{figure}
     \centering
     \begin{subfigure}[b]{0.495\textwidth}
         \centering
         \includegraphics[width=\textwidth]{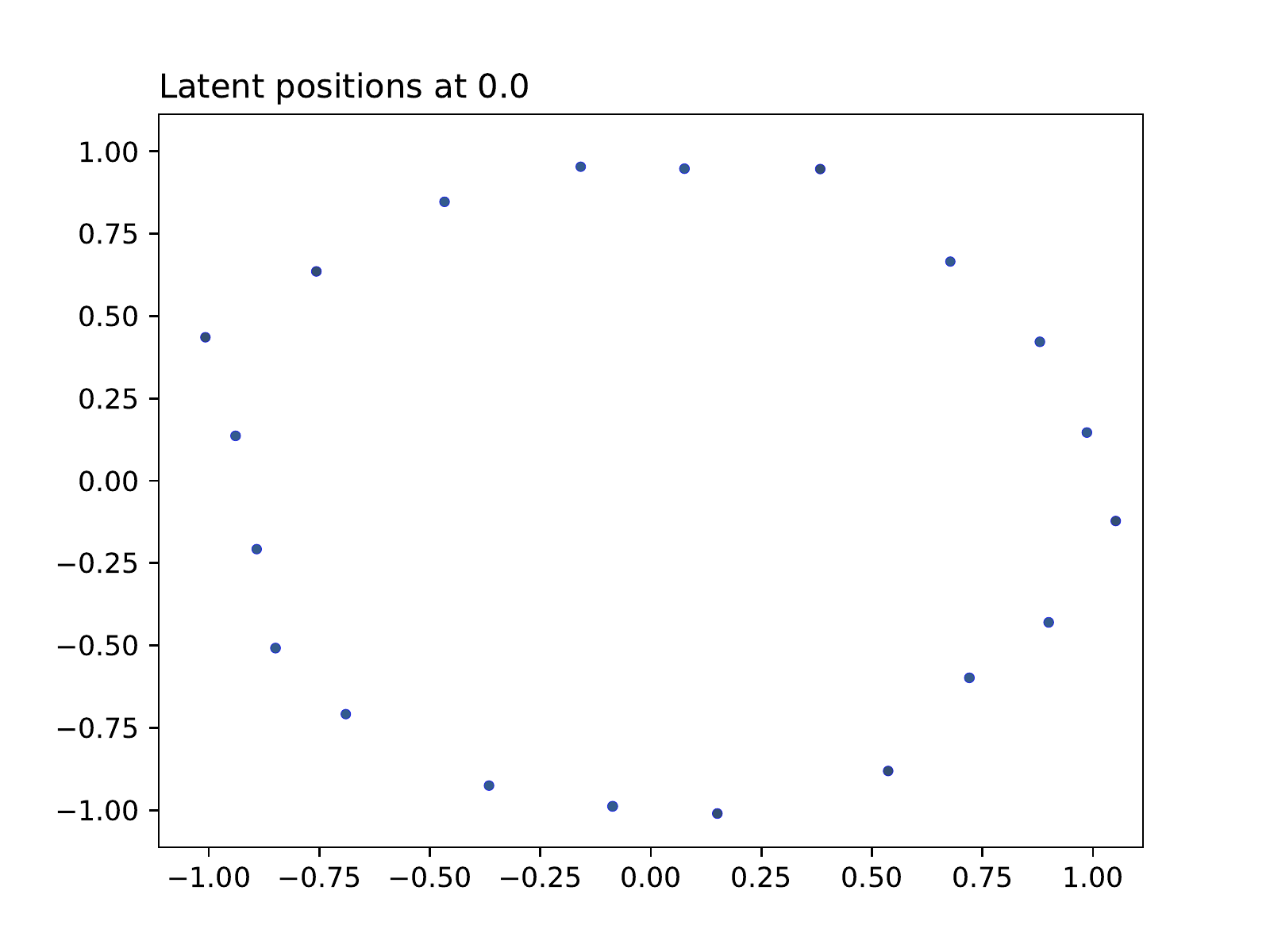}
     \end{subfigure}
     \hfill
     \begin{subfigure}[b]{0.495\textwidth}
         \centering
         \includegraphics[width=\textwidth]{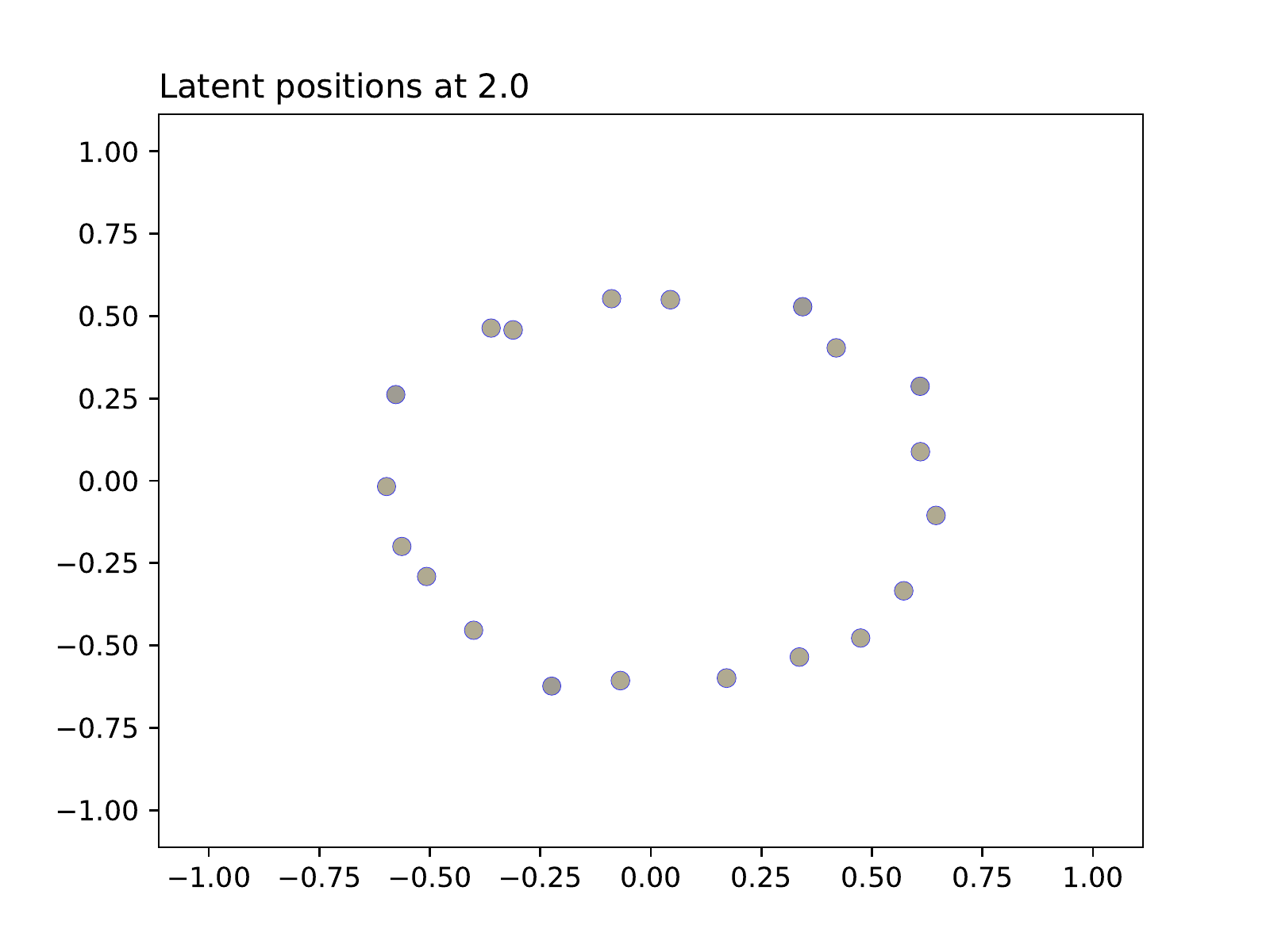}
     \end{subfigure}
     \centering
     \begin{subfigure}[b]{0.495\textwidth}
         \centering
         \includegraphics[width=\textwidth]{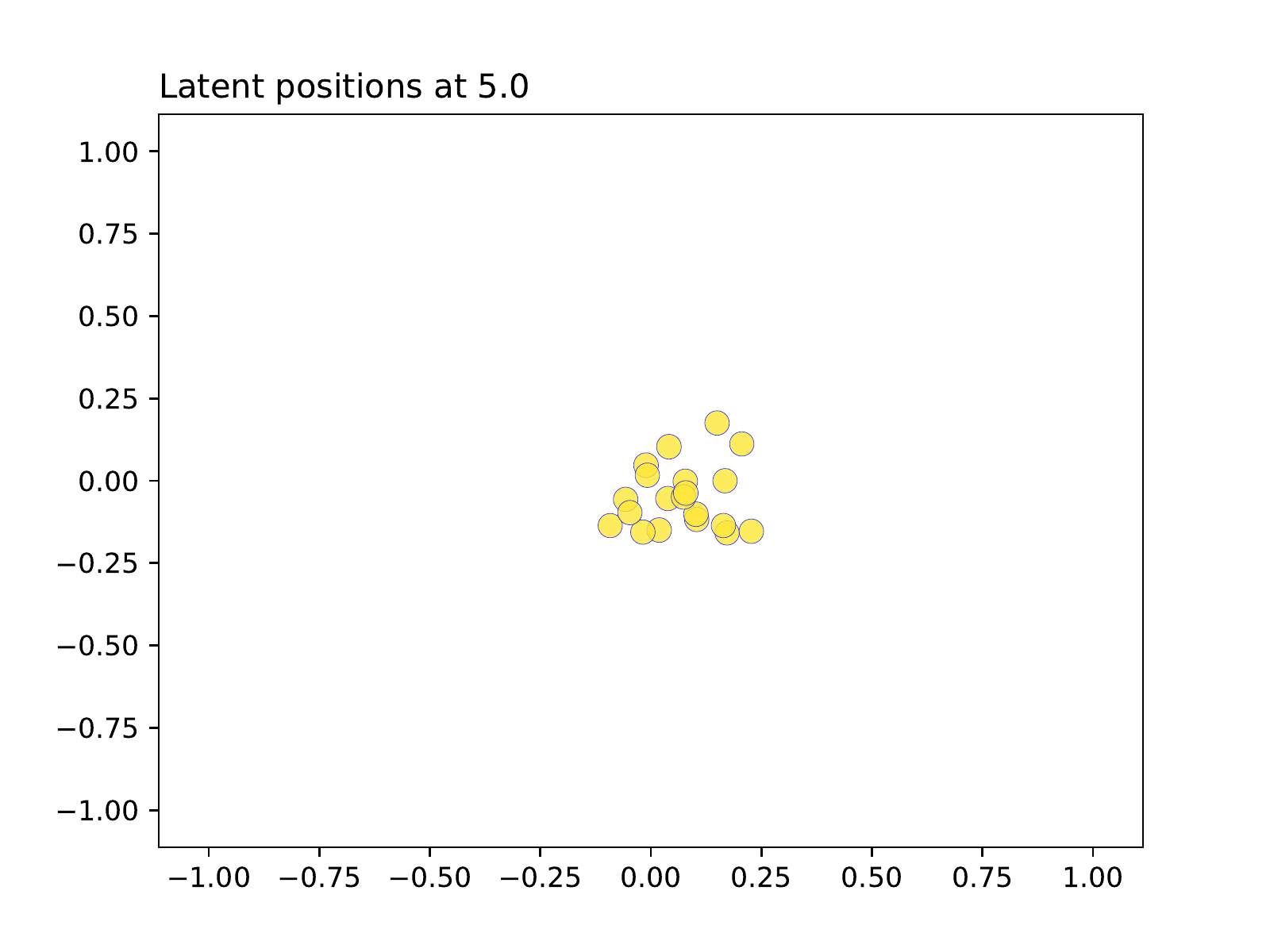}
     \end{subfigure}
     \hfill
     \begin{subfigure}[b]{0.495\textwidth}
         \centering
         \includegraphics[width=\textwidth]{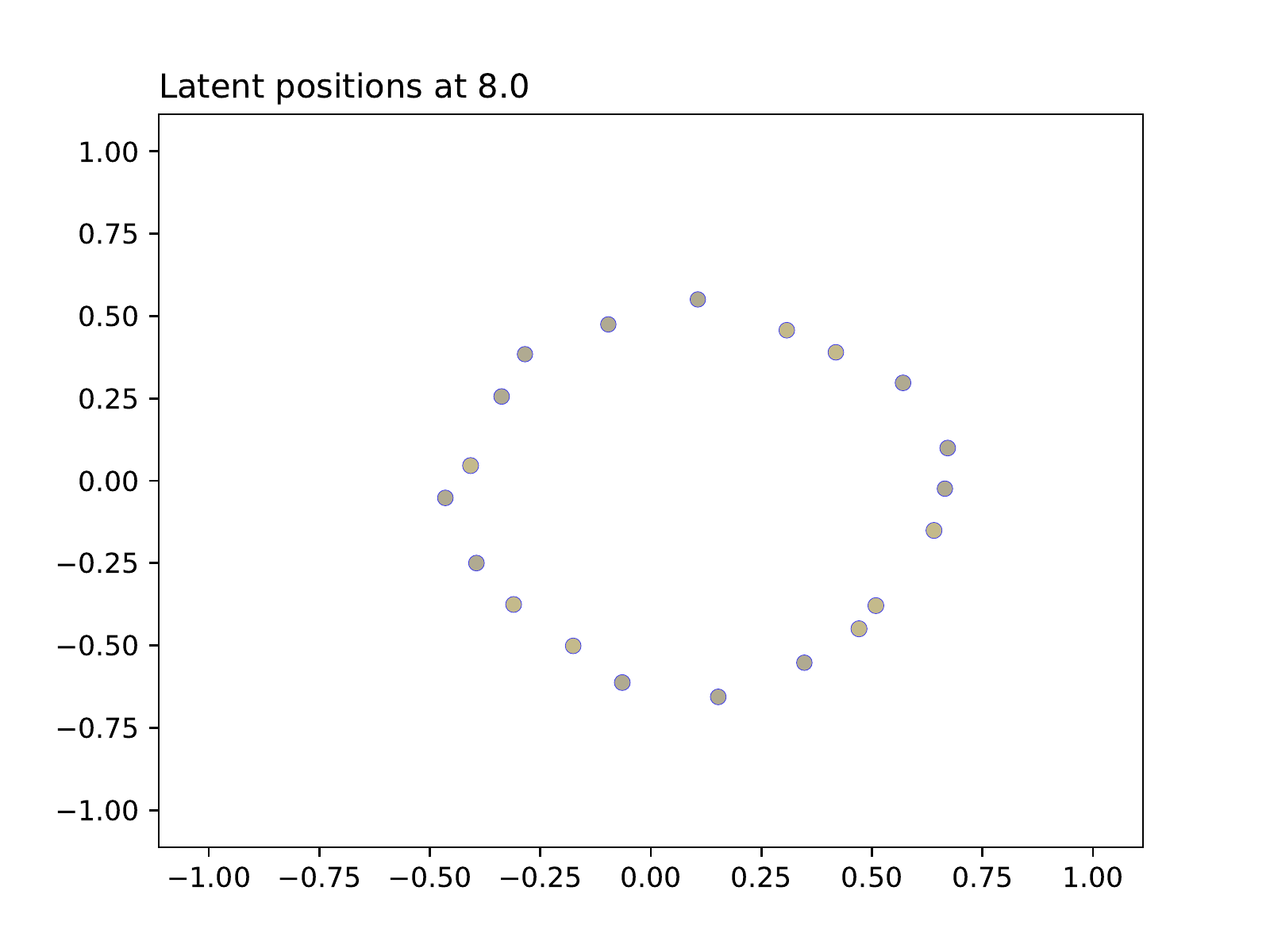}
     \end{subfigure}
        \caption{\textbf{Simulation study 3}: snapshots for the distance model.}
        \label{fig:sim_b_distance_results}
\end{figure}
The model can capture really well the contraction and expansion of the latent space, and the individual trajectories of the nodes are closely following the theoretical counterparts.
The scale of the latent space is also correctly estimated since the largest ring has approximately radius $1$.

There are some important remarks to make.
First, after $5$ seconds, i.e. when all nodes are located close to the centre, it is understandable that a rotation or reflection (with respect to the origin of the space) may happen.
This is inevitable since the solution can only be recovered up to a rotation/reflection of all the latent trajectories, but also because the first $5$ seconds and the last $5$ seconds can technically be seen as two independent problems.
The collapse to zero can be seen as a reset in terms of orientation of the latent space.
That is because the penalization terms only work with two consecutive change points, so, if we view them as identifiability constraints, they would lose their effectiveness when all the nodes collapse to zero for some time. 
A second fundamental remark is that the estimation procedure can lead to good results only if we observe an appropriate number of interactions. 
This is a specific trait of latent position models in general, since we can only guess the position of one node accurately when we know to whom it connects (or, in this context, how frequently), as we would tend to locate it close to its neighbors.
In our simulated setting, there are few to no interactions when nodes are along the largest ring, so it makes sense that the results seem a bit more noisy in those instants.

\section{Applications}\label{sec:applications}
In this section, we illustrate our approach over $3$ real datasets, highlighting how we can characterize the trajectories of individual nodes, the formation and dissolution of communities, and other types of connectivity patterns.
From the simulation studies, we have pointed out that the distance model generally provides a more convenient and appropriate framework to study these aspects of the data.
In addition, it is also easier to interpret, so, we only show the results for the distance model and redirect the reader to the associated code repository where the complete results can be found.

\subsection{ACM Hypertext conference dataset}
The ACM Hypertext 2009 conference was held over three days in Turin, Italy, from June 29 to July 1. 
At the conference, $113$ attendees wore special badges which recorded an interaction whenever two badges were facing each other at a distance of $1.5$ meters or less, for at least $20$ seconds.
For each of these interactions, a timestamp was recorded as well as the identifiers of the two personal badges.

This interaction dataset was first analysed by \textcite{konect:sociopatterns}, and is publicly available from \textcite{konect:2017:sociopatterns-hypertext}.
Similarly to \textcite{corneli2016exact}, we focus our analysis on the first day of the conference. 
On the first day, the main events that took place included a poster session in the morning (starting from 8 a.m.), a lunch break around 1 p.m., and a cheese and wine reception in the evening between 6 p.m. and 7 p.m.
\begin{figure}[t]
    \centering
    \includegraphics[width=.6\textwidth]{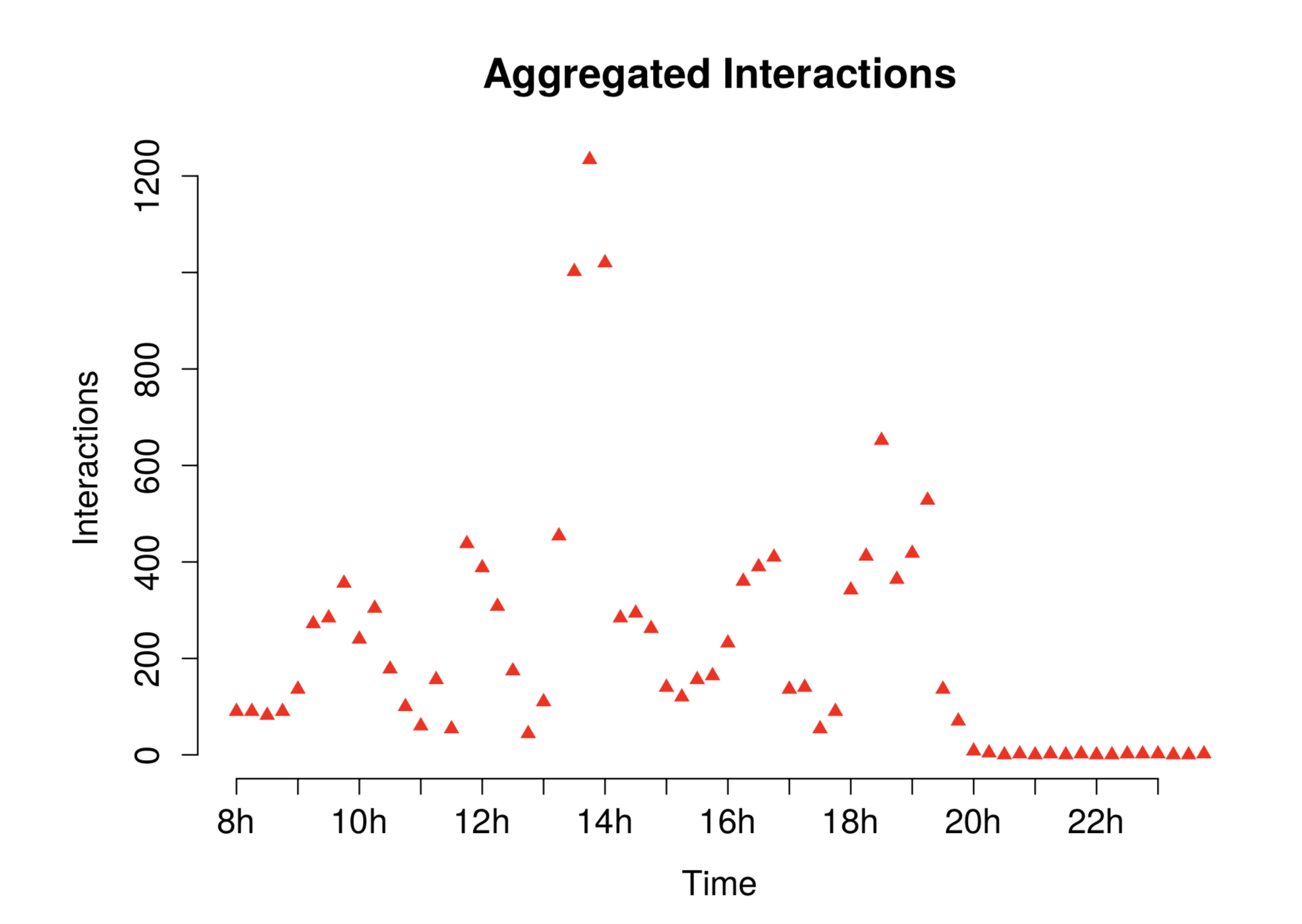}
    \caption{\textbf{ACM application}: cumulative number  of interactions for each quarter hour (first day).}
    \label{fig:acm_hist}
\end{figure}
We use our distance \texttt{CLPM} to provide a graphical representation of these data, and to note how the model responds to the various gatherings that happened during the day. We set a change point ($\eta_k$) every fifteen minutes. 
Figures \ref{fig:acm_1} and \ref{fig:acm_2} provide a number of snapshots highlighting some of the relevant moments of the day. 
\begin{figure}
     \centering
     \begin{subfigure}[b]{0.495\textwidth}
         \centering
         \includegraphics[width=\textwidth]{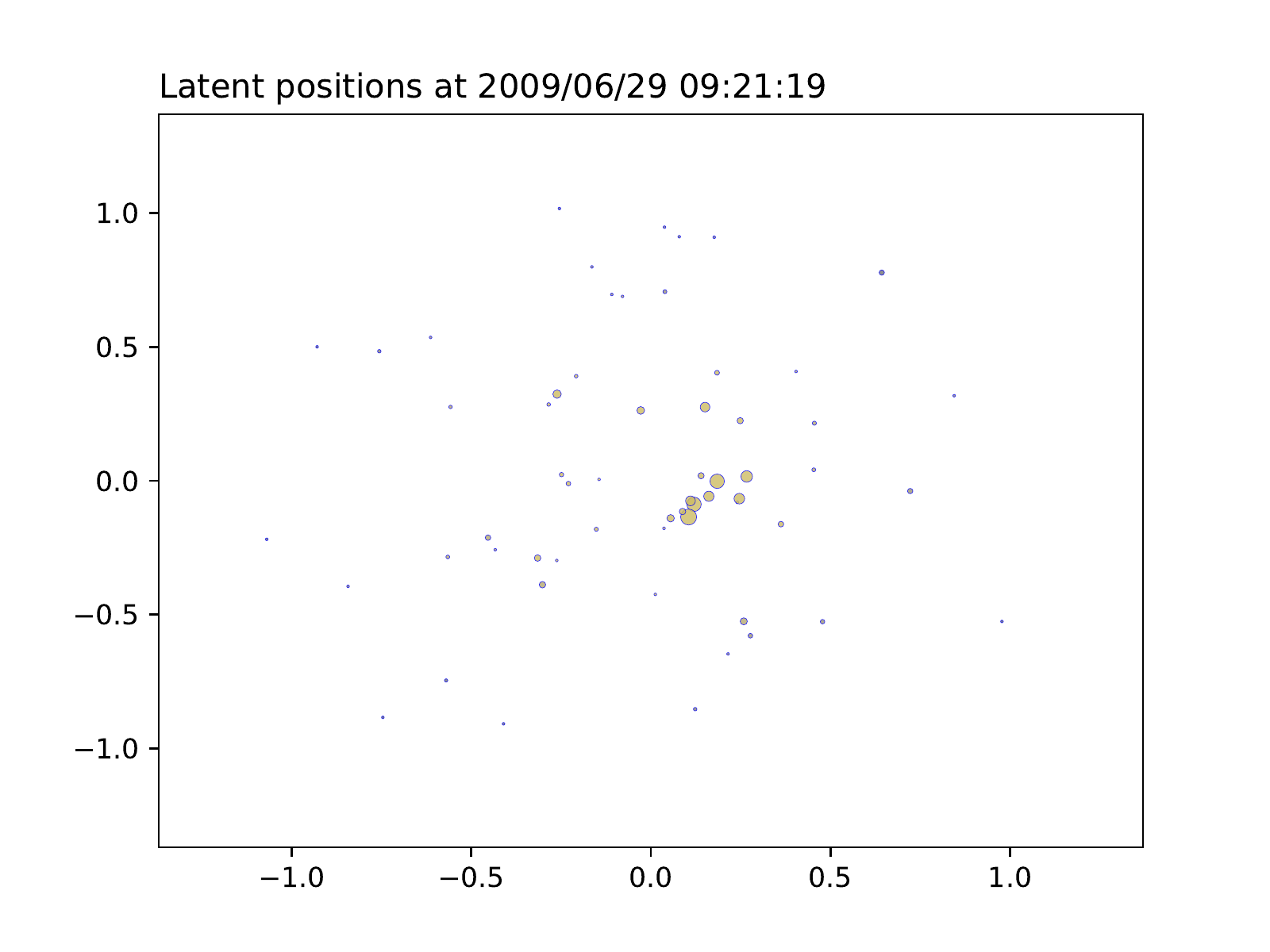}
     \end{subfigure}
     \hfill
     \begin{subfigure}[b]{0.495\textwidth}
         \centering
         \includegraphics[width=\textwidth]{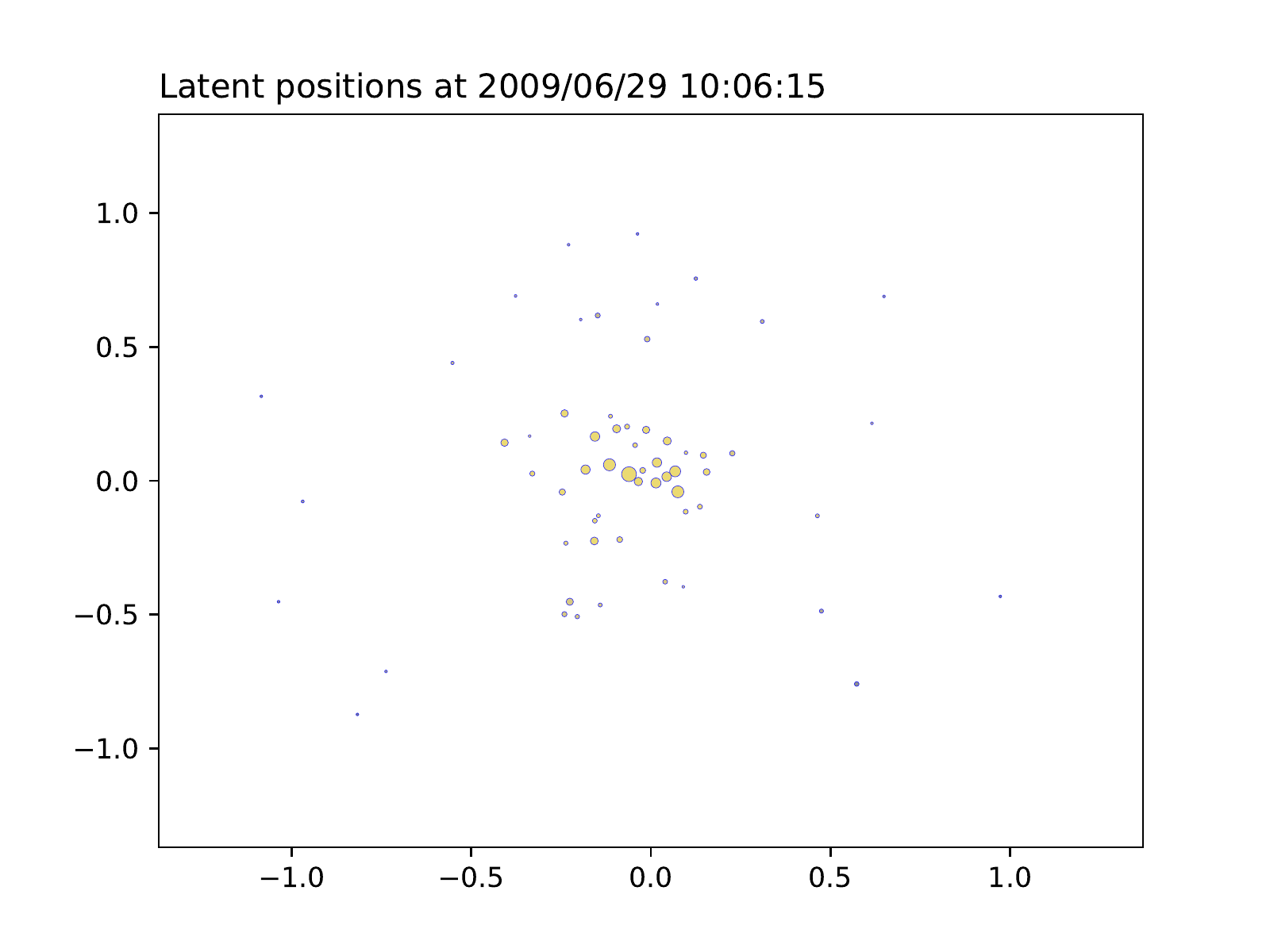}
     \end{subfigure}
     \centering
     \begin{subfigure}[b]{0.495\textwidth}
         \centering
         \includegraphics[width=\textwidth]{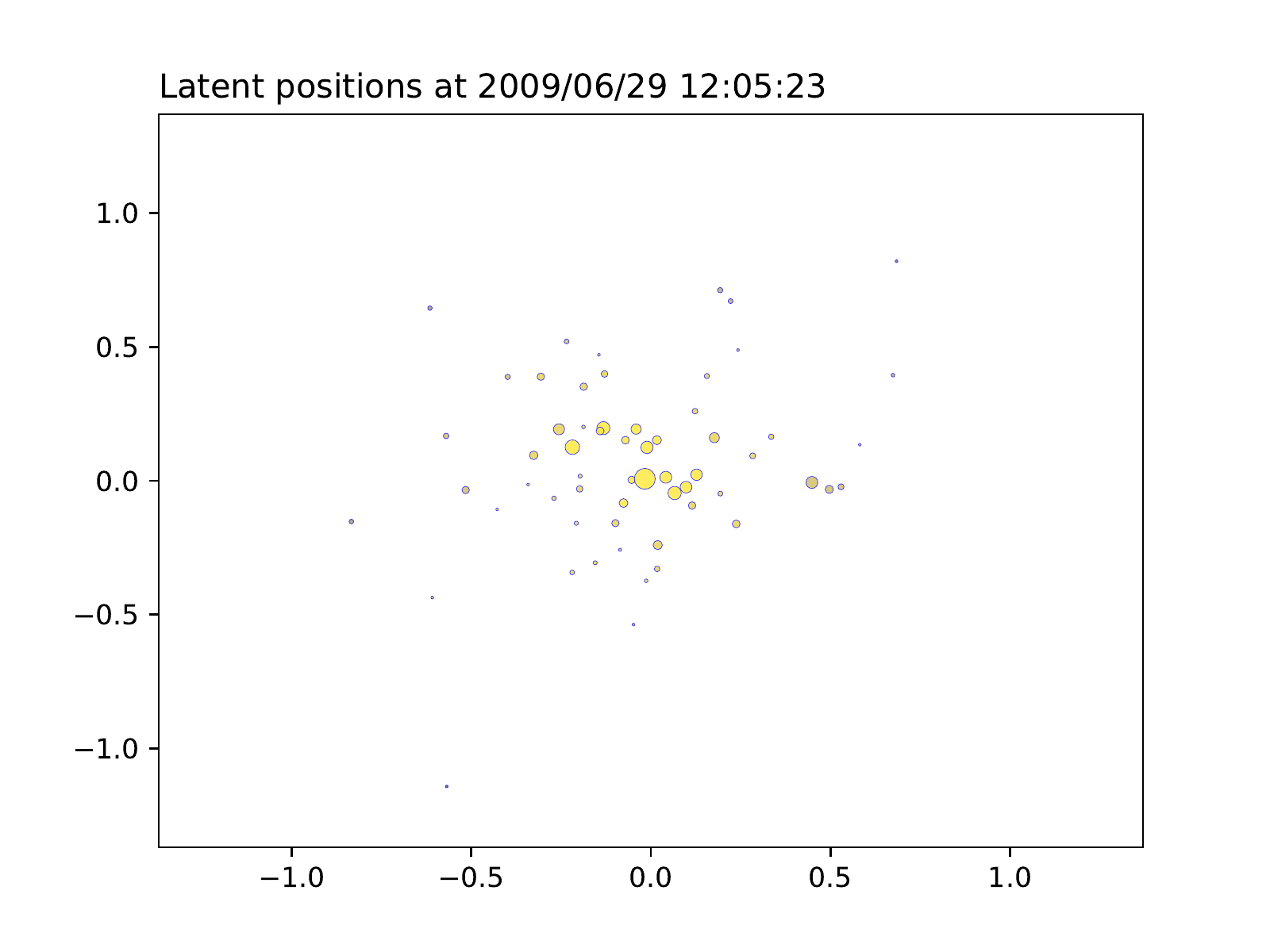}
     \end{subfigure}
     \hfill
     \begin{subfigure}[b]{0.495\textwidth}
         \centering
         \includegraphics[width=\textwidth]{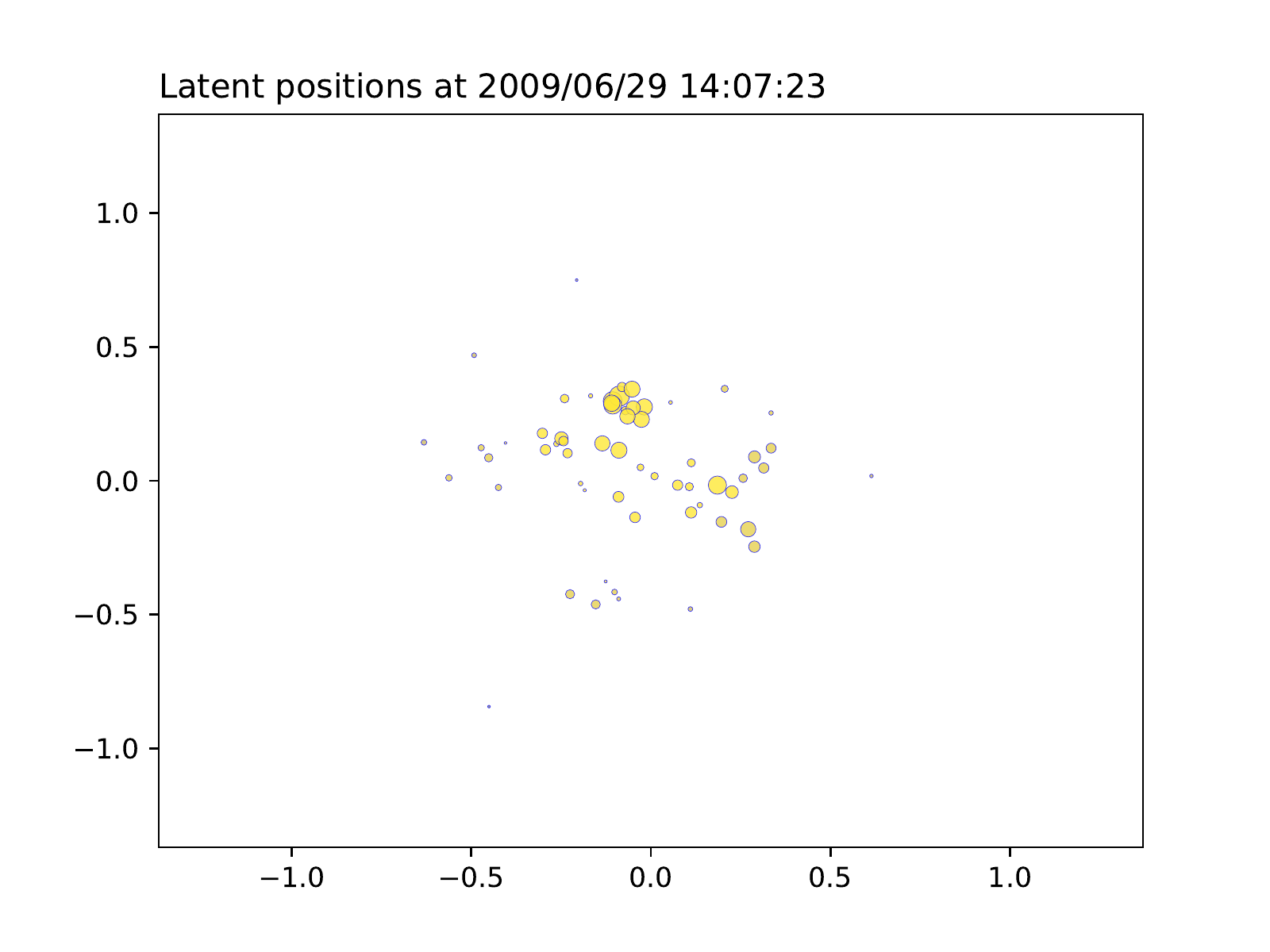}
     \end{subfigure}
        \caption{\textbf{ACM application}: snapshots for the distance model (morning hours).}
        \label{fig:acm_1}
\end{figure}
\begin{figure}
     \centering
     \begin{subfigure}[b]{0.495\textwidth}
         \centering
         \includegraphics[width=\textwidth]{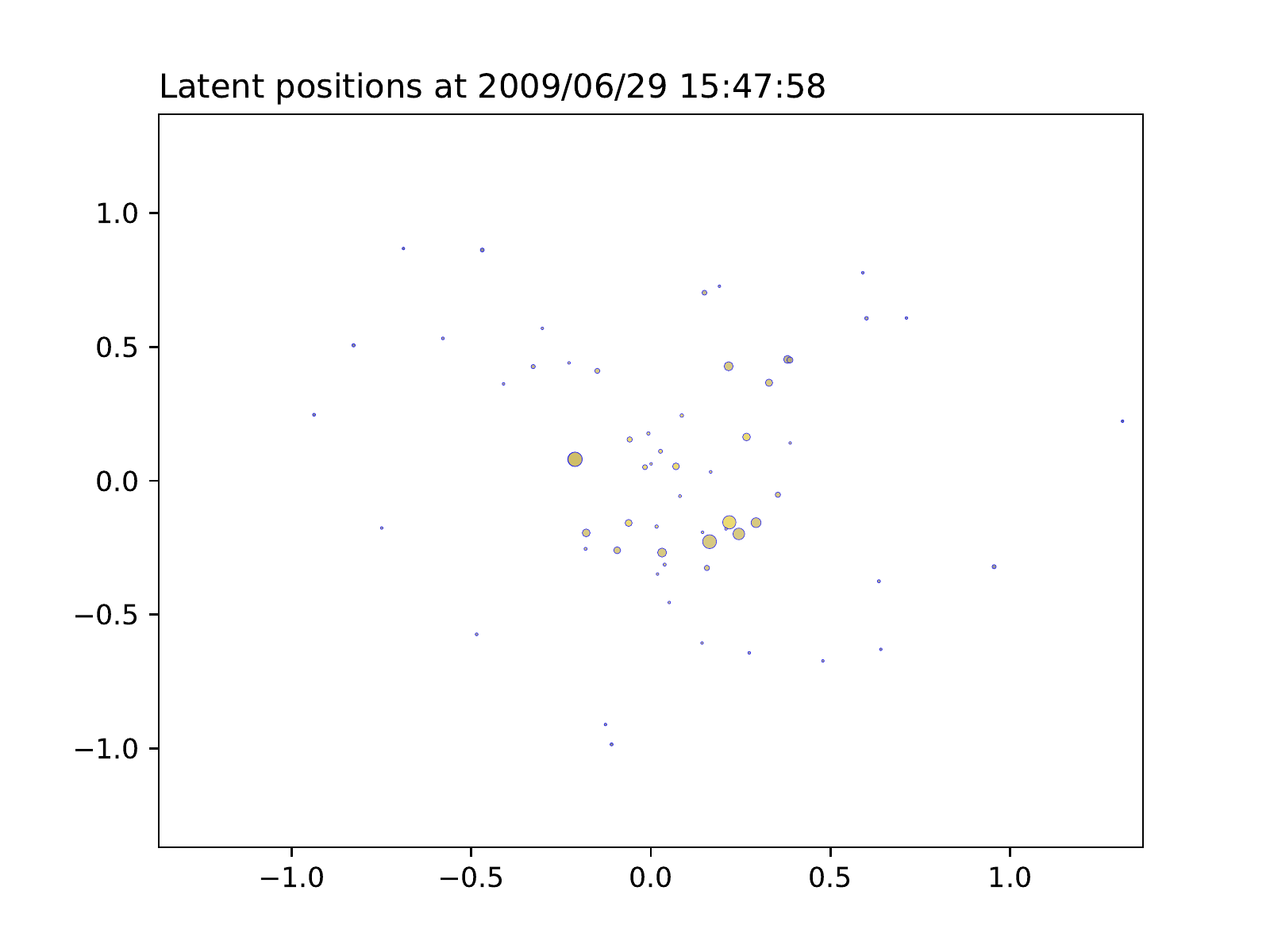}
     \end{subfigure}
     \hfill
     \begin{subfigure}[b]{0.495\textwidth}
         \centering
         \includegraphics[width=\textwidth]{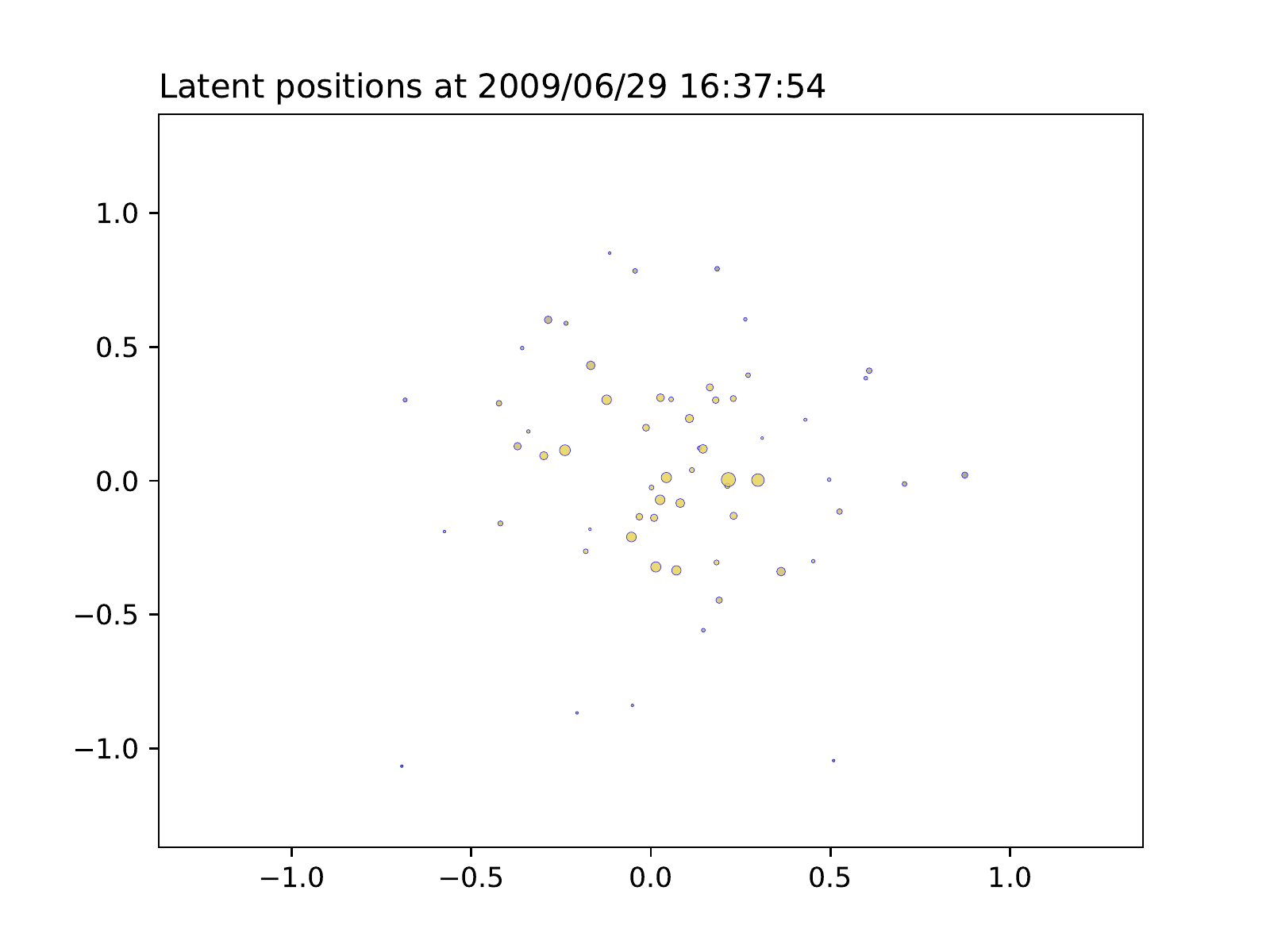}
     \end{subfigure}
     \centering
     \begin{subfigure}[b]{0.495\textwidth}
         \centering
         \includegraphics[width=\textwidth]{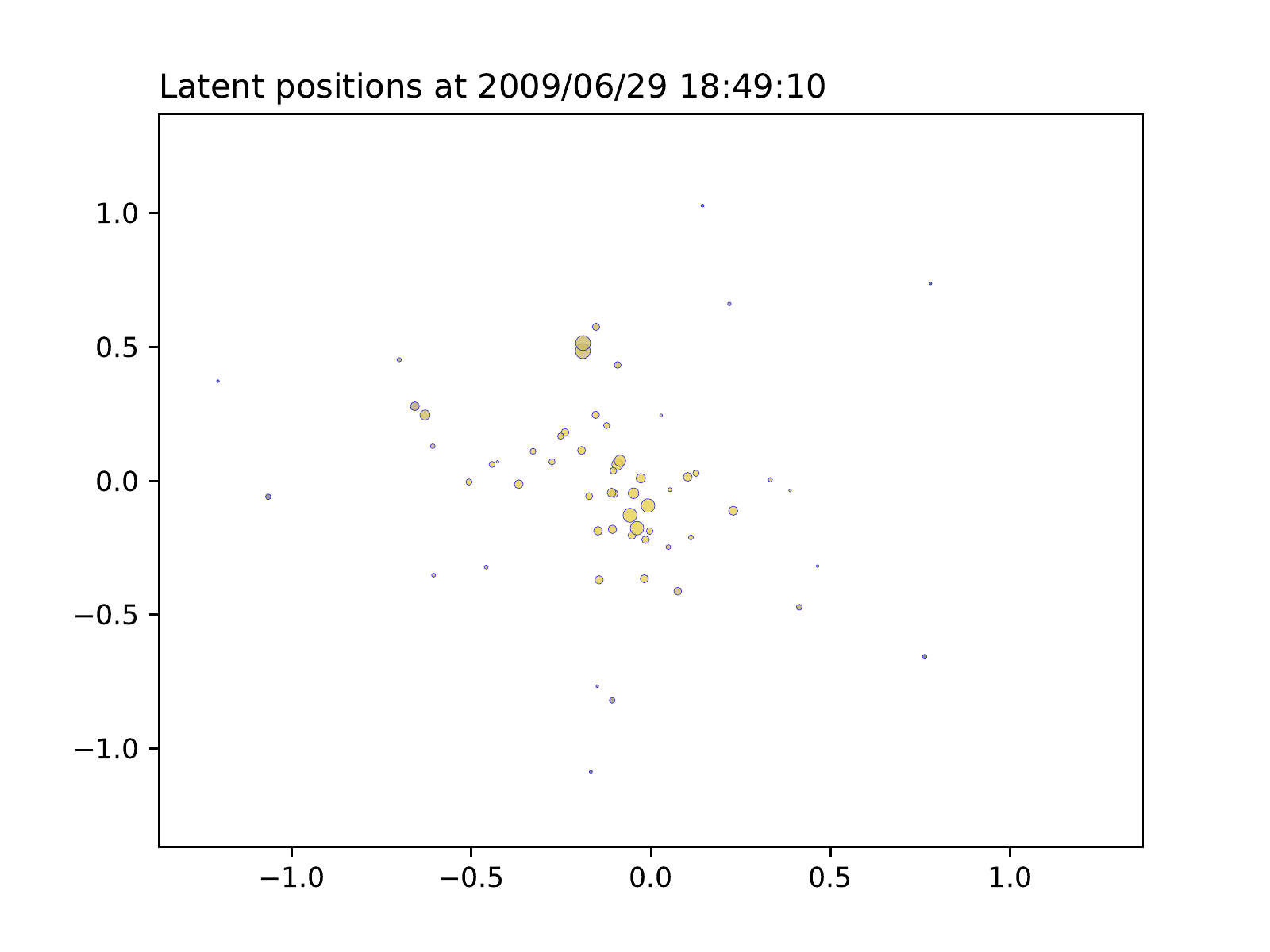}
     \end{subfigure}
     \hfill
     \begin{subfigure}[b]{0.495\textwidth}
         \centering
         \includegraphics[width=\textwidth]{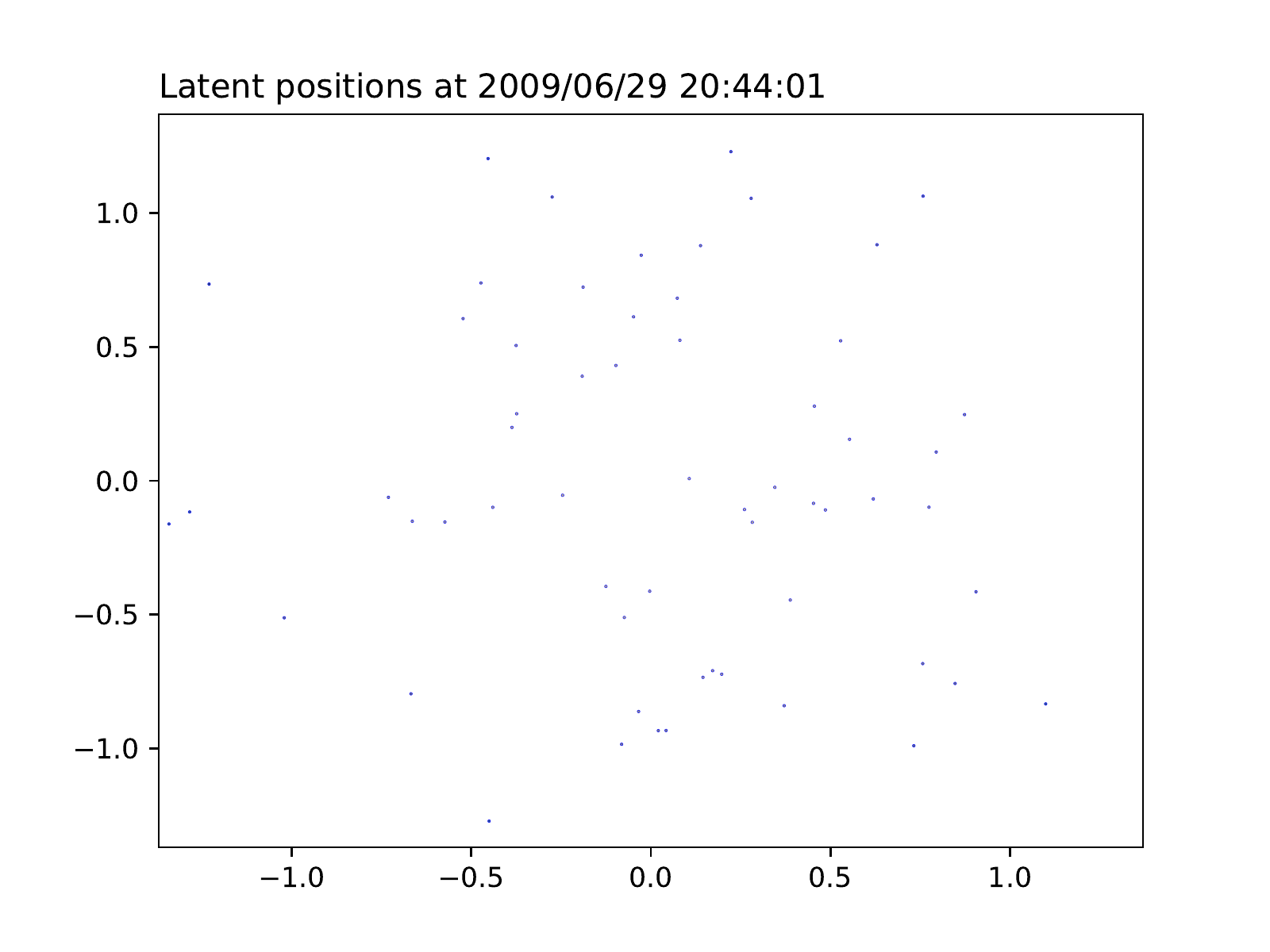}
     \end{subfigure}
        \caption{\textbf{ACM application}: snapshots for the distance model (afternoon hours).}
        \label{fig:acm_2}
\end{figure}
The complete results, shown as a video, can be found in the code repository that accompanies this paper.

We can see that, in the morning, there is a high level of mixing between the attendees. 
The visitors tend to merge and split into different communities that change very frequently and very randomly. 
These communities reach a high level of clusteredness, which signals that the participants of the study are mixing into different groups.
This is perfectly in agreement with the idea that the participants are moving from one location to another, as it usually happens during poster sessions and parallel talk sessions.
The nodes exhibit different types of patterns and behaviours, in that some nodes are central and tend to join many communities, whereas others have lower levels of participation and remain at the outskirts of the space.

In the late morning, we see a clear close gathering around 12 p.m., whereby almost all nodes move towards the centre of the space.
This is emphasized even more at 2 p.m., which corresponds to the lunch break.
It is especially interesting that, even though the space becomes more contracted at this time, we can still clearly see a strong clustering structure.

In the afternoon, we go back to the same patterns as in the morning, whereby the participants mix in different groups and move around the space.
The wine reception is also clearly captured around 6 p.m. where we see again some level of contraction of the space, to signal a large gathering of the participants.

After this event, the overall rate of interactions diminishes sharply, and as a consequence we see the nodes spreading out in the space.

\subsection{Reality mining}
The reality mining dataset \parencite{eagle2006reality} is derived from the Reality Commons project, which was run at the Massachusetts Institute of Technology (MIT) from $14$ September 2004 to $5$ May 2005.
The dataset describes proximity interactions in a group of $96$ students, collected primarily through bluetooth devices.
An overview of this network dataset is also given by \textcite{rastelli2019exact}.

In the context of this paper, the proximity interactions can be reasonably considered as instantaneous interactions, due to the study being $9$ months long.
With our latent space representation, we aim at highlighting the patterns of connections of the students during the study, and any social communities that arise and how these change over time.

Figure \ref{fig:app_reality} shows a few snapshots of our fitted distance \texttt{CLPM}. 
\begin{figure}
     \centering
     \begin{subfigure}[b]{0.495\textwidth}
         \centering
         \includegraphics[width=\textwidth]{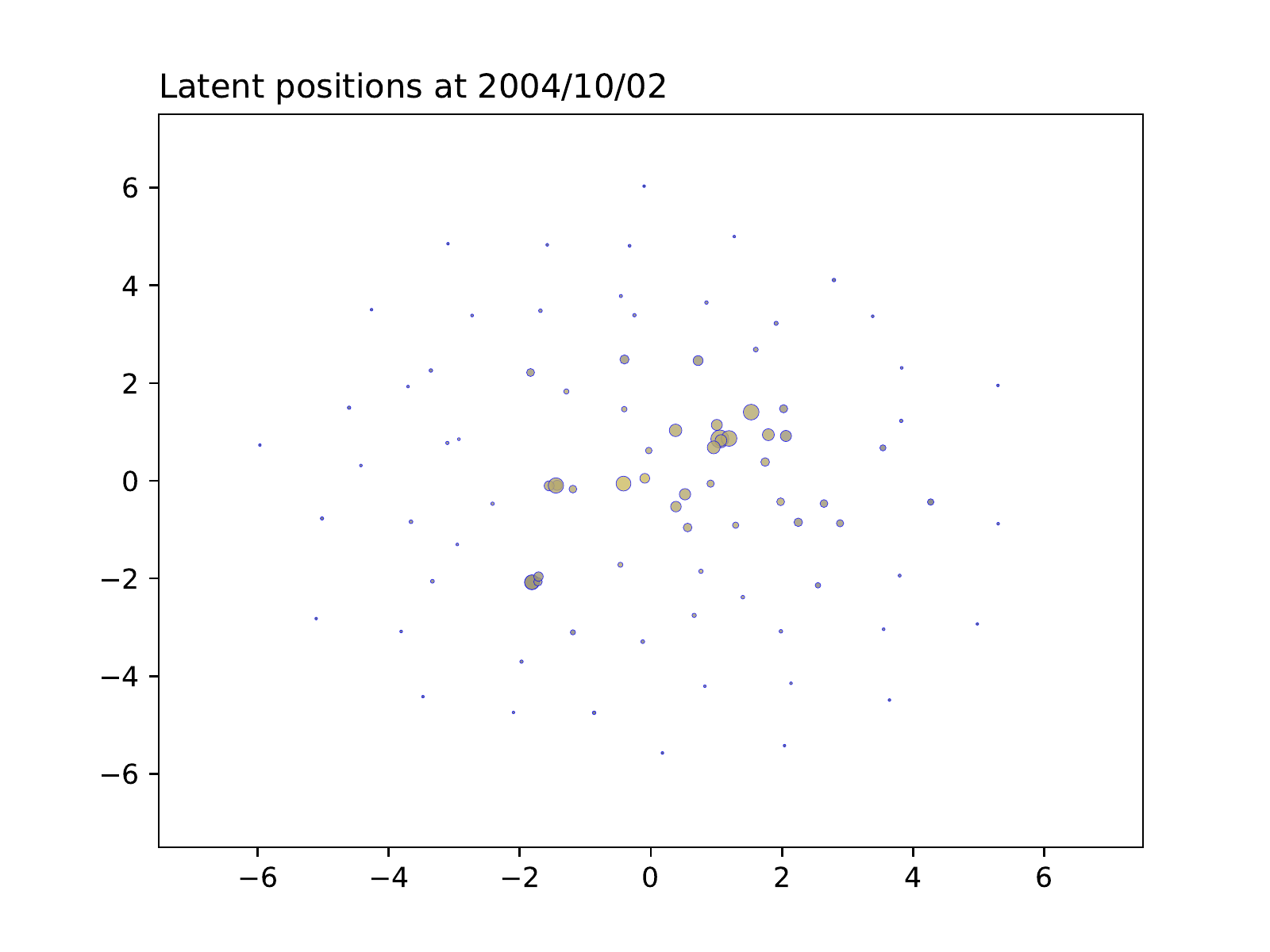}
     \end{subfigure}
     \hfill
     \begin{subfigure}[b]{0.495\textwidth}
         \centering
         \includegraphics[width=\textwidth]{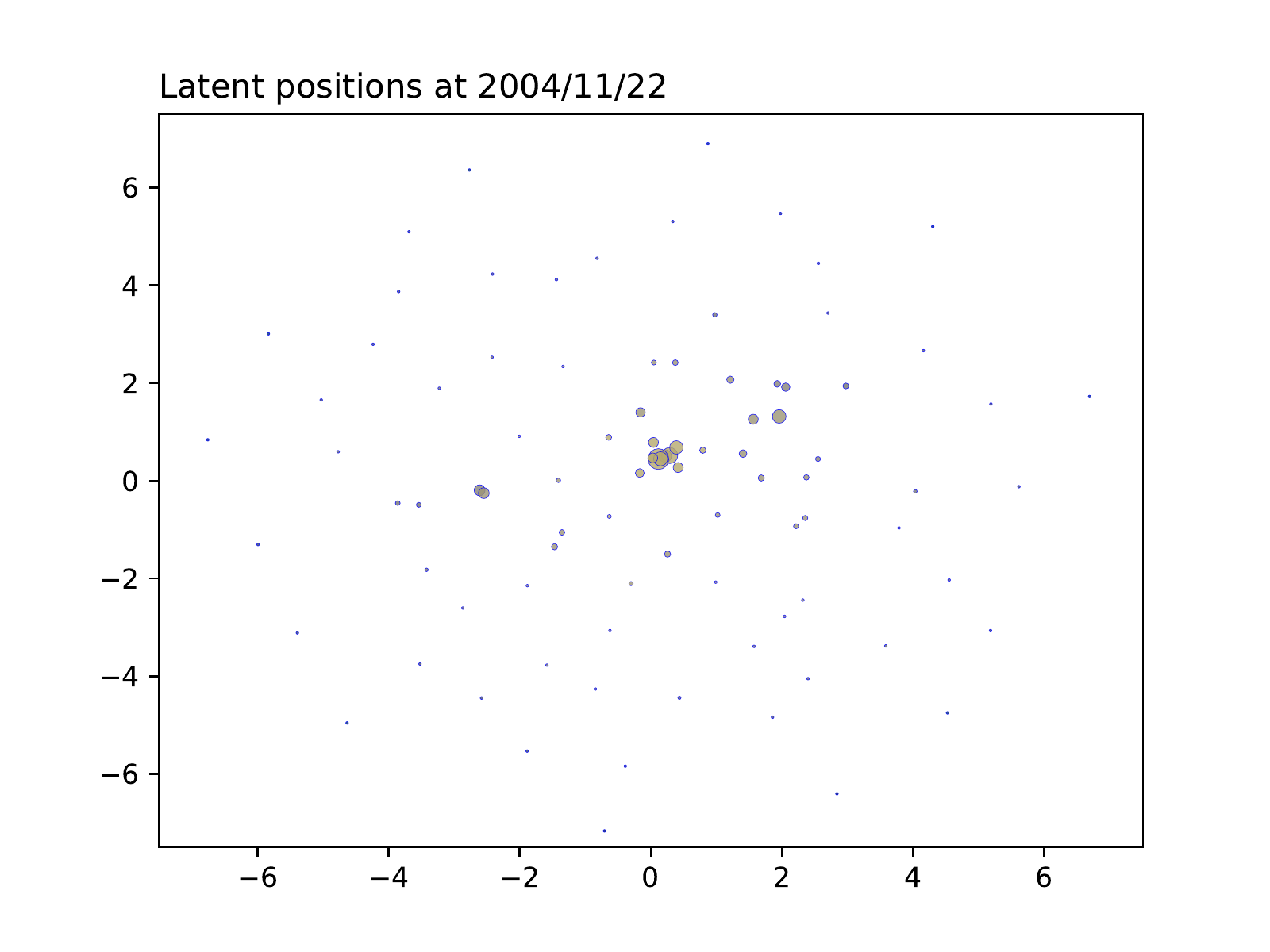}
     \end{subfigure}
     \centering
     \begin{subfigure}[b]{0.495\textwidth}
         \centering
         \includegraphics[width=\textwidth]{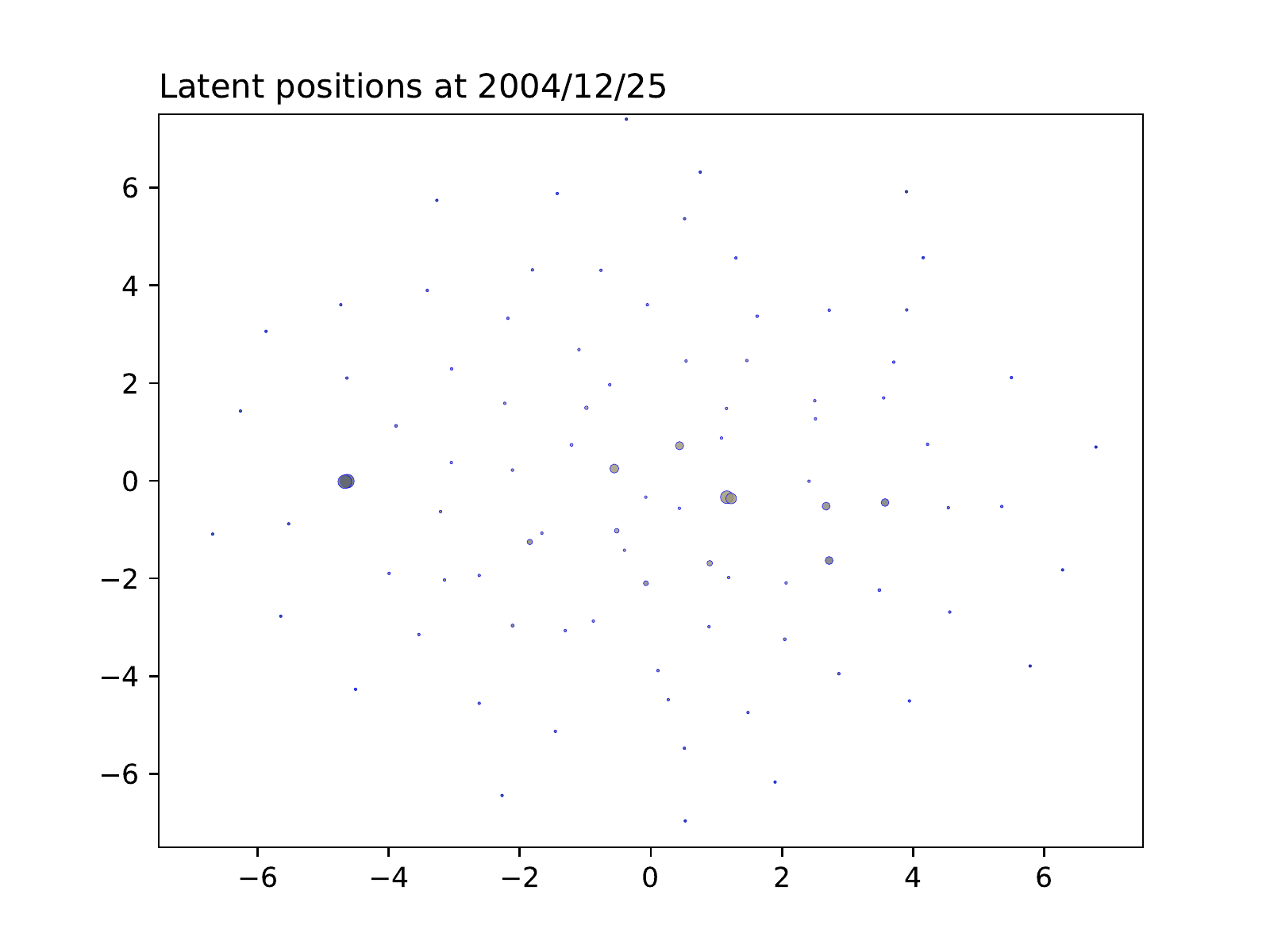}
     \end{subfigure}
     \hfill
     \begin{subfigure}[b]{0.495\textwidth}
         \centering
         \includegraphics[width=\textwidth]{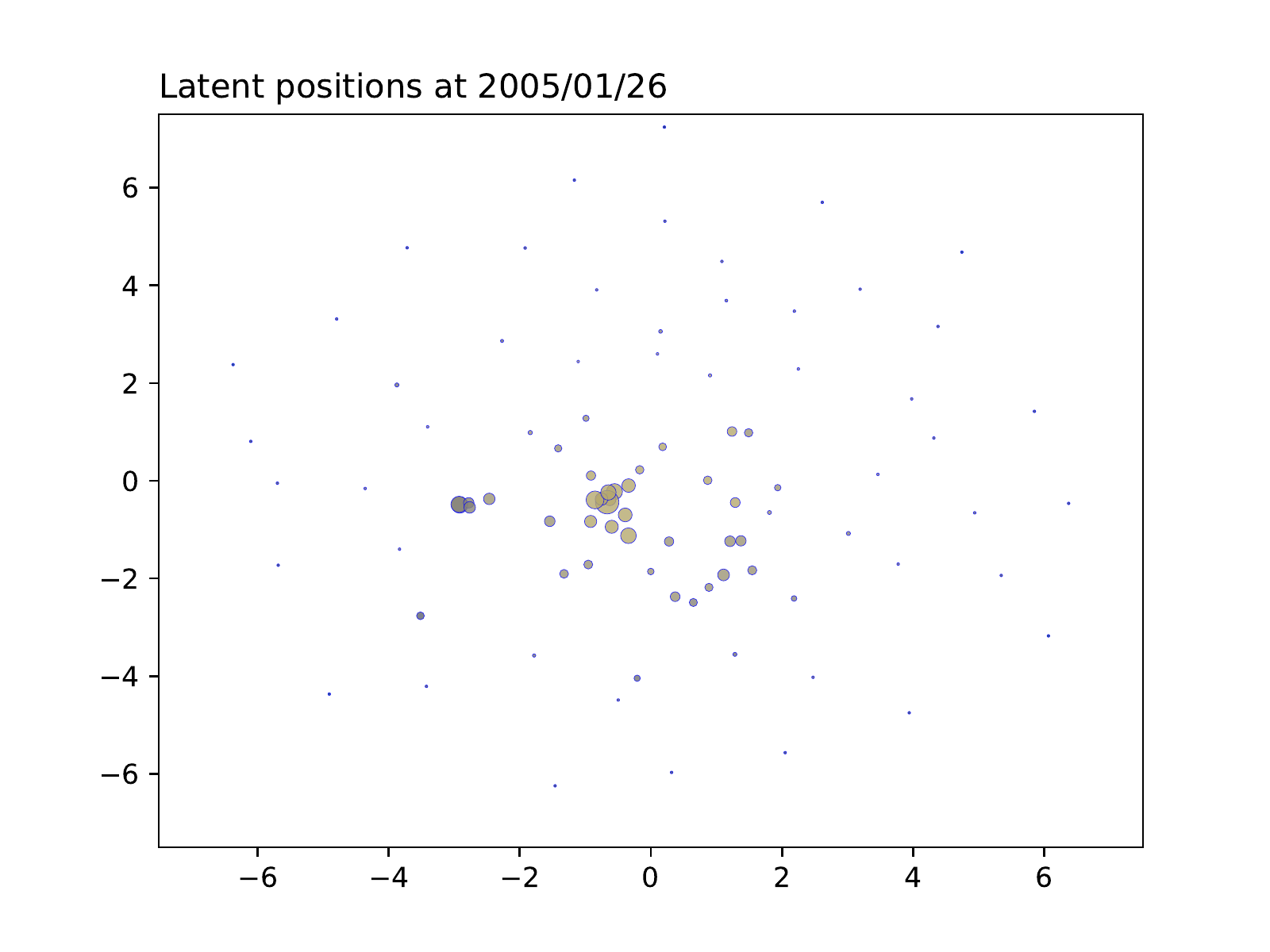}
     \end{subfigure}
     \begin{subfigure}[b]{0.495\textwidth}
         \centering
         \includegraphics[width=\textwidth]{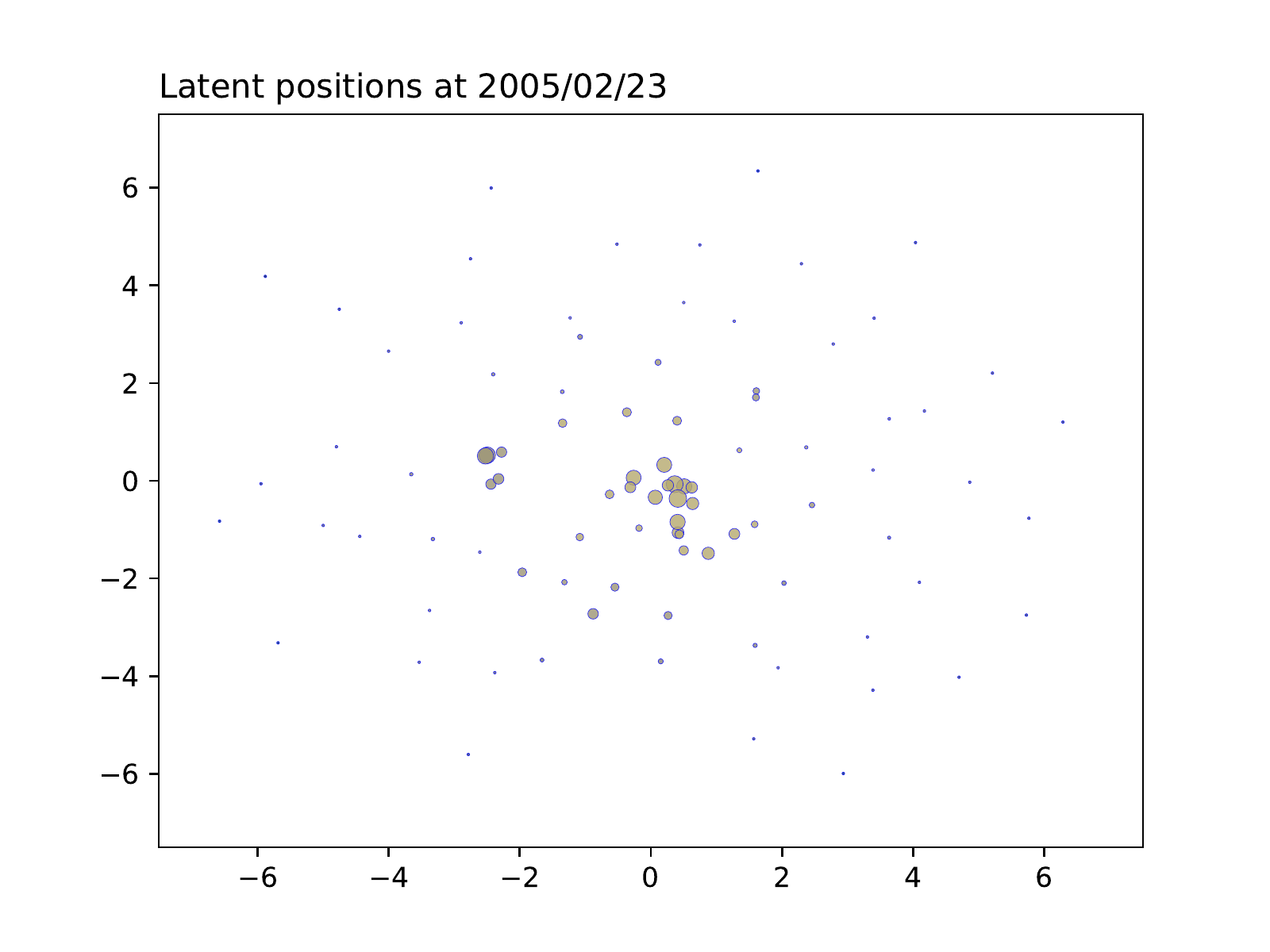}
     \end{subfigure}
     \hfill
     \begin{subfigure}[b]{0.495\textwidth}
         \centering
         \includegraphics[width=\textwidth]{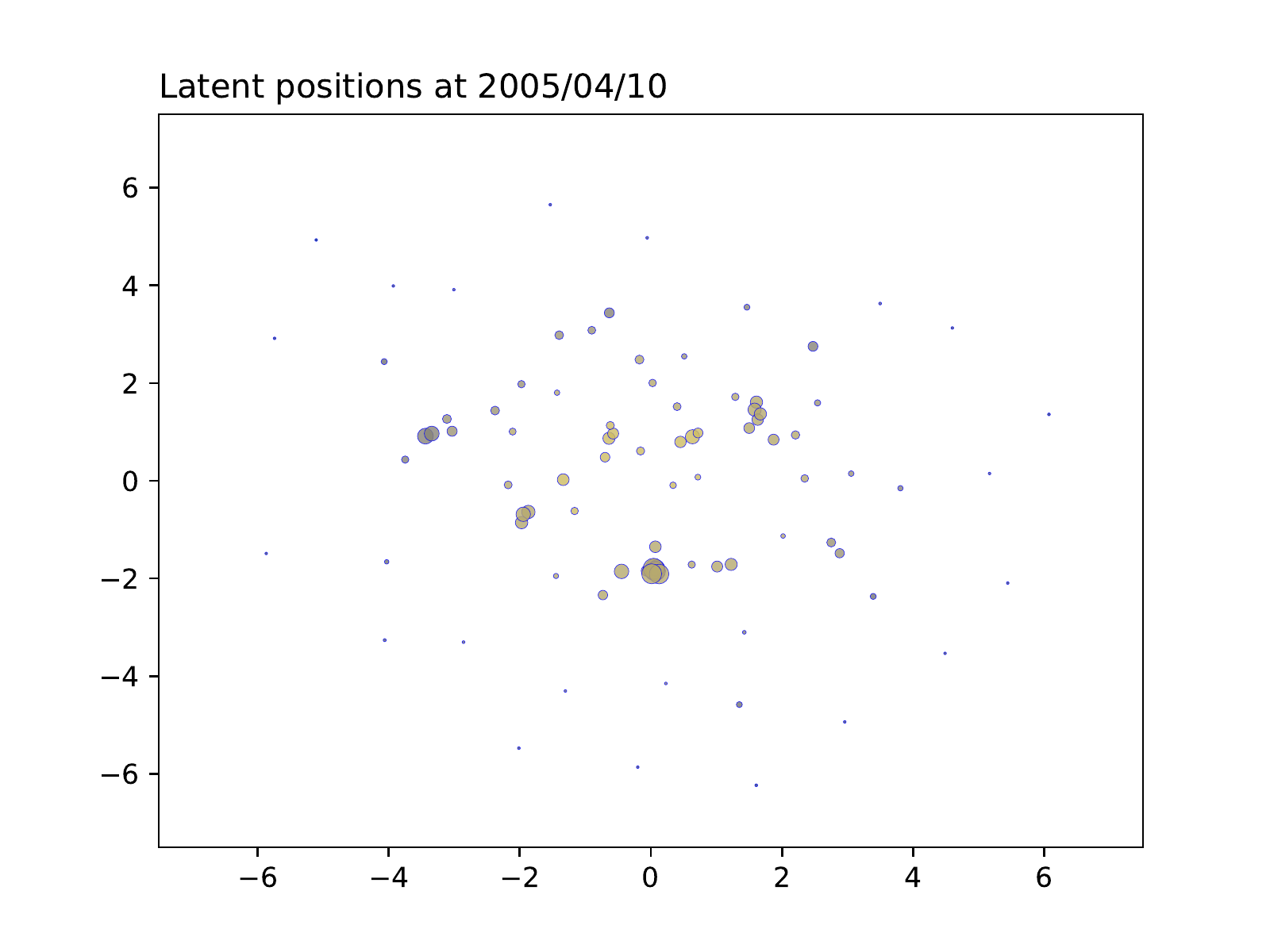}
     \end{subfigure}
        \caption{\textbf{MIT application}: snapshots for the distance model.}
        \label{fig:app_reality}
\end{figure}
The complete results, shown as a video, can be found in the code repository that accompanies this paper.

We observe that, in general, the students are quite separated and few communities arise.
This does not necessarily mean that the nodes do not interact, but it is a sign that there are no subgroups of students with an uncommonly high interaction rate.
Also, it is important to point out that the latent space is very expanded: coupled with an estimated intercept value of $3.9$, this signals that the latent space has a strong effect on modeling the interaction rates and that it can capture well the variability in the data.

Over time, the students tend to mix in different social groups, thus quickly forming and undoing communities.
This could be explained by the interactions that the students have due to college activities or other daily activities.
Near the end of the study, a large cluster appears, signaling a large gathering to which the students participated.
This may correspond to the period before a deadline, as outlined in \textcite{eagle2006reality}.

\subsection{London bikes}
Infrastructure networks provide an excellent example of instantaneous interaction data.
In this section we consider a network of bike hires which is collected and publicly distributed by \textcite{tfl2016cycle}.
We focus on a specific weekend day (Sunday 6 September 2015), and study the patterns of interactions between all bike hire stations in London.
The bike hire stations correspond to the nodes of our network, whereas an instantaneous interaction between two nodes at time $t$ simply means that a bike started a journey from one station towards the other, at that time (we consider undirected connections).

In Figure \ref{fig:app_bikes} we show a collection of snapshots at some critical time points during the day, for the distance model.
\begin{figure}
     \centering
     \begin{subfigure}[b]{0.495\textwidth}
         \centering
         \includegraphics[width=\textwidth]{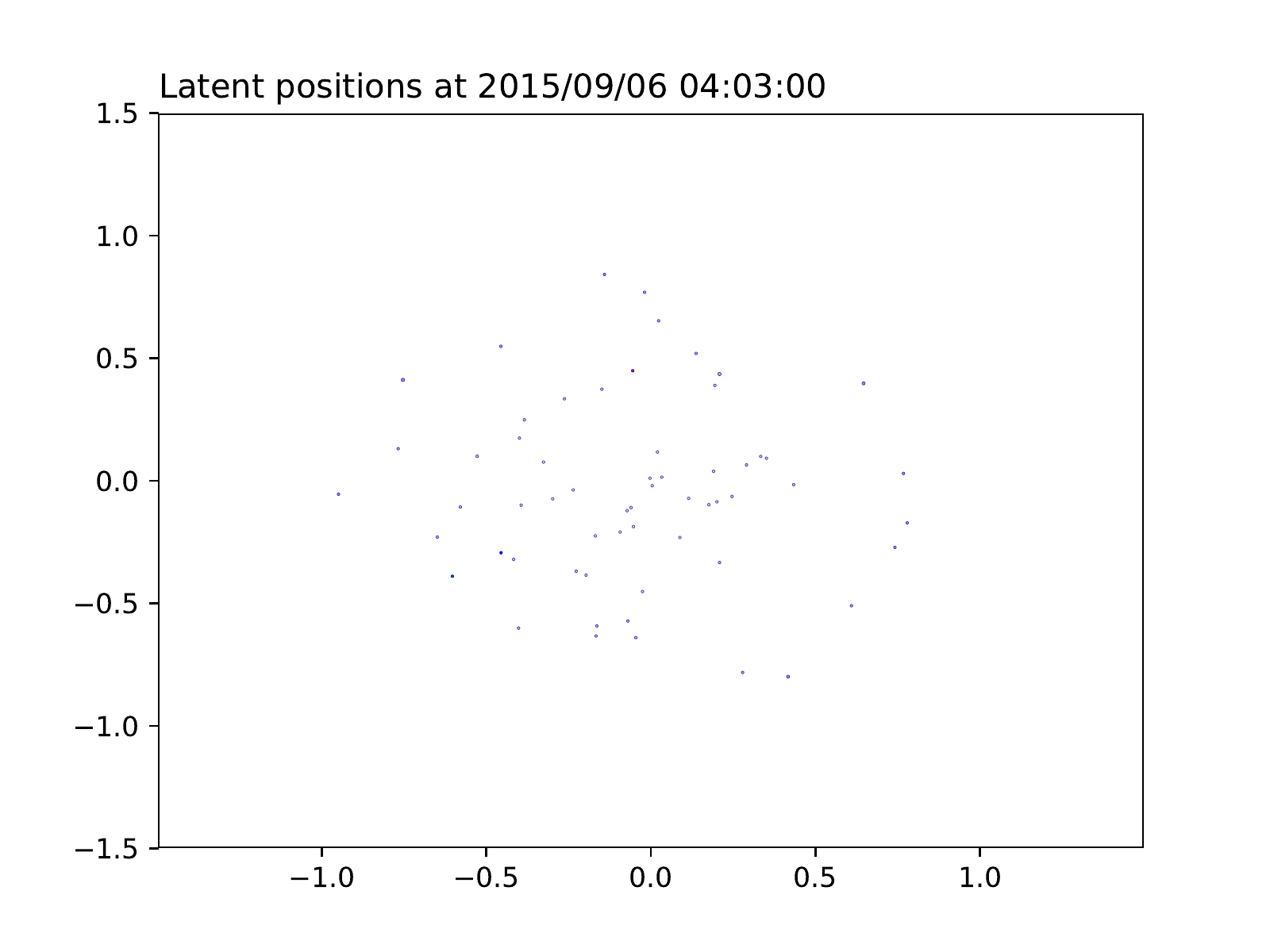}
     \end{subfigure}
     \hfill
     \begin{subfigure}[b]{0.495\textwidth}
         \centering
         \includegraphics[width=\textwidth]{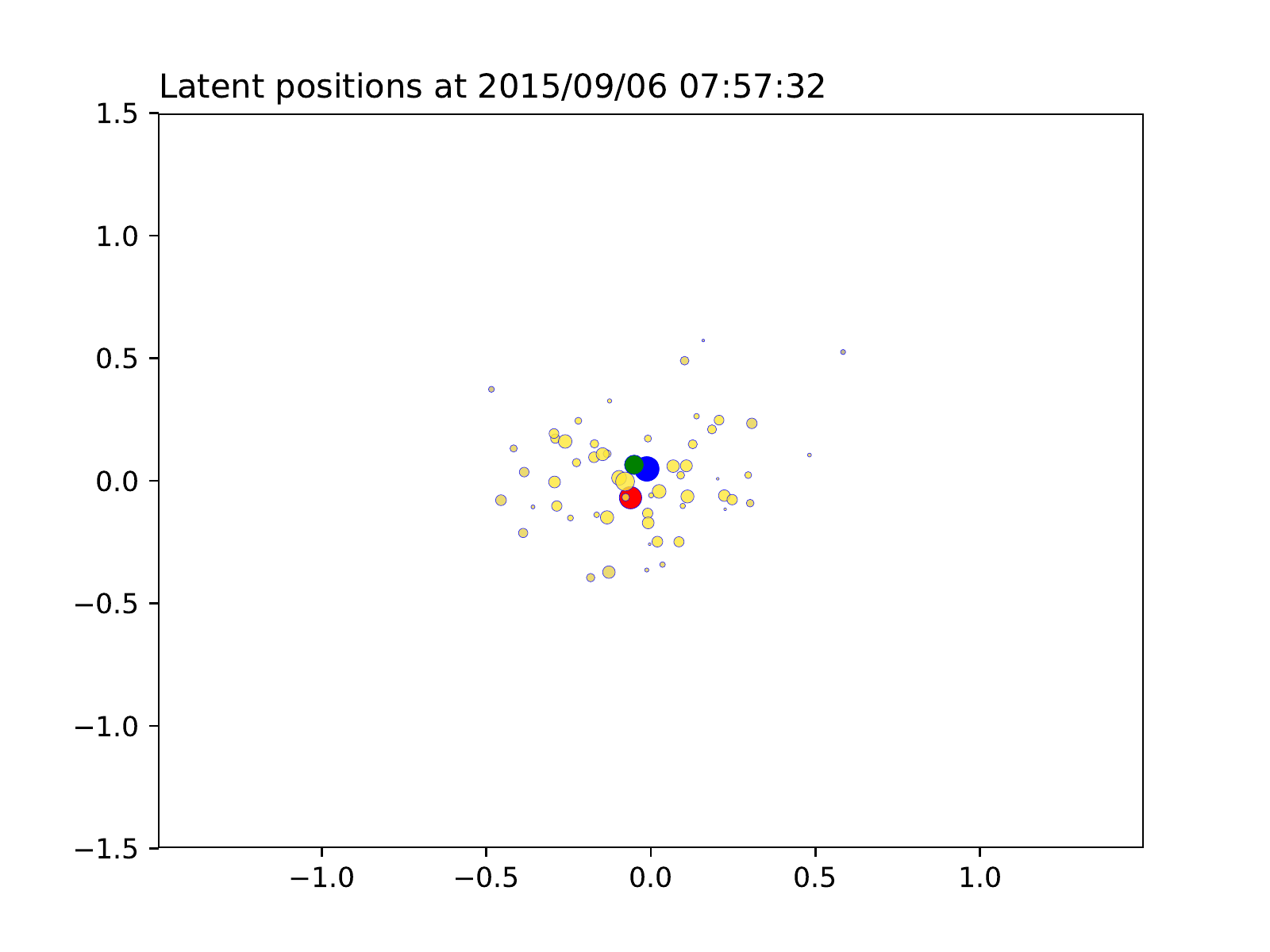}
     \end{subfigure}
     \centering
     \begin{subfigure}[b]{0.495\textwidth}
         \centering
         \includegraphics[width=\textwidth]{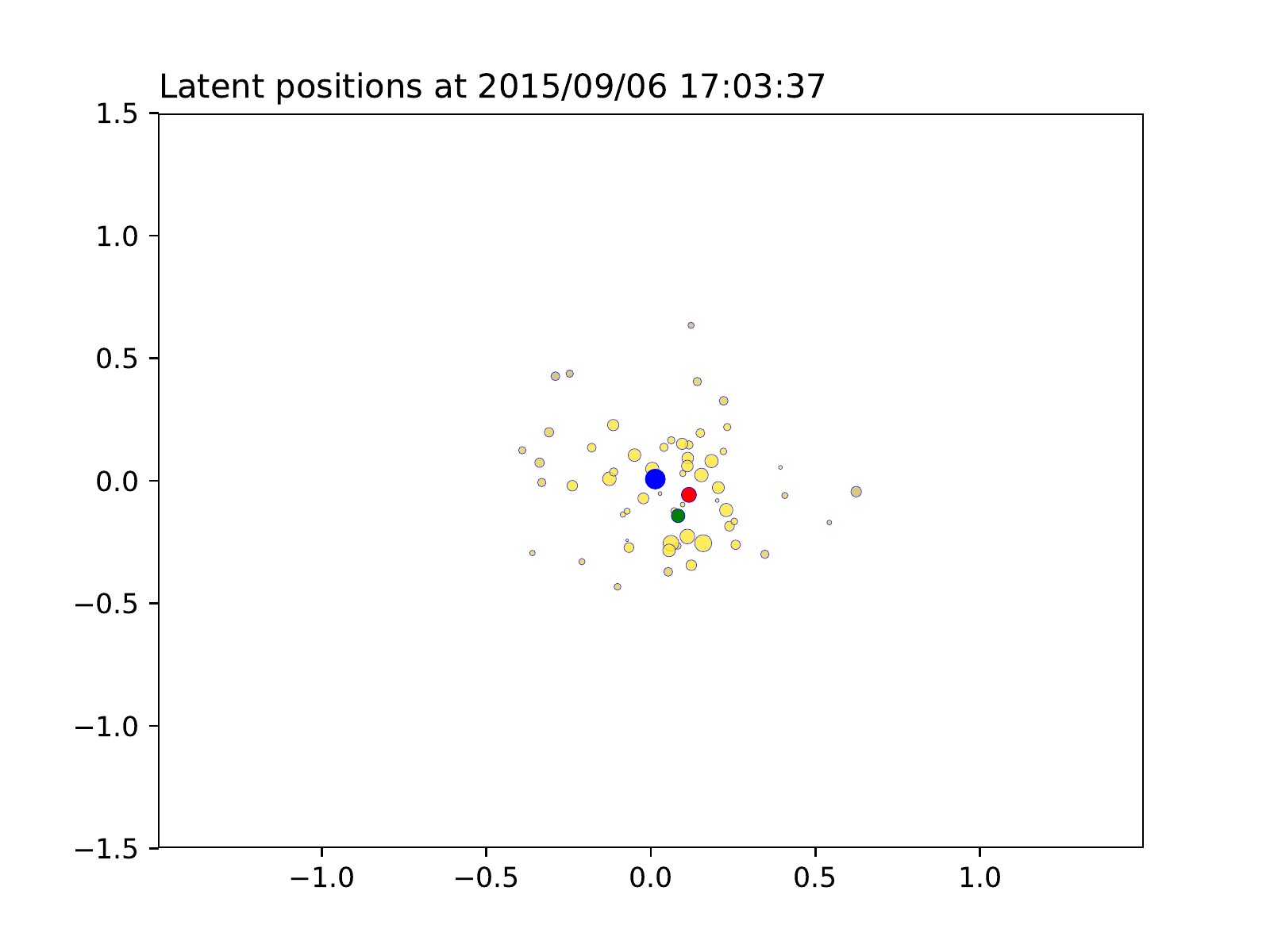}
     \end{subfigure}
     \hfill
     \begin{subfigure}[b]{0.495\textwidth}
         \centering
         \includegraphics[width=\textwidth]{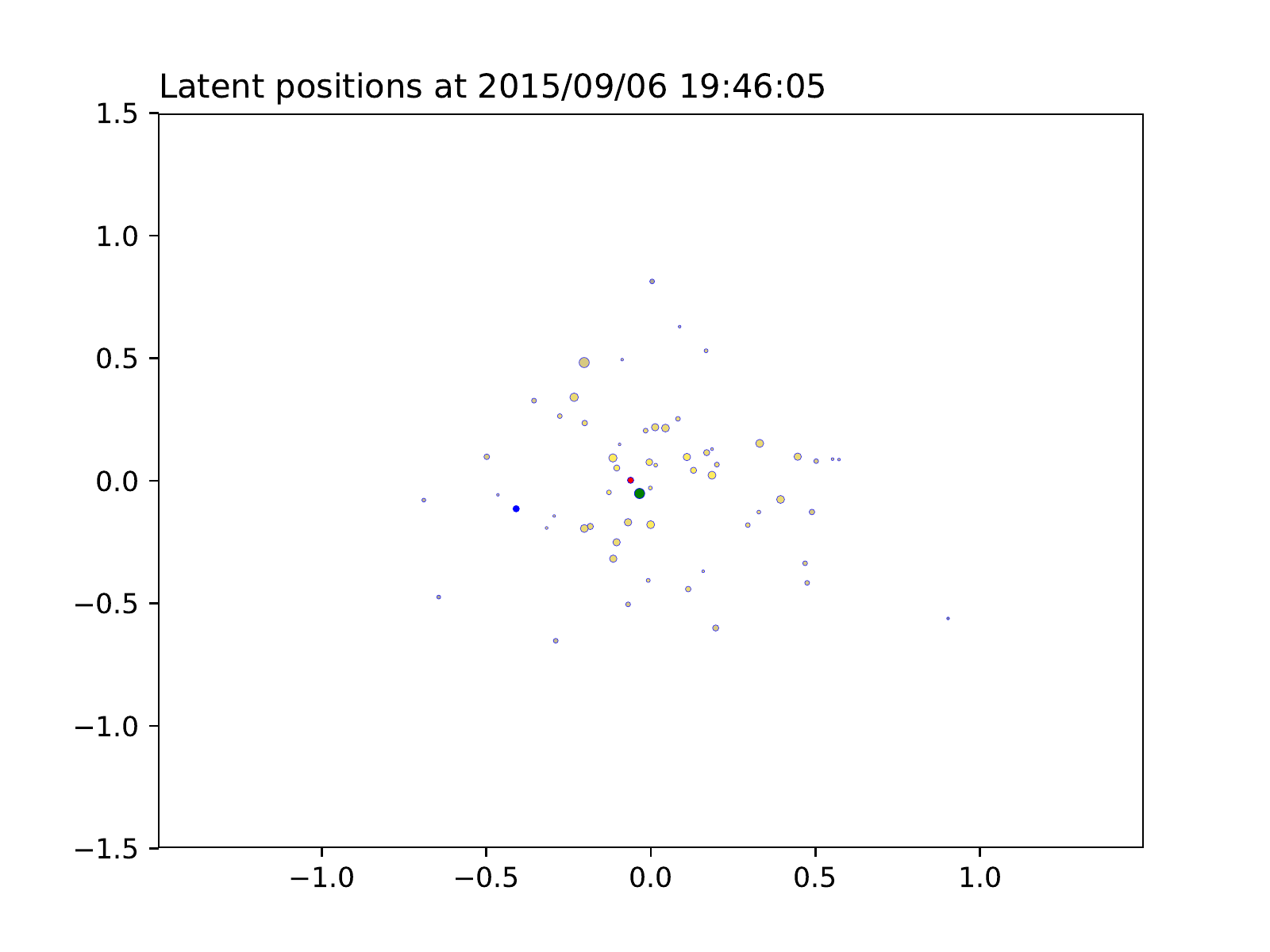}
     \end{subfigure}
        \caption{\textbf{London bikes application}: snapshots for the distance model.}
        \label{fig:app_bikes}
\end{figure}
The complete results, shown as a video, can be found in the code repository that accompanies this paper.

Although there are a total of $818$ stations that are active in this dataset, we provide a visualization for the $60$ most active stations only.
However, we emphasize that the results were obtained using the whole dataset.
In addition, we highlight with a different color the $3$ stations with the highest number of interactions, overall. 
These stations are:
\begin{itemize}
    \item Belgrove Street, King's Cross, situated next to King's Cross Square (shown in red);
    \item Finsbury Circus, Liverpool Street, situated next to Liverpool Street station (shown in blue);
    \item Newgate Street, St. Paul's, situated next to St. Paul's Cathedral (shown in green).
\end{itemize}

The first aspect that we notice is that the latent space expands during inactive times, and it contracts during busy hours.
The contractions and expansions are not homogeneous, rather they highlight the presence of dense and less dense clusters of stations.

The estimate for the intercept parameter is $-5.2$, and the dispersion of the points in the latent space is not particularly large. This signals that the latent space characterization is not having a very strong effect on the rate of interactions, and the model does not capture much variability in the rates of interactions. 
This highlights that these connectivity data follow patterns that cannot be completely explained by the purely geometrical nature of our model.
That is, the connections are determined by a variety of factors that cannot be framed into this latent positions context, and the geographical information on bike hiring accounts for only a part of the problem.

\section{Conclusions}\label{sec:conclusions}
We have introduced a new time-continuous version of the well known and widely used latent position model, as an extension which can model instantaneous interactions between entities.
We have proposed a new methodology which provides good flexibility while also allowing for an efficient inferential framework.
The methodology is implemented in our software \texttt{CLPM} which accompanies this paper and is publicly available.
This provides an essential additional tool for practitioners that are interested in deriving latent space visualizations from observed instantaneous interaction data.

The framework that we propose is highly inspired by the work of \textcite{hoff2002latent}, and by the vast literature that has followed in this direction.
Our work combines some crucial theoretical and statistical aspects of latent position modeling, with a pragmatic approach to inference and visualization of the results.
Crucially, we provide simulation studies and real data applications to demonstrate how our method leads to sensible and accurate results, with low computational demands.

As regards extensions and future work, our research opens up several new directions, to address and potentially change some crucial parts of our procedure.
A fundamental challenge is related to the geometric nature of the latent space. In this paper, and in the literature cited here, affine latent spaces are considered, endowed with the standard dot product, which, in turn, induces the Euclidean distance. However, some important works in the literature of the \emph{static} latent position model consider the latent space to be spherical~\parencite{mccormick2015latent} or hyperbolic~\parencite{krioukov2010hyperbolic,asta2015geometric}. 
As expected, since latent position models are generative models, the geometry of the latent space has crucial consequences on the properties of the simulated network~\parencite{smith2019geometry}. Indeed, recently, \textcite{lubold2020identifying} introduced a method to consistently estimate the manifold type, dimension, and curvature from a class of latent spaces. Addressing these topics in the context of dynamic latent position models is a promising avenue of research that can extend our work.
Another challenging aspect of our methodology regards inference: in this paper, we propose an optimization approach to maximize a penalized likelihood. An interesting alternative would be to consider a different approach that could allow one to also quantify uncertainty around the parameter estimates.
Finally, in terms of modeling, we use piece-wise continuous trajectories due to their flexibility and easy tractability, however, alternatives to this parametrization may also be considered.

\section*{Acknowledgements}
The authors would like to thank Prof. Charles Bouveyron for useful feedback in the early stages of this work, but also for supporting and coordinating a research visit which fostered this project.

\printbibliography

\clearpage
\appendix

\small

\section{Log-likelihood for the projection model}\label{sec:proof_projection}
To make the forthcoming results clearer, we introduce the following notation for the dot products:
\begin{equation}\label{eq:projection_scalar_notation}
S_{ij}^{gh} := \scalar{\bz_i(\eta_g)}{\bz_j(\eta_h)}
\end{equation}
for some nodes $i$ and $j$, and for any $g,h \in \left\{1,\dots,K\right\}$.
Most commonly, $g$ and $h$ will correspond to the labels of two consecutive change points, thus we will use them to identify two breakpoints for the trajectories, and to reconstruct the positions in between.
Also, we denote:
$$S_{ij}(t) := \scalar{\bz_i(t)}{\bz_j(t)}$$
for the dot product at a generic time $t \in [0,T]$.

\begin{proposition}\label{prop:projection_1}
Under the projection model, the log-likelihood is exactly equal to:
\begin{equation}\label{eq:log_likelihood_1}
    \log \mathcal{L}\left(\bZ\right) = \sum_{i,j:\ i < j} \left\{ \sum_{e \in \mathcal{E}_{ij}} \log S_{ij}(\tau_e)
    - \frac{1}{6} \sum_{g=1}^{K-1} \left(\eta_{h} - \eta_{g}\right)\left( 2S_{ij}^{gg} + S_{ij}^{gh} + S_{ij}^{hg} + 2S_{ij}^{hh}\right) \right\}
\end{equation}
where $h = g+1$.
\end{proposition}

\begin{proof}
The only non-trivial part in formula \eqref{eq:log_lik_proj} regards the integral:
$$
\int_{0}^{T} \scalar{\bz_i(s)}{\bz_j(s)}ds
$$
where we have that, for $s \in [\eta_g,\eta_{g+1}]$:
$$
\bz_i(s) = (1-t)\bz_i(\eta_g) + t\bz_i(\eta_{g+1})
$$
and, thanks to Eq. \eqref{eq:trajectories_1}, $t \in [0,1]$ is such that $s = (1-t)\eta_g + t\eta_{g+1}$.
Back to the calculation of the integral, we can decompose this into $K-1$ integrals using the $K$ change points as follows:
$$
\int_{0}^{T} \scalar{\bz_i(s)}{\bz_j(s)}ds
= \sum_{g=1}^{K-1} \int_{\eta_g}^{\eta_{g+1}} \scalar{\bz_i(s)}{\bz_j(s)}ds
$$
and then, in each integral, we apply the transformation $s = (1-t)\eta_g + t\eta_{g+1}$ to obtain:
\begin{equation*}
\begin{split}
\int_{0}^{T} \scalar{\bz_i(s)}{\bz_j(s)}ds
&= \sum_{g=1}^{K-1} \left(\eta_{g+1} - \eta_g\right)\int_{0}^{1} \scalar{\bz_i((1-t)\eta_g+t\eta_{g+1})}{\bz_j((1-t)\eta_g+t\eta_{g+1})}dt \\
&= \sum_{g=1}^{K-1} \left(\eta_{g+1} - \eta_g\right)\int_{0}^{1} \left[(1-t)^2S_{ij}^{gg} + t(1-t)(S_{ij}^{gh}+S_{ij}^{hg}) + t^2S_{ij}^{hh}\right]dt \\
\end{split}
\end{equation*}
where $S_{ij}^{gh}$ is defined in Eq. \eqref{eq:projection_scalar_notation} and $h = g+1$.
Since the $S$ terms do not depend on $t$, the integral can be solved analytically:
\begin{equation*}
\begin{split}
\int_{0}^{T} \scalar{\bz_i(s)}{\bz_j(s)}ds
&= \sum_{g=1}^{K-1} \left(\eta_{g+1} - \eta_g\right)\left[ S_{ij}^{gg}\int_{0}^{1}(1-t)^2dt + (S_{ij}^{gh}+S_{ij}^{hg})\int_{0}^{1}t(1-t)dt + S_{ij}^{hh}\int_{0}^{1}t^2dt \right] \\
&= \sum_{g=1}^{K-1} \left(\eta_{g+1} - \eta_g\right)\left[ \frac{S_{ij}^{gg}}{3} + \frac{S_{ij}^{gh}+S_{ij}^{hg}}{6} + \frac{S_{ij}^{hh}}{3}
\right]
\end{split}
\end{equation*}
which is equivalent to the result of the integral appearing in Eq. \eqref{eq:log_likelihood_1}.
\end{proof}
\section{Log-likelihood for the distance model}\label{sec:closed_form}

We focus on the integral in Eq.~\eqref{eq:log_l} and prove that it can be explicitly solved. 

\begin{proof}
In force of  Eq.~\eqref{eq:trajectories_1}, it reads
\begin{equation}\label{eq:ll_computing}
\begin{split}
&e^{\beta}\left(\sum_{g=1}^{K-1}\int_{0}^{1}  \exp\left\{-\norm{ (1-t)(\bz_i(\eta_g) - \bz_j(\eta_g)) + t(\bz_i(\eta_{g+1})-\bz_{j}(\eta_{g+1}))}^2\right\} (\eta_{g+1} - \eta_{g}) dt  \right)   \\
=&\quad e^{\beta}\left(\sum_{g=1}^{K-1}(\eta_{g+1} - \eta_{g}) \int_{0}^{1} \exp\left\{ - \parallel t (\Delta_i^g - \Delta_j^g) + (\bz_i^g - \bz_j^g)  \parallel_2^2  \right\} dt \right)
\end{split}
\end{equation}
where the variable change  $t = \frac{s-\eta_g}{\eta_{g+1}-\eta_g}$ was performed and the following notations were adopted to simplify the exposition
\begin{equation*}
\begin{split}
    \Delta_i^g :&= \bz_i(\eta_{g+1}) - \bz_i(\eta_{l}) \\      \Delta_j^g :&= \bz_j(\eta_{g+1}) - \bz_j(\eta_{l}) \\
    \bz_i^g :&= \bz_i(\eta_g) \\
    \bz_j^g :&= \bz_j(\eta_g)
\end{split}    
\end{equation*}
By denoting $g(t):=- \parallel t (\Delta_i^g - \Delta_j^g) + (\bz_i^g - \bz_j^g)  \parallel_2^2$ the exponent inside the integral, we can ``complete the square'' as follows
\begin{equation*}
\begin{split}
g(t) &=-\norm{\Delta_i^g - \Delta_j^g}^2\left(t^2 + 2t\scalar{\frac{\Delta_i^g - \Delta_j^g}{\norm{\Delta_i^g - \Delta_j^g}}}{\frac{\bz_i^g - \bz_j^g}{\norm{\Delta_i^g - \Delta_j^g}}}  + \frac{\norm{\bz_i^g - \bz_j^g}^2}{\norm{\Delta_i^g - \Delta_j^g}^2} \right) \\
&=-\norm{\Delta_i^g - \Delta_j^g}^2\left(t - \scalar{\frac{\Delta_i^g - \Delta_j^g}{\norm{\Delta_i^g - \Delta_j^g}}}{\frac{\bz_j^g - \bz_i^g}{\norm{\Delta_i^g - \Delta_j^g}}}\right)^2 \\
&- \left(\norm{\bz_i^g - \bz_j^g}^2 -  \left(\scalar{\frac{\Delta_i^g - \Delta_j^g}{\norm{\Delta_i^g - \Delta_j^g}}}{\bz_j^g - \bz_i^g} \right)^2 \right) \\
&= -\frac{1}{2(\sigma^2)_{ij}^g}(t - \mu_{ij})^2 - \left(\norm{\bz_i^g - \bz_j^g}^2 - (\norm{\Delta_i^g - \Delta_j^g}\mu_{ij}^g)^2 \right)
\end{split}
\end{equation*}
where 
\begin{equation}
\begin{split}
    \mu_{ij}^g :&= \scalar{\frac{\Delta_i^g - \Delta_j^g}{\norm{\Delta_i^g - \Delta_j^g}}}{\frac{\bz_j^g - \bz_i^g}{\norm{\Delta_i^g - \Delta_j^g}}}, \\
    (\sigma^2)^g_{ij} :&=\frac{1}{2\norm{\Delta_i^g - \Delta_j^g}^2}.
\end{split}
\end{equation}
By plugging all this into Eq.~\ref{eq:ll_computing} it follows that
\small
{\color{black}
\begin{equation}\label{eq:ll_close}
\begin{split}
\int_0^T   & e^{\beta - \parallel \bz_{i}(t)  - \bz_{j}(t) \parallel_2^2} dt = \\
&=\sqrt{2\pi} e^{\beta}\left\{\sum_{l=1}^{K-1}e^{-(\norm{\bz_i^g - \bz_j^g}^2 - (\norm{\Delta_i^g - \Delta_j^g}\mu_{ij}^g)^2 )} \sigma_{ij}^g (\eta_{g+1} - \eta_g) \left[\Phi\left(\frac{1-\mu_{ij}^g}{\sigma_{ij}^g}\right)- \Phi\left(\frac{0-\mu_{ij}^g}{\sigma_{ij}^g}\right)\right] \right\}
\end{split}
\end{equation}}
\end{proof}
\end{document}